\documentclass[12pt,onecolumn]{IEEEtran}

\usepackage{graphicx,wrapfig}
\usepackage{flushend}
\usepackage{amsthm}
\usepackage{amsmath}
\usepackage{tikz}
\usetikzlibrary{arrows}
\usetikzlibrary{datavisualization.formats.functions}

\usepackage{pgfplots}
\usepackage[normalem]{ulem}
\usepackage[ruled,linesnumbered]{algorithm2e}% For algorithms
\usepackage[english]{babel}
\usepackage{multirow}
\usepackage{enumerate}
\usepackage[font=footnotesize, labelfont=bf, justification=centering]{caption}
\usepackage[font=footnotesize, labelfont=bf, justification=centering]{subcaption}
\usepackage{enumitem}
\usepackage{pifont}
\usepackage{cite}
\usepackage{amsfonts}
\usepackage{amssymb}
\usepackage{fancyhdr}

 \newtheorem{thm}{Theorem}[section]
 
 \newtheorem{lem}{Lemma}[section]
 
 \newtheorem{defn}{Definition}[section]
 \newtheorem{assum}{Assumption}[section]

 \newtheorem{rem}{Remark}[section]

\newcommand{\argmax}[1]{\underset{#1}{\operatorname{arg}\,\operatorname{max}}\;}

 \newcommand{\blue}{\color{blue}}

\usepackage{nomencl}
\makenomenclature

\begin{document}

\title{{\small{\blue ACM Transactions on Sensor Networks}\vspace{-6pt}}\\ POSE.R: Prediction-based Opportunistic Sensing for Resilient and Efficient Sensor Networks}

\author{ James~Z.~Hare,~Junnan~Song,~Shalabh~Gupta,~Thomas~A.~Wettergren
	\thanks{J. Z Hare is with the Signal and Information Processing Division of the CCDC Army Research Laboratory, Adelphi, MD 20783 USA (e-mail: james.z.hare.civ@mail.mil). J. Song and S. Gupta are with the Laboratory of Intelligent Networks and Knowledge-Perception Systems (LINKS), University of Connecticut, Storrs, CT 06268 USA (email: \{junnan.song, shalabh.gupta\}@uconn.edu). T. A. Wettergren is with the Naval Undersea Warfare Center, Newport, RI 02841 (email: t.a.wettergren@ieee.org).}% <-this % stops a space
	}
\maketitle

\begin{abstract}
The paper presents a distributed algorithm, called \emph{\underline{P}rediction-based \underline{O}pportunistic \underline{Se}nsing for \underline{R}esilient and Efficient Sensor Networks} (POSE.R), where the sensor nodes utilize predictions of the target's positions to probabilistically control their multi-modal operating states to track the target. There are two desired features of the algorithm: energy-efficiency and resilience. If the target is traveling through a high node density area, then an optimal sensor selection approach is employed that maximizes a joint cost function of remaining energy and geometric diversity around the target's position. This provides energy-efficiency and increases the network lifetime while preventing redundant nodes from tracking the target. On the other hand, if the target is traveling through a low node density area or in a coverage gap (e.g., formed by node failures or non-uniform node deployment), then a potential game is played amongst the surrounding nodes to optimally expand their sensing ranges via minimizing energy consumption and maximizing target coverage. This provides resilience, that is the self-healing capability to track the target in the presence of low node densities and coverage gaps. The algorithm is comparatively evaluated against existing approaches through Monte Carlo simulations which demonstrate its superiority in terms of tracking performance, network-resilience and network-lifetime.
\end{abstract}


\maketitle

\thispagestyle{empty}

\nomenclature{$s_i$}{Sensor node $i$}
\nomenclature{$\tau_\ell$}{Target $\ell$}
\nomenclature{$\mathbf{u}^x$}{Position of either sensor node x or target x}
\nomenclature{$R_{LPS}$}{Sensing range of Low Power Sensing Device}
\nomenclature{$R_{HPS}\in\{R_1,...,R_L\}$}{Adjustable sensing ranges of High Power Sensing Device}
\nomenclature{$R_{c}$}{Communication Range}
\nomenclature{$\mathcal{N}^{s_i}$}{Set of neighbors of node $s_i$}
\nomenclature{$E_{\Delta T}^{s_i}(k)$}{Energy consumption of node $s_i$ from time $k-1$ to $k$}
\nomenclature{$e^{s_i}_j$}{Energy consumption of device $j$ on node $s_i$}
\nomenclature{$D^{\tau_\ell}$}{Coverage degree of target $\tau_\ell$}
\nomenclature{$D_b^{\tau_\ell}, D_e^{\tau_\ell}$}{Base and extended coverage degree of target $\tau_\ell$. Note that using the tilde notation indicates a predicted coverage degree.}
\nomenclature{$\hat{\mathbf{x}}^{s_i}(k|k), \hat{\mathbf{\Sigma}}^{s_i}(k|k), \hat{\mathbf{W}}^{s_i}(k|k)$}{Target state, covariance, and Kalman filter gain estimates at time $k$.}
\nomenclature{$P_{LPS}^{\tau_\ell,s_i}, P_{HPS}^{\tau_\ell,s_i}$}{Low power and high power sensing devices probability of detection}
\nomenclature{$\mathbf{z}(k)$}{Set of measurements collected by the HPS device at time $k$}
\nomenclature{$\mathbf{f}(\mathbf{x}(k),k)$, $\mathbf{h}(\mathbf{x}(k),k)$}{Target motion model and measurement model.}
\nomenclature{$N_{sel}$}{Number of sensor nodes selected}
\nomenclature{$\Theta=\{\theta_1,\theta_2,\theta_3\}$}{Set of states of the PFSA.}
\nomenclature{$p_{x,y}^{s_i}(k)$}{The probability of transitioning from state $x$ to state $y$ at time $k$.}
\nomenclature{$\hat{P}_{HPS}^{s_i}$}{Probability of successful target detection.}
\nomenclature{$\mathcal{N}_{HPS}^{s_i}$}{Set of neighbors in the HPS state.}
\nomenclature{$\mathcal{S}^*(k+1)$}{Optimal set of nodes selected to track the target.}
\nomenclature{$\mathcal{R}^*(k+1)$}{Optimal set of high power sensing ranges for the nodes in $\mathcal{S}^*(k+1)$.}
\nomenclature{$\mathcal{N}_{RC}$}{Set of nodes that receive target state information.}
\nomenclature{$\hat{\mathbf{I}}^{s_i}(k)$}{Information ensemble received by node $s_i$ at time $k$.}
\nomenclature{$\mathcal{N}_T^{s_i}$}{Set of trustworthy neighbors.}
\nomenclature{$\hat{\mathbf{I}}^{s_i}(k)$}{Set of trustworthy information ensembles.}
\nomenclature{$\Omega_{can}^{R_s}(k+1)$}{Candidate region surrounding the target.}
\nomenclature{$\mathcal{S}_{can}^{R_s}(k+1)$}{Set of candidate nodes to track the target}
\nomenclature{$E_{can}^{R_s}(k)$}{Set of remaining energies of the candidate nodes.}
\nomenclature{$\mathcal{S}'(k+1)$}{Set of nodes that are players in the potential game. }
\nomenclature{$a=(a_i,a_{-i})$}{Joint actions of the potential game. }
\nomenclature{$E_{rem}^{s_i}(k)$}{Energy remaining of node $s_i$}
\nomenclature{$\Omega_u$}{Target partition region. }
\nomenclature{$v_{g,h}$}{A cell $g,h$ of the partition region. }
\nomenclature{$w_{g,h}$}{The worth assigned to cell $v_{g,h}$ of the partition region.}
\nomenclature{$E_c(a_j)$}{Predicted energy consumption of sensor $s_j$.}
\nomenclature{$J_{g,h}$}{Number of nodes covering cell $v_{g,h}$}
\nomenclature{$N'_{sel}$}{Maximum number of players for the potential game.}
\nomenclature{$B_{g,h}(J_{g,h}(a))$}{Coverage Function.}
\nomenclature{$\Delta T$}{Amount of time between discrete time steps $k-1$ and $k$}
\nomenclature{$\Delta b_1, \Delta b_2$}{Slopes of the coverage function.}
\nomenclature{$\hat{\mathbf{z}}^{s_i,c}(k+1|k)$, $\hat{\mathbf{\Sigma}}_\mathbf{z}^{s_i,c}(k+1|k)$}{Predicted location of target $c$ and its associated covariance matrix. }
\nomenclature{$p_{sleep}$}{The probability of remaining in the Sleep state.}
\nomenclature{$\rho$}{Network Density}
\nomenclature{$p_{rand}$}{Probability of transitioning to the HPS state in the LPS-HPS scheduling method.}
\nomenclature{$\lambda$}{Number of targets located in a tube $\Omega_{\gamma}$ within the network.}

\nomenclature{}{}

\section{Introduction}
A critical challenge of \emph{Distributed Sensor Networks} (DSNs), that are used for various intelligence, surveillance, and reconnaissance (ISR) operations~\cite{JSRGD2012, MGRW2011}, is to maintain performance of their intended task (e.g., target tracking~\cite{HGS14}) in the presence of sensor node failures. Sensor nodes are prone to failures~\cite{YSALS2014} due to component degradations, hardware malfunctions, malicious attacks, battery depletions, or environmental uncertainties \cite{MA2010}, causing changes in the network topology. If multiple co-located sensors fail, a sector of the network may be uncovered, causing missed detections when a target travels through such coverage gap. This results in poor network performance, information delays, and mission failures. Additionally, the sensor nodes may be non-uniformly distributed, resulting in high and low density regions. Therefore, the development of an opportunistic self-healing network that provides resilience to the effect of low node densities and coverage gaps is essential to maintain network performance.

To account for node failures, two proactive approaches have been proposed in the literature: (i) redundant node deployment and (ii) intelligent network control for energy-efficiency and life-extension. The former approach deploys redundant sensor nodes throughout the \emph{Region Of Interest} (ROI) to ensure that every point is observed by $\kappa>1$ nodes \cite{AD2010, HT2005}. This creates a fault-tolerant network that allows for $\kappa-1$ nodes to fail before a coverage gap is formed; however, it is costly. Moreover, this approach does not provide resilience if multiple spatially co-located nodes fail, for example, an attack in a battlefield scenario.

The second proactive approach incorporates an intelligent network control strategy that minimizes node failures caused by energy depletion due to inefficient use. One such strategy, known as \emph{Opportunistic Sensing} \cite{HGW2017}, consists of selecting and activating sensor nodes only in the local regions around targets' predicted positions, while the nodes away from the targets are not selected and deactivated to conserve energy. This method maximizes the network lifetime while maintaining high tracking accuracy and low missed detection rates, via forming dynamic clusters of activated sensors around moving targets.

Another control strategy~\cite{CWL2006} aims to optimize the nodes' ranges and activation times to minimize energy consumption and missed detections. This approach assumes that the nodes' sensing ranges can vary based on the amount of power supplied to their sensing devices. However, these approaches assume that the targets are fixed and known \emph{a priori} and do not consider tracking mobile targets. Furthermore, these strategies only address energy-efficient control and do not address the problem of resilience to sensor failures that have already occurred.

\begin{figure*}[t!]
	\centering
	\includegraphics[width=0.95\textwidth]{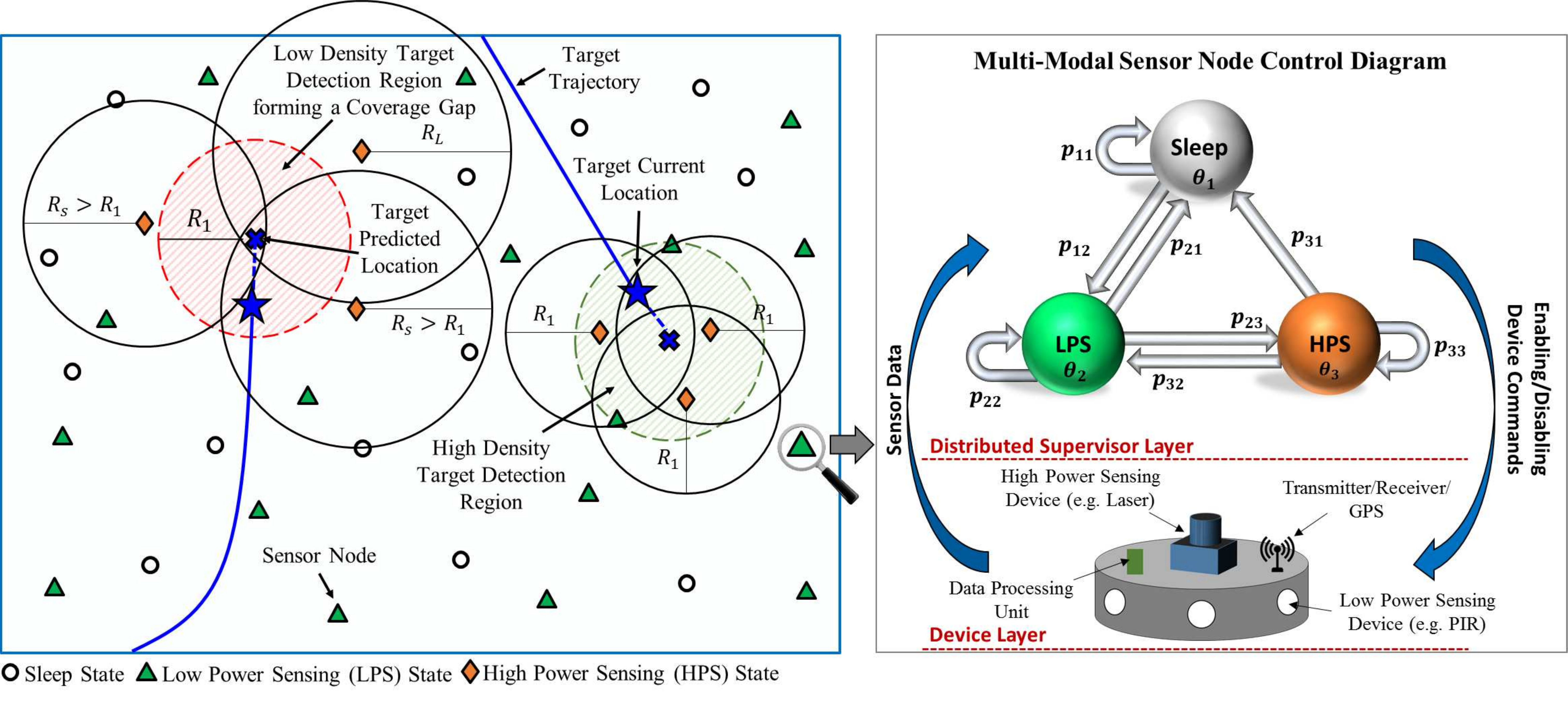} \vspace{-6pt}
	\caption{Illustration of the POSE.R algorithm with distributed PFSA-based supervisory control} \label{fig:main}
	\vspace{-15pt}
\end{figure*}

In this regard, this paper proposes a distributed supervisory control algorithm, called \emph{Prediction-based Opportunistic Sensing for Resilient and Efficient Sensor Networks} (POSE.R). The objective of POSE.R is two-fold: i)
provide resilience to the effects of low node densities and coverage gaps to maintain tracking performance and ii) provide energy-efficient target tracking~\cite{hare2015decentralized} in areas of high node densities for network lifetime extension. This algorithm extends the POSE~\cite{HGW2017} and POSE.3C~\cite{HGW2019} algorithms (Appendix~\ref{app:net_comp}) by incorporating resilient and efficient tracking in the presence low node densities, sensor failures, and non-uniform node deployment. This is achieved by including an adaptive distributed node selection approach that dynamically selects the optimal nodes and their sensing ranges to track the mobile targets.

The sensor selection approach adapts to the density of the sensor nodes around the targets' predicted positions, as seen in Fig.~\ref{fig:main}. For high density regions ($\ge N_{sel}$ nodes around the target), a novel sensor node selection method is developed, called \textit{Energy-based Geometric Dilution of Precision} (EGDOP), to select and activate geometrically diverse nodes with high remaining energies to track the target with their minimum sensing range, $R_{1}$. This method preserves energy, prevents nodes from dying by utilizing high energy sensors, and minimizes tracking error.

On the other hand, for low density regions or coverage gaps, a Game-theoretic sensor node selection method is developed, using \emph{Potential Games}, to select the optimal nodes as well as their optimal sensing ranges between $[R_{1}, R_{L}]$, to accommodate for the insufficient number of nodes or a coverage gap and maintain the tracking performance while minimizing energy consumption. This method provides the following advantages: (1) non-cooperative games allow for scalable distributed computing in a DSN, (2) Potential games ensure that an equilibrium exists, and (3) maximizing the local objective function guarantees that the global objective is maximized. Thus, POSE.R algorithm provides resilience, that enables opportunistic self-healing by adjusting the sensing ranges of nodes surrounding the targets' predicted positions to maintain tracking accuracy.

The underlying distributed network controller is built using a \emph{Probabilistic Finite State Automaton} (PFSA), which is embedded on each sensor node to control its heterogeneous (i.e., multi-modal) operating states by probabilistically enabling/disabling its devices at each time step. The states of the PFSA include: 1) \emph{Sleep}, 2) \emph{Low Power Sensing} (LPS), and 3) \emph{High Power Sensing} (HPS). The \emph{Sleep} state preserves maximum energy by disabling all devices on the node. The LPS state utilizes the LPS devices for the purpose of target detection while conserving energy. The HPS state utilizes the HPS devices for precise target measurements and state estimation. The range of the HPS devices are varied from $[R_{1}, R_{L}]$ based on the proposed distributed adaptive node selection method to ensure target coverage while minimizing redundancy and energy consumption. The transceiver is enabled in both the LPS and HPS states to allow for information sharing and collaboration with neighbors.

The state transition probabilities of the PFSA are dynamically updated based on the adaptive sensor selection algorithm and the information observed with the node's on-board sensing suite. The probabilities are designed to transition a node to the HPS state only when it is selected for tracking a target that is predicted to travel within it's coverage area. On the other hand, a node transitions between low power consuming states, i.e., LPS or \emph{Sleep}, to conserve energy when not selected. This is illustrated in Fig.~\ref{fig:main}, where $N_{sel}=3$ nodes are selected to be in the HPS state around each target's predicted position to ensure high tracking accuracy, while the remaining nodes conserve energy to provide significant energy savings. As seen in Fig.~\ref{fig:main}, in the presence of a coverage gap, the POSE.R network is able to adapt the sensing ranges of surrounding nodes to fill the gap and maintain tracking performance, thus providing resilience.

The main contribution of this paper is the development of a distributed supervisory control algorithm, that facilitates resilient and efficient target tracking, using a distributed node selection approach that adapts to the network density around the targets' predicted positions, such that: \vspace{-3pt}
\begin{enumerate}[label=\Roman*]
\item [a)] for high density regions, the EGDOP node selection method provides energy-efficiency, and
\item [b)] for low density regions, the Game-theoretic node/range selection method provides resilience.
\end{enumerate}
The remainder of this paper is organized as follows. Section~\ref{sec:rw} discusses the current literature of fault-tolerant and adjustable range WSNs. Section~\ref{sec:pf} presents the problem and the objectives. Section~\ref{sec:poser} discusses the POSE.R algorithm while Section~\ref{sec:dsc} presents the distributed collaboration method for sensor and range selection. Section~\ref{sec:results} presents the validation results and the conclusions are stated in Section~\ref{sec:con}. Appendices~\ref{app:EGDOP}-\ref{app:net_comp} are provided to supplement the main paper.

\section{Related Work} \label{sec:rw}
This section presents a literature review of fault-tolerant control strategies and the maximum network lifetime problem in sensor networks and their limitations that are addressed in this paper.

\vspace{-6pt}
\subsection{Fault-tolerant WSN}
Fault-tolerance requires that the sensor nodes:  (1) detect node faults and (2) react to mitigate the faults. Failure detection is typically achieved using active and passive monitoring approaches. Active monitoring approaches utilize a centralized or cluster-based network topology \cite{YMM2007} and consist of requesting constant updates (e.g., heartbeat signals) from nodes. Passive monitoring methods can be implemented in centralized, cluster-based, or distributed network topologies, by observing the traffic already present in the network to infer the nodes' health \cite{PH2007}. These methods assume that healthy measurements are spatially correlated in a local neighborhood, while faulty measurements are uncorrelated~\cite{LC2008}. Therefore, for fault detection, a node can compare its data  with the median of their neighborhood measurements \cite{DLTCC2007}, perform Bayesian model comparison \cite{KI2004}, or perform hypothesis testing \cite{LDH2006, J2009}. For a detailed discussion of fault detection methods, see~\cite{MS2017}. Although these approaches can identify faulty measurements, they cannot directly identify coverage gaps caused by node failures.

Once the faulty nodes are detected, it is critical to employ a recovery mechanism to mitigate their effects. This consists of two approaches: proactive and reactive. Most proactive approaches deploy redundant sensor nodes to ensure  $\kappa$-coverage \cite{YMSXS2019, GKJ2016, AD2010, HT2005, AKJ2005} or $\kappa$-connectivity \cite{HCLS2010}, where $\kappa$ is the number of sensors that can cover a target or is the number of communication paths. Other approaches examine various deployment topologies that provide fault tolerant properties \cite{BBM2009}. However, they require significantly more nodes to be deployed, and if multiple co-located nodes fail, then they fail to ensure coverage.

Some approaches study optimal sensor placement~\cite{TGJP2017, KMGG2008, OSU2018} for fault-tolerance. These approaches utilize submodular functions to identify the best network configuration, while accounting for (possible) node failures. This has also been extended to non-submodular functions in~\cite{TJP2018}. However, these approaches do not consider the problem of adjusting the sensing radius of the active nodes for resilient target tracking.

Reactive approaches aim to recover coverage or connectivity that was lost due to the failed nodes. For stationary sensor networks, single sensor failure recovery methods have been proposed. These include, storing redundant data for data recovery \cite{CM2005}, re-routing connectivity paths around the failed node or adjusting packet size sent to the failed node \cite{li2019bim2rt,LC2012}, and re-configuring clusters to recover child nodes from a failed cluster head \cite{CK2017, HL2018, GY2003, AY2007, BXM2011, AKJ2015}. Fault recovery approaches for multiple co-located node failures have not been proposed for stationary sensor networks. The closest approach by Younis et. al. \cite{YSALS2014} requires identifying and placing optimal relay nodes to ensure connectivity around partitioned segments of the network. However, these methods typically apply to communication networks and do not address the problem of healing the coverage gaps in target tracking networks.

\vspace{-6pt}
\subsection{Maximum Network Lifetime Problem}
The second problem addressed in this paper is the \emph{Maximum Network Lifetime with Adjustable Range} (MNLAR) problem for (static) target coverage \cite{CWL2006}.  The objective of the MNLAR problem is two fold: (1) perform energy-efficient scheduling by activating and deactivating nodes periodically, and (2) select the active nodes and adjust their sensing ranges to ensure that every target is covered. This problem has been formulated as an optimization problem in the form of Integer Programming \cite{CWL2006, LWCL2009}, Linear Programming \cite{DVZLP2006, CDR2012, MSR2014}, Voronoi Graphs \cite{ZDG2009, BCLPS2012}, and improved Memetic optimization \cite{AB2017}. This problem is NP-complete \cite{CWL2006}; thus, for real-time performance, many heuristic solutions have been proposed.

The centralized heuristics aim to identify the family of cover sets that achieve coverage of all (static) targets. The objective is to optimize: the nodes' sensing ranges within a cover set and a sequence of cover sets that maximizes the network lifetime. This problem was solved using the \emph{Adjustable Range Set Covers} (AR-SC) algorithm \cite{CWL2006} which develops a Linear Programming heuristic to approximate the Integer Programming solution. The \emph{Sensor Network Lifetime Problem} (SNLP) \cite{DVZLP2006} utilized the Garg-Konemann algorithm to approximate the optimal linear programming solution within a small factor. The Column Generation algorithm by Cerulli et. al. \cite{CDR2012} used a greedy heuristic which was adjusted by Mohamadi et. al. \cite{MSRM2015} using a learning automata-based algorithm. Additionally, the MNLAR problem was extended to include directional (e.g., camera) sensor networks in \cite{MSR2014, RSM2017, CRSV2018, AMM2019}.

For distributed heuristics, many approaches follow a greedy-based scheme. AR-SC \cite{CWL2006} has each sensor node operate in rounds. During each round, a node computes its wait time, which is a representation of how much energy and contribution the sensor adds to the group. Once a node's wait time is up, it selects the minimum sensing range that can cover all the uncovered targets and transmits this information to its neighbors. This approach was  extended in the \emph{Adjustable Sensing Range Connected Sensor Cover} (ASR-CSC) algorithm \cite{LWCL2009} to allow for connectivity.

The \emph{Variable Radii Connected Sensor Cover} (VRCSC) algorithm \cite{ZDG2009} uses a Voronoi partition based algorithm that partitions the region into a Voronoi Graph and selects the sensing and communication ranges of each node to ensure $\kappa$-coverage and $\kappa$-connectivity. The node waits to make a decision based on its sleeping benefit and then determines its minimum sensing range to occupy the cell that contains a target.
A similar approach was presented in the \emph{Sensor Activation and Radius Adaptation} (SARA) algorithm \cite{BCLPS2012} using Voronoi-Laguerre diagrams.

Dhawan et. al. \cite{DAP2010} proposed two distributed heuristics, \emph{Adjustable Range Load Balancing Protocol} (ALBP) and \emph{Adjustable Range Deterministic Energy Efficient Protocol} (ADEEPS), where ALBP balances the energy depletion, while ADEEPS utilizes load balancing and reliability.

\vspace{-6pt}
\subsection{Research Gaps}
As stated earlier, all of the MNLAR proposed solutions rely on the assumption that the targets in the network are static and that their locations are known \emph{a priori} by all of the nodes. However, in target tracking applications, the targets are dynamic and travel through the network or may also randomly appear and disappear within the network. Therefore, this paper aims to solve the MNLAR problem for dynamic targets whose locations are unknown \emph{a priori}.

Additionally, the proposed MNLAR problems do not consider sensor failures. Fault-tolerance is only proactive, where the network deploys redundant sensor nodes. Thus, if a single node or multiple co-located nodes fail and create a coverage gap around a moving target's position, the network will fail to track the target causing a decrease in tracking performance.

The following literature gaps are studied and addressed in this paper.
\begin{enumerate}
\item Resilient Tracking: A reactive fault recovery method that provides resilience to coverage gaps (caused by co-located node failures, non-uniform node distribution, or very low network densities). Such a network enables a distributed self-healing mechanism that can opportunistically fill the coverage gaps around the moving targets in an energy-efficient manner.
\item Energy-efficient Tracking: A solution to the MNLAR problem for dynamic unknown targets.
\end{enumerate}

\section{Problem Formulation} \label{sec:pf}
Let $\Omega$ $\subset \mathbb{R}^2$ be the ROI with area $A_{\Omega}$. Let $\mathcal{S}=\{s_1,s_2,...s_n\}$ be the set of $n$ heterogeneous (multi-modal) sensor nodes randomly deployed throughout $\Omega$, where each node $s_i$ is static and its position is denoted as $\mathbf{u}^{s_i} \in \Omega$. Additionally, let $\mathcal{T}=\{\tau_1,\tau_2,...\tau_{m}\}$ be the set of $m$ targets traveling through $\Omega$. Let the actual position of a mobile target $\tau_{\ell} \in \mathcal{T}$ at time $k$ be denoted as $\mathbf{u}^{\tau_{\ell}}(k) \in \Omega$.

\subsection{Description of a Sensor Node}
A sensor node is a multi-modal autonomous agent that contains a heterogeneous sensor suite, a data processing unit (DPU), a transmitter/receiver, and a GPS device. The sensor suite contains several Low Power Sensing (LPS) devices which are passive binary detectors consuming very little energy (e.g., Passive Infrared (PIR) sensors). It also contains High Power Sensing (HPS) devices which are active sensors providing the range and azimuth measurements of targets (e.g., Laser Range Finders) \cite{LWCL2009}. The DPU performs  computations to make device scheduling decisions.

The use of LPS and HPS devices on a single node is practical in target tracking applications \cite{K2006} since tracking with only HPS devices is costly \cite{ADBKZNMCDG2004}. This allows the node to first detect a target using the LPS device and then accurately track it by activating the HPS device. It is assumed that the detection areas of LPS and HPS devices are circular. While the LPS devices have a fixed sensing range $R_{LPS}$, the range of HPS devices can be adjusted by controlling the amount of power supplied to the sensors \cite{CWL2006}. Thus, each node $s_i \in \mathcal{S}$ can adjust the range of it's  HPS  device from $L$ levels depending on the need, such that $R_{HPS}^{s_i}(k)$ $\in$ $\{R_{1},R_{2},...R_{L}\}$, where $R_{1}<R_{2}<...<R_{L}$, and $R_1$ is the default HPS sensing range.

\begin{defn}[\textbf{Neighborhood}]
The neighborhood of a node $s_i \in \mathcal{S}$ is defined as
\begin{eqnarray}
\mathcal{N}^{s_i}  \triangleq  \left\{s_j\in \{\mathcal{S}\setminus s_i\}: ||\mathbf{u}^{s_j}-\mathbf{u}^{s_i}|| \leq R_c\right\},
\end{eqnarray}
that includes all nodes within a circle of radius $R_c\ge 2R_L$, which can communicate with the node $s_i$.
\end{defn}

\begin{rem} \label{rem:com} This paper assumes reliable communication using a wireless broadcasting scheme. Future work will study the effects of communication failures on the sensor network's performance.
\end{rem}
\subsection{Energy Consumption and Network Lifetime}

\begin{defn}[\textbf{Energy Consumption}] The energy consumed \cite{CCLS2011} by a node $s_i$ during a $\Delta T$ time interval is defined as
\begin{eqnarray} \label{eq:energy}
E^{s_i}_{\Delta T}(k) & = & \sum_j \chi^{s_i}_j(k)e^{s_i}_{j}.\Delta T,
\end{eqnarray}
where the subscript $j$ $\in$ $\{$LPS, HPS, DPU, transmitter (TX), receiver (RX), clock$\}$; $e^{s_i}_{j}$ is the rate of energy consumed by device $j$ per unit time; and $\chi^{s_i}_j(k)$ $\in$ $\{0,1\}$ is the device status, ON or OFF, at time $k$.
\end{defn}
\vspace{-6pt}
The energy consumption rates $e^{s_i}_j$ in Eq.~(\ref{eq:energy}) are assumed constant for all devices except for the TX and HPS  devices. The transmission energy cost depends on the number of transmissions that have occurred during the $\Delta T$ time step, i.e., $e^{s_i}_{TX}(n_{TX})=n_{TX}e_{TX}$, where $n_{TX}$ is the number of transmissions and $e_{TX}$ is a constant value. The energy cost of the HPS device depends on the adaptive sensing range \cite{LWCL2009} $R_{HPS}^{s_i}(k)$ of the HPS device, such that
\begin{eqnarray}
e^{s_i}_{HPS}\left(R_{HPS}^{s_i}(k)\right) = w R_{HPS}^{s_i}(k), \label{eq:elin}
\end{eqnarray}
where $w$ is the proportionality constant. The total energy consumed by $s_i$ up to time $k$ can be computed as $E^{s_i}(k)$$=$$\sum^k_{k'=1} E^{s_i}_{\Delta T}(k')$. Thus, the total energy consumed by the network is given as
\begin{eqnarray}
E^{Net}(k) = \sum_{s_i \in \mathcal{S}} E^{s_i}(k).
\end{eqnarray}

Since this paper considers mobile targets, the network lifetime is defined as follows.

\begin{defn}[\textbf{Network Lifetime}]\label{defn:nlife}
Consider a trajectory $\gamma$ in the region $\Omega$ that is followed by the maximum number of targets. Now consider a cylindrical tube $\Omega_{\gamma}\subset \Omega$ of radius $R_{LPS}$ around $\gamma$, which contains a set of sensors $\mathcal{S}_{\gamma}\subset \mathcal{S}$. Since the maximum number of targets travel through $\Omega_{\gamma}$, the nodes in  $\mathcal{S}_{\gamma}$ will die earliest in the network. Thus, the expected network lifetime, $\overline{T}_{Life}$, is defined as the time when the energy of sensor nodes in $\mathcal{S}_{\gamma}$ reduces to a certain fraction $\eta \in [0,1)$, s.t.
\begin{align*}\label{eq:nlife}
\frac{\sum_{s_j \in \mathcal{S}_{\gamma}} \left(E^{s_j}_0- E^{s_j}\left(\overline{T}_{Life}\right)\right)}{\sum_{s_j \in \mathcal{S}_{\gamma}} E^{s_j}_0} = \eta
\end{align*}
where $E^{s_j}_0$ is the initial energy of node $s_j$.
\end{defn}

\begin{wrapfigure}{r}{0.34\textwidth} %\vspace{-48pt}
        \includegraphics[width=0.34\textwidth]{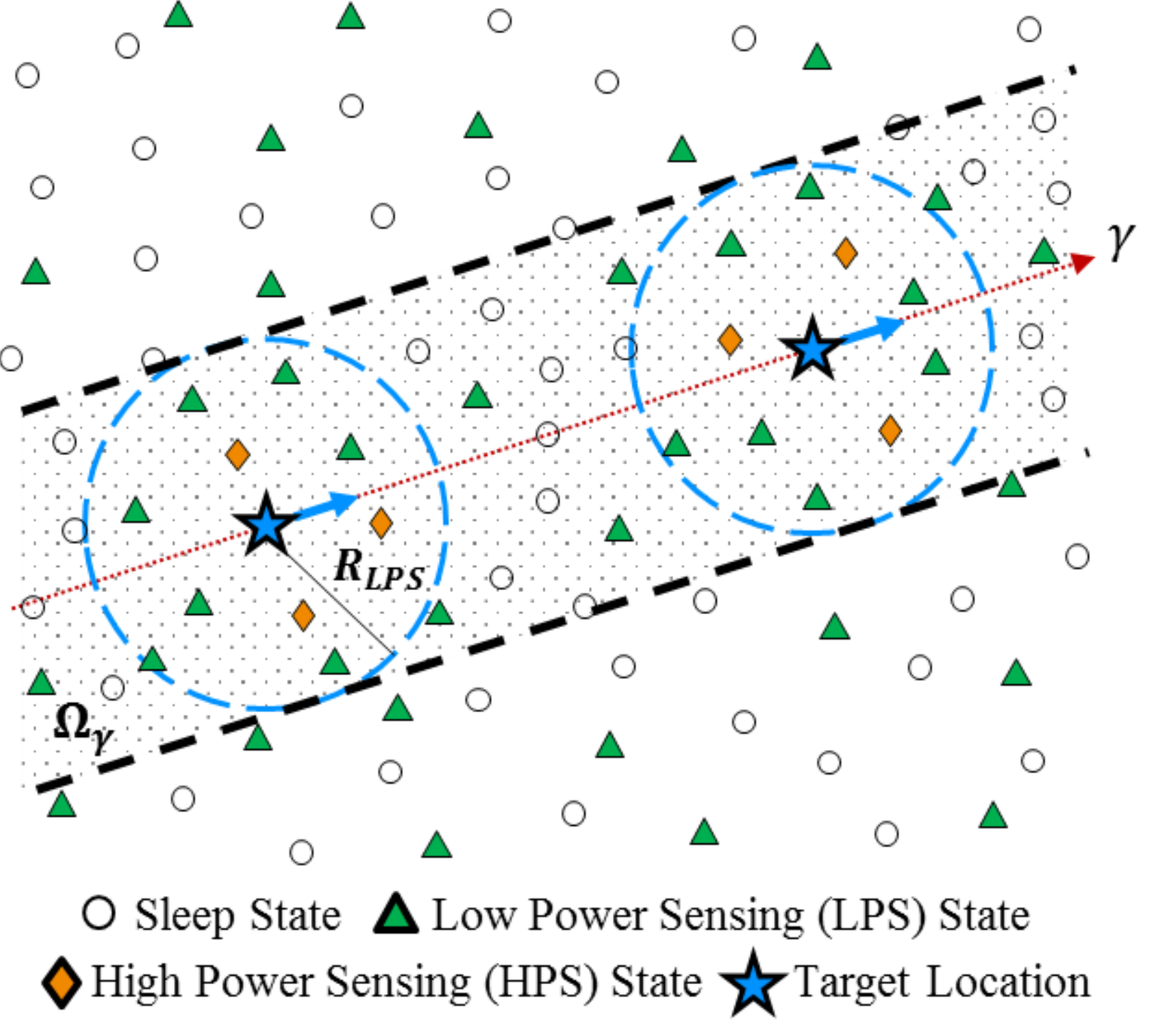}\vspace{-10pt}
        \caption{Example of the tube $\Omega_\gamma$.} \label{fig:net_life_eg}
	\vspace{6pt}
\end{wrapfigure} Fig.~\ref{fig:net_life_eg} shows a tube $\Omega_\gamma$ with two targets. The network lifetime is computed over $\Omega_\gamma$  because the nodes first detect the target in the LPS state, then initialize the target's state in the  HPS  state to start the adaptive node selection process. Thus, once all nodes within $\Omega_\gamma$ deplete their energies, the network will no longer be able to detect and track the targets.

%\vspace{26pt}

\begin{rem} Defn.~\ref{defn:nlife} refers to the worst case when the targets follow the same trajectory. If their trajectories differ, then the tube's width will be expanded, resulting in an increased network lifetime since additional nodes will be available. Additionally, the nodes outside of the tube are operating in a low energy state, i.e., Sleep or LPS state, which allows them to conserve energy, as discussed in Section~\ref{sec:poser}.
\end{rem}

\subsection{Target Coverage and Coverage Degree}
First, we first describe the coverage area of a sensor node and that of the entire sensor network.

\begin{defn}[\textbf{Coverage Area}]
The coverage area of a node $s_i \in \mathcal{S}$ at time $k$ is defined as
\begin{eqnarray}
\Omega^{s_i}(k) \triangleq \left\{\mathbf{u}\in \Omega: ||\mathbf{u}-\mathbf{u}^{s_i}||\le R_{HPS}^{s_i}(k)\right\},
\end{eqnarray}
where it could measure the target using it's HPS devices with sensing range $R_{HPS}^{s_i}(k)$. Thus, the total coverage area of the entire sensor network at time $k$ is $\Omega^{Net}(k) \triangleq \bigcup_{s_i \in \mathcal{S}} \Omega^{s_i}(k)$.
\end{defn}

\begin{wrapfigure}{r}{0.48\textwidth} \vspace{-12pt}
	\begin{subfigure}{0.23\textwidth}
        \centering
        \includegraphics[width=\textwidth]{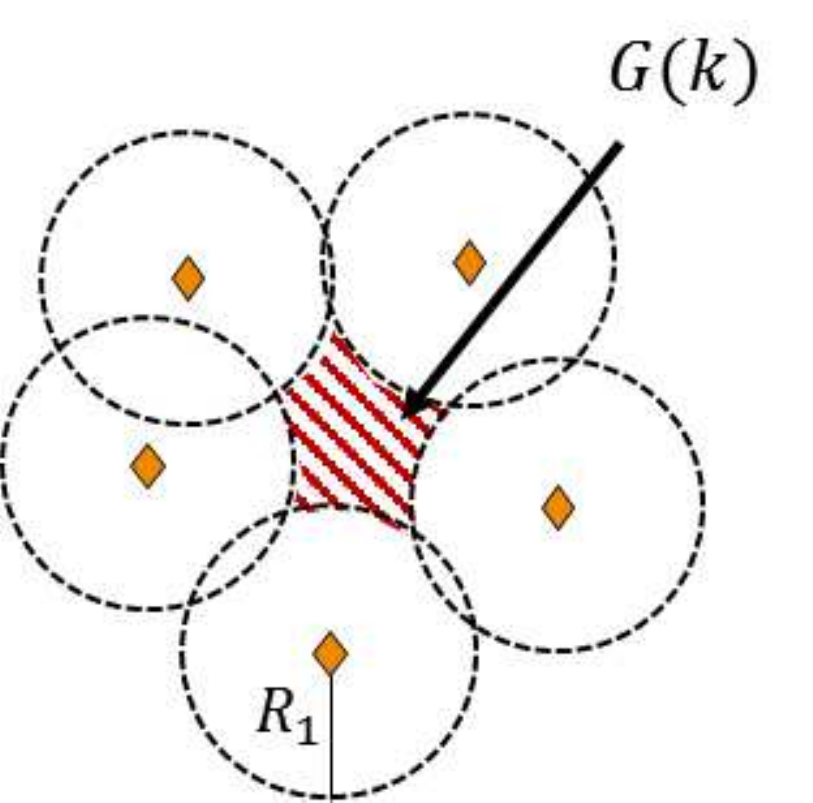}
        \caption{Standard network with $R_{HPS}^{s_i}(k)=R_1$ $\forall k$}\label{fig:cov1}
    \end{subfigure}	
	\begin{subfigure}{0.23\textwidth}
        \centering
        \includegraphics[width=\textwidth]{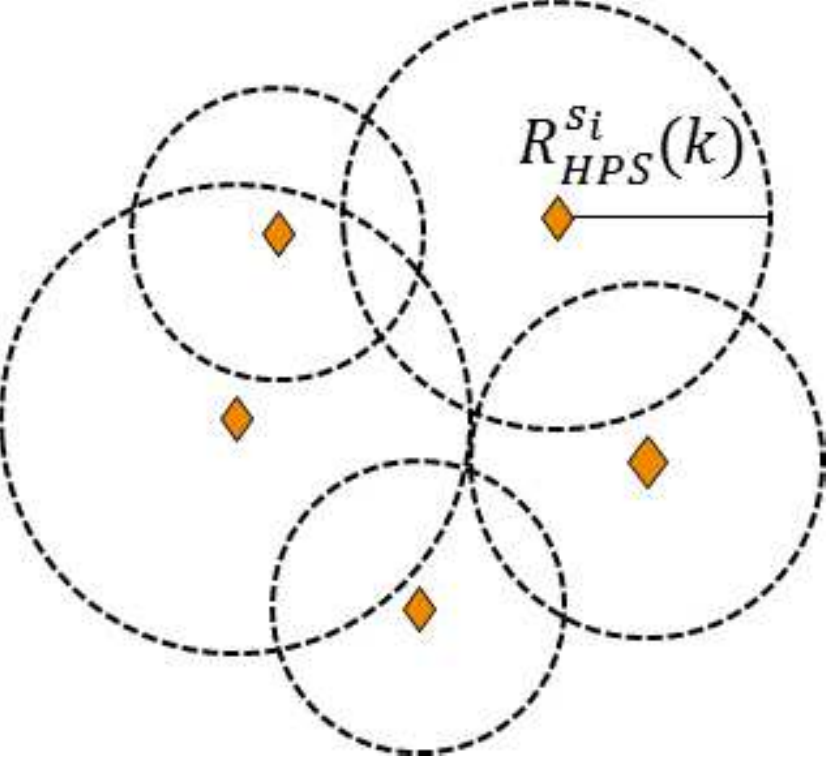}
        \caption{POSE.R network with varying $R_{HPS}^{s_i}(k)$}\label{fig:cov2}
    \end{subfigure}
	\vspace{-6pt}
	\caption{Example of how the POSE.R network adapts to heal a coverage gap present in the network.} \label{fig:cov_gap_eg} \vspace{-12pt}
\end{wrapfigure}
In practice, it is possible that $\Omega^{Net}(k) \subset \Omega$, thus causing coverage gaps, as shown in Fig. \ref{fig:cov_gap_eg}.

\begin{defn}[\textbf{Coverage Gap}]
A connected region $G(k) \neq \emptyset$ is defined as a coverage gap if $G(k) \subseteq \Omega \setminus {\Omega^{Net}(k)}$, that means no sensor node could track the target when it travels in $G(k)$.
\end{defn}

\begin{rem}
Coverage gaps could be present due to sparse or non-uniform initial node deployment, or they may also gradually develop over time due to sensor failures or other reasons. Thus, the goal of the POSE.R algorithm is to expand the HPS sensing ranges of selected nodes around the target to opportunistically heal the coverage gaps present in the network, as seen in Fig.~\ref{fig:cov_gap_eg}.
\end{rem}

Next, we define target coverage.
\begin{defn}[\textbf{Target Coverage}]
A target $\tau_\ell \in \mathcal{T}$ is said to be covered at time $k$, if $\mathbf{u}^{\tau_\ell}(k) \notin \Omega/{\Omega^{Net}(k)}$, that is it does not fall in any coverage gap. For the full target set $\mathcal{T}$, target coverage is said to be complete at time $k$, if coverage is achieved for $\forall \tau_\ell \in \mathcal{T}$.
\end{defn}

Next, we define the concept of target coverage degree.

\begin{defn}[\textbf{Target Coverage Degree}]
The coverage degree $D^{\tau_{\ell}}(k)$ of a target $\tau_{\ell}$ is defined as the number of nodes that are covering the target at time $k$.
\end{defn}

To ensure high tracking accuracy and low missed detection rates, POSE.R performs distributed sensor fusion for target state prediction, thus we formulate the target coverage problem such that $D^{\tau_{\ell}}(k)$=$N_{sel}$ $>2$, $\forall$ $k$. This ensures that geometrically diverse state estimates are fused to improve the state estimation accuracy. At the same time, $N_{sel}$ should be small for energy-efficiency and low complexity since state fusion complexity increases as the number of states increases. In this paper, we consider $N_{sel}=3$ to improve state estimation and fusion while reducing the overall complexity. However, the network designer can select this parameter based on his specific requirements.

\begin{wrapfigure}{r}{0.32\textwidth} \vspace{12pt}
        \includegraphics[width=0.32\textwidth]{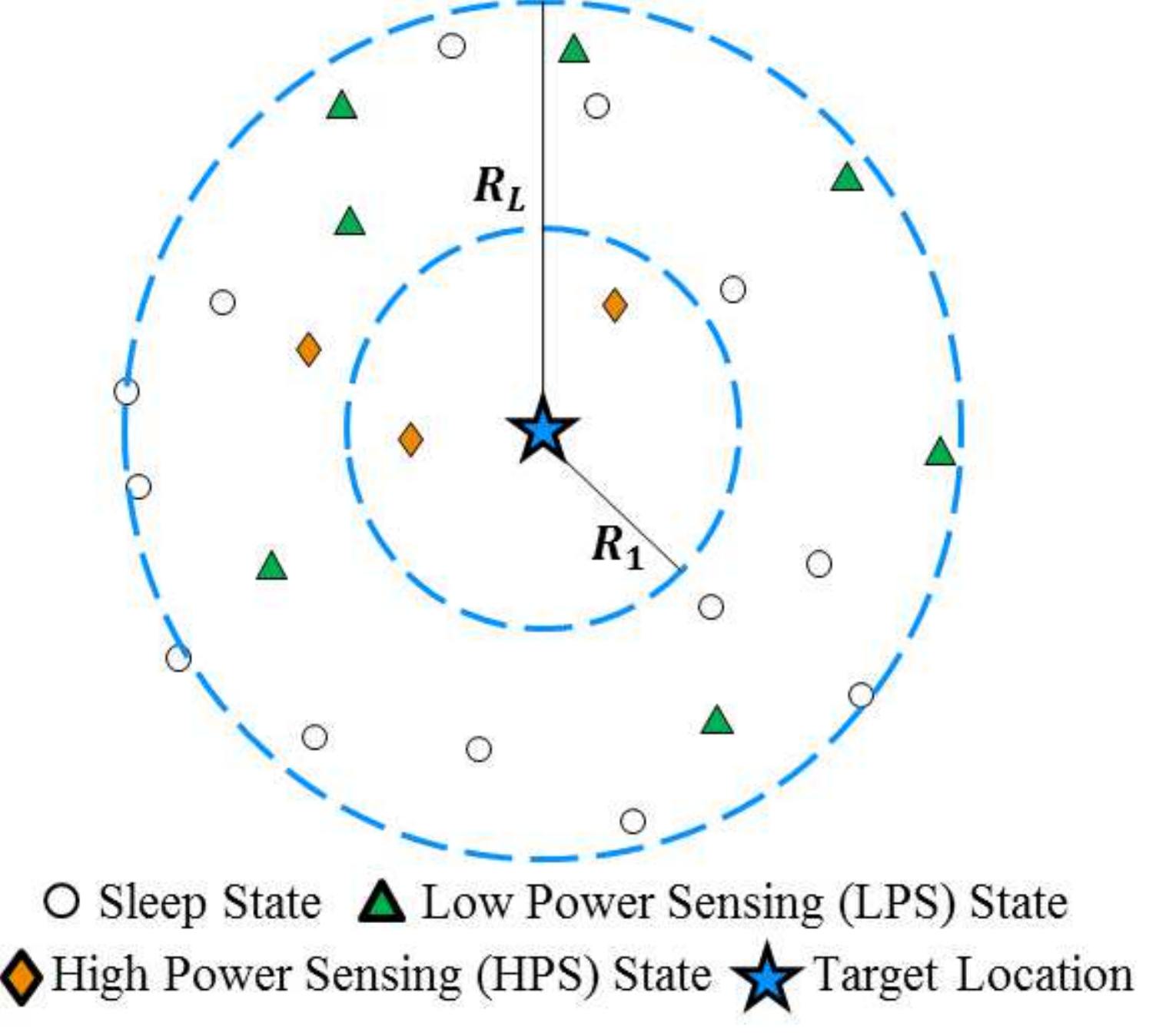}\vspace{-3pt}
        \caption{Example of a base coverage degree $D_b^{\tau_\ell}(k)=2$ and an extended coverage degree $D_e^{\tau_\ell}(k)=3$. } \label{fig:base_ext_deg_eg}
	\vspace{-20pt}
\end{wrapfigure}

The target coverage degree is further defined to be one of the following two types.

\begin{defn}[\textbf{Base and Extended Coverage Degrees}]
The base coverage degree $D_b^{\tau_{\ell}}(k)$ of a target $\tau_{\ell} \in \mathcal{T}$ at time $k$ is defined as the number of nodes that are covering the target with their base sensing range $R_1$. Similarly, the extended coverage degree $D_e^{\tau_{\ell}}(k)$ of a target $\tau_{\ell} \in \mathcal{T}$ at time $k$ is defined as the number of nodes that are covering the target with their base as well as extended sensing ranges in the set $\{R_1,..., R_L\}$.
\end{defn}

An example of the base and extended coverage degrees is shown in Fig.~\ref{fig:base_ext_deg_eg}.  Here, there are only $2$  HPS  nodes that are capable of covering the target with a range $R_1$, while there are $3$  HPS  nodes that can cover the target with any sensing range. Thus, the base coverage degree is $D_b^{\tau_\ell}(k)=2$ and the extended coverage degree is $D_e^{\tau_\ell}(k)=3$.

\begin{rem}
Extended coverage is required at time $k$ only if the base coverage degree is insufficient, i.e., if $D_b^{\tau_{\ell}}(k)$<$N_{sel}$. This is described in Section~\ref{sec:ass}.
\end{rem}

\subsection{Target Detection and Measurement} \label{sec:detectnmeasure}
After describing the sensor node, energy consumption, and target coverage, here we describe how a target is actually detected and measured by sensors. The motion of a target, $\tau_{\ell}$, is modeled using a \emph{Discrete White Noise Acceleration} (DWNA) model \cite{BLK2004} as follows
\begin{eqnarray}\label{eq:1}
\mathbf{x}(k+1) & = & \mathbf{f}(\mathbf{x}(k),k) + \boldsymbol{\upsilon}(k),
\end{eqnarray}
where $\mathbf{x}(k)\triangleq[x(k), \dot{x}(k), y(k), \dot{y}(k), \psi(k)]'$ is the target state at time $k$, which includes the position $(x(k),y(k))$, velocity $(\dot{x}(k),\dot{y}(k))$, and turning rate $\psi(k)$; $\mathbf{f}(\mathbf{x}(k),k)$ is the state transition matrix, $\boldsymbol{\upsilon}(k)$ is the zero-mean white Gaussian process noise. In this work, it is assumed that the target travels according to the nearly coordinated turning model~\cite{BLK2004}.

A sensor node $s_i$ can use it's  LPS  devices for target detection. We adopt the detection model proposed in \cite{ZC2004}. The probability of $s_i$ detecting a target $\tau_\ell$ is given as:
\begin{equation} \label{eq:pd}
P_{LPS}^{\tau_{\ell},s_i}(k) = \left\{
\begin{array}{ll}
\alpha & d(\tau_{\ell},s_i) < R_{r} \\
 \alpha e^{-\beta (d(\tau_{\ell},s_i)-R_{r})} & R_{r} \le d(\tau_{\ell},s_i) \le R_{LPS} \\
\end{array} \right.
\end{equation}
where $d(\tau_\ell,s_i) = ||\mathbf{u}^{\tau_\ell}(k)-\mathbf{u}^{s_i}||$; $R_{r}$ is the reliable sensing radius of the LPS device; $\alpha$ is the detection probability within $R_r$; and $\beta$ is the decay rate of detection probability with distance greater than
$R_r$. If the target lies beyond $R_{LPS}$, then $s_i$ can receive false alarms with a probability $p_{fa}=1-e^{-f_r \Delta T}$~\cite{W2008}, where $f_r$ is the false alarm rate during a $\Delta T$ second scan.

On the other hand, a node $s_i$ can use it's  HPS  devices to collect the measurements, $\mathbf{z}(k)=\{\mathbf{z}_j(k)\}_{ j=1,...o}$, of the target at time $k$, such that
\begin{eqnarray} \label{eq:tar_mod}
\mathbf{z}_j(k) & = & \mathbf{h}(\mathbf{x}(k),k) + \mathbf{w}(k),
\end{eqnarray}
where each $\mathbf{z}_j(k)$ includes the range and azimuth measurements; $\mathbf{h}(\mathbf{x}(k),k)$ is a nonlinear measurement model that translates the target's state into a measurement $\mathbf{z}_j(k)$ \cite{BLK2004}; and $\mathbf{w}(k)$ is the zero-mean white Gaussian measurement noise. The measurements of $\tau_\ell$ are received by $s_i$ with a probability $P_{HPS}^{\tau_{\ell},s_i}(k)=p_d$,  if $d(\tau_{\ell},s_i) \leq R_{HPS}^{s_i}(k)$, where $p_d$ is the probability of detection of the HPS sensor. It is assumed that even if the targets are blocking each other, a measurement is received for each target detected within the HPS sensing range. As future work, more realistic detection models will be considered. Furthermore, the measurements $\mathbf{z}_j(k)$ may also contain some false measurements along with the true target measurements due to the target traversing through a cluttered environment. The number of false measurements received at each time step $k$ are generated according to a Poisson distribution with mean $\mu_{cl}$~\cite{BDH2009}. The locations of false measurements are drawn from a uniform distribution within the node's coverage area.

\subsection{Objective} \label{obj}
The main objective of the target tracking problem addressed in this paper is to develop a distributed autonomy approach that employs a node-level probabilistic switching control of the devices to achieve energy-efficiency and resilience, while maintaining high tracking accuracy and low missed detection rates. The two primary features of the POSE.R network are discussed below.
\begin{itemize}
\item [1.] \textbf{Energy-efficiency:} This is essential to improve the network lifetime. For energy-efficiency, POSE.R performs opportunistic sensing, where the aim is to form a cluster of nodes with their HPS  devices activated, in regions around the current and predicted positions of the target. The nodes away from these regions preserve energy by either using LPS devices to stay aware or sleeping. For this purpose, it is necessary to predict the target's state at every time step via distributed fusion. This is followed by distributed adaptive node selection around the predicted state of the target to form a cluster of optimal nodes with high energies and geometric diversity. The cluster size is chosen small ($N_{sel}$=3) to avoid computational burden of distributed optimization and to save energy. These selected nodes track the target with high accuracy. As target moves, this cycle continues with dynamic cluster selection to maintain continuous target tracking with significant energy savings.

\vspace{3pt}
\item [2.] \textbf{Resilience:} This is essential to maintain the tracking performance in regions of low node density or coverage gaps (caused by node failures, or non-uniform/sparse node distribution). In practical networks, the tracking performance can degrade and the target can be lost while travelling inside the coverage gaps, and when it reappears, state re-initialization is required to start tracking it again. In this regard, resilience imparts the network with the capability of opportunistic self-healing to track the target even when it passes through a coverage gap by proactively extending the sensing ranges of selected nodes. For this purpose, first a cluster is formed around each target's predicted position using a node selection process. Then, the coverage degree is computed by each cluster independently. If $D_b^{\tau_{\ell}}(k)$<$N_{sel}$, then POSE.R performs distributed optimization to select nodes outside the regular sensing range around the targets' predicted positions, to achieve $D_e^{\tau_{\ell}}(k)$=$N_{sel}$. These selected nodes can then optimally extend their HPS ranges to maximize coverage while minimizing energy consumption.  By optimal extension of the ranges of these selected  HPS sensors, the coverage gap reduces or even completely disappears during the transition of a target.
\end{itemize}

The formal objective functions for the above are discussed in Section~\ref{sec:ass}.

\section{POSE.R Algorithm} \label{sec:poser}
This section describes the POSE.R algorithm where each sensor node is equipped with a $PFSA$-based supervisor for distributed probabilistic control of its devices, as shown in Fig.~\ref{fig:main}.

\begin{defn}[\textbf{PFSA}]
A PFSA~\cite{VTDCC2005} is defined as a $3$-tuple $\Xi=\langle \Theta,A, P \rangle$, where
\begin{itemize}
\item $\Theta$ is a finite set of states,
\vspace{1pt}
\item $A$ is a finite alphabet,
\vspace{1pt}
\vspace{1pt}
\item $p: \Theta \times \Theta \rightarrow [0,1]$ are the state transition probabilities which form a stochastic matrix $P\equiv[p_{i,j}]$, where $p_{i,j}\equiv p(\theta_i,\theta_j)$, $\forall \theta_i,\theta_j \in \Theta$, s.t. $\sum_{\theta' \in \Theta} p(\theta,\theta') = 1, \ \forall \theta \in \Theta.$
\end{itemize}
\end{defn}

The alphabet $A = \{\epsilon, 0, 1\}$, where $\epsilon$ is the null symbol emitted when no information is available, $0$ indicates no target detection, and $1$ indicates target detection. A symbol is emitted at each state transition, thus a symbol sequence is generated which keeps track of the node's target detection history.
The state set $\Theta$ consists of three states: \emph{Sleep} ($\theta_1$),  LPS  ($\theta_2$), and HPS ($\theta_3$), as shown in Fig.~\ref{fig:main}.

Consider a node $s_i\in \mathcal{S}$ which can operate in one of the three states at one time. The $PFSA$-based supervisor runs a unique algorithm within each state to dynamically update it's state transition probabilities based on the information acquired about targets' whereabouts. These probabilities control the transition of the node from one state to another. The details of this probabilistic switching control are presented in Alg.~\ref{alg:lps}. A summary of the algorithms within each state are described below.

\begin{figure*}[t!]
	\begin{subfigure}{\textwidth}
        \centering
        \includegraphics[width=.55\textwidth]{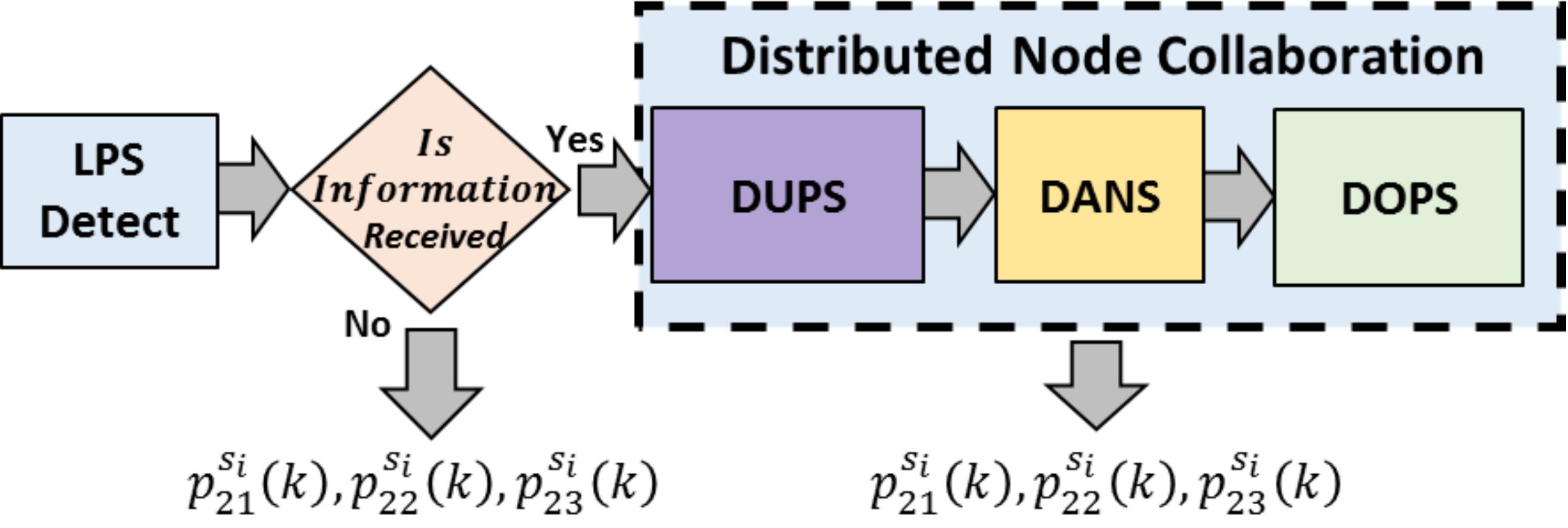}
        \caption{Low Power Sensing State Flowchart}\label{fig:lps}
    \end{subfigure}	
	\begin{subfigure}{\textwidth}
        \centering
        \includegraphics[width=.9\textwidth]{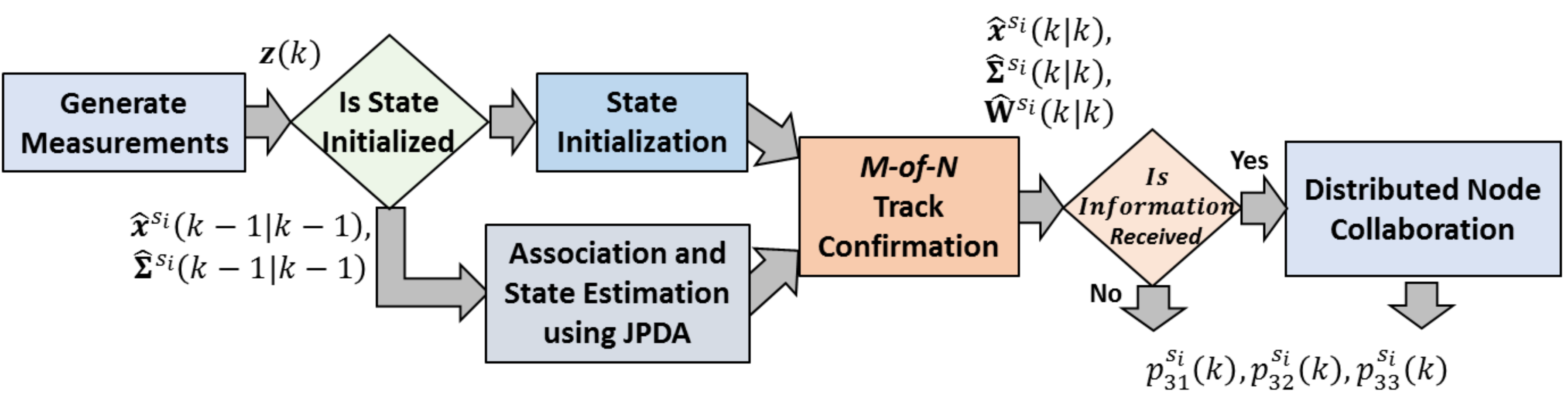}
        \caption{High Power Sensing State Flowchart}\label{fig:hps}
    \end{subfigure}
	\vspace{-6pt}
	\caption{Flowcharts for the algorithms within the LPS and  HPS  states. The distributed node collaboration consists of: DUPS (Distribution Fusion for Prediction of Target State), DANS (Distributed Adaptive Node Selection), and DOPS (Distributed Computation of the Probability of Success of Target Detection).} \label{fig:alg} \vspace{-12pt}
\end{figure*}

\vspace{-6pt}
\subsection{Sleep State} \label{sec:sleep}
The \emph{Sleep} state, $\theta_1$, is designed to minimize energy consumption by disabling all devices on the node $s_i$ except for a clock and the DPU to allow for state transitions. After every time interval $\Delta T$, $s_i$ can continue to sleep with a probability $p^{s_i}_{1,1}(k)=p_{sleep}$ or it can transition to the  LPS  state with a probability $p^{s_i}_{1,2}(k)=1-p_{sleep}$, where $p_{sleep} \in [0,1]$ is a design parameter. From the \emph{Sleep} state, $s_i$ cannot directly transition to the  HPS  state, i.e. $p^{s_i}_{1,3}(k)=0$. \textbf{Line 3 of Alg.~\ref{alg:lps}} shows the state transition probabilities. A node reaches the \emph{Sleep} state if the target is located far away or if the node is not selected for tracking.

\vspace{-6pt}
\subsection{Low Power Sensing State} \label{sec:lps}
The LPS state, $\theta_2$, is designed to detect the target and stay aware while conserving energy. In this state, the DPU, the transceiver, and the  LPS  devices are enabled while the  HPS  devices are disabled. Fig.~\ref{fig:lps} shows the flowchart for the algorithm, which is described below.

\begin{algorithm}[t!]
\SetKwInOut{Input}{input}\SetKwInOut{Output}{output}
\SetKwComment{Comment}{}{}
\SetCommentSty{footnotesize}
\DontPrintSemicolon
\caption{Probabilistic State Switching Control of $s_i$.}
\label{alg:lps}
\Input{$\mathbf{u}^{s_i}$, $\mathcal{N}^{s_i}_{HPS}$, $\hat{\mathbf{I}}^{s_i}(k)$, $N_{sel}$, $\theta^{s_i}(k)$, and $H$}
\Output{$p_{id,1}^{s_i}(k), p_{id,2}^{s_i}(k)$, and $p_{id,3}^{s_i}(k)$}
\If(\tcp*[f]{In Sleep state}){$\theta^{s_i}(k)=\theta_1$}{
$id \leftarrow 1$\\
$p_{id,1}^{s_i}(k)\leftarrow p_{sleep}, \ \ p_{id,2}^{s_i}(k)\leftarrow 1-p_{sleep}, \ \ p_{id,3}^{s_i}(k)\leftarrow 0$}
\If(\tcp*[f]{In LPS state}){$\theta^{s_i}(k)=\theta_2$}{
$id \leftarrow 2$\\
\uIf(\tcp*[f]{No information is received from neighbors}){$\mathcal{N}^{s_i}_{HPS} = \emptyset$}{
{$p_{id,1}^{s_i}(k) \leftarrow 1-P_{LPS}^{\tau_{\ell},s_i}(k), \ \ p_{id,2}^{s_i}(k) \leftarrow 0, \ \ p_{id,3}^{s_i}(k) \leftarrow P_{LPS}^{\tau_{\ell},s_i}(k)$ }}}
\If(\tcp*[f]{In HPS  state}){$\theta^{s_i}(k)=\theta_3$}{
$id \leftarrow 3$\\
\uIf(\tcp*[f]{No information is received from neighbors}){$\mathcal{N}^{s_i}_{HPS} = \emptyset$}{
{
$p_{id,1}^{s_i}(k) \leftarrow 0, \ \ p_{id,2}^{s_i}(k) \leftarrow 1-P_{HPS}^{\tau_{\ell},s_i}(k), \ \ p_{id,3}^{s_i}(k) \leftarrow P_{HPS}^{\tau_{\ell},s_i}(k)$}}}
\If(\tcp*[f]{In LPS or  HPS  state, and information was received from neighbors. Call distributed node collaboration (DNC)}){$\theta^{s_i}(k)\in \{\theta_2,\theta_3\}$ \& $\mathcal{N}^{s_i}_{HPS} \ne \emptyset$}{
	$\left\{\hat{\mathbf{x}}^{s_i}(k+1|k), [\mathcal{S}^*, R^*, \hat{P}_{HPS}^{s_i}](k+1)\right\}$$=$DNC($\hat{\mathbf{I}}^{s_i}(k)$) \\
		\uIf(\tcp*[f]{Node $s_i$ is selected as an optimal node}){$s_i \in \mathcal{S}^*(k+1)$}{
		$p_{id,1}^{s_i}(k) \leftarrow 0, \ \ p_{id,2}^{s_i}(k) \leftarrow 1-\hat{P}_{HPS}^{s_i}(k+1), \ \ p_{id,3}^{s_i}(k) \leftarrow \hat{P}_{HPS}^{s_i}(k+1)$
		}
	\ElseIf(\tcp*[f]{Node $s_i$ is not selected as an optimal node}){$s_i \notin \mathcal{S}^*(k+1)$}{
		\uIf(\tcp*[f]{Target's predicted position is within the range $R_1$ of $s_i$}){$||\mathbf{u}^{s_i}-H\hat{\mathbf{x}}^{s_i}(k+1|k)||\le R_1$}{
        $R_{HPS}^{s_i}(k+1)=R_1$\\
		$p_{id,1}^{s_i}(k) \leftarrow 1-\hat{P}_{HPS}^{s_i}(k+1), \ \ p_{id,2}^{s_i}(k) \leftarrow \hat{P}_{HPS}^{s_i}(k+1), \ \ p_{id,3}^{s_i}(k) \leftarrow 0$
		}		
		\ElseIf(\tcp*[f]{Target's predicted position is farther than the range $R_1$ of $s_i$}){$||\mathbf{u}^{s_i}-H\hat{\mathbf{x}}^{s_i}(k+1|k)||> R_1$}{
        Compute $D^{\tau_{\ell}}_{b}(k+1)$ \\
		\uIf(\tcp*[f]{Base coverage degree is sufficient}){$D^{\tau_{\ell}}_{b}(k+1)=N_{sel}$}
		{
		$p_{id,1}^{s_i}(k) \leftarrow 1, \ \ p_{id,2}^{s_i}(k) \leftarrow 0, \ \ p_{id,3}^{s_i}(k) \leftarrow 0$
		}
		\ElseIf(\tcp*[f]{Base coverage degree is insufficient}){$D^{\tau_{\ell}}_{b}(k+1)<N_{sel}$}
		{
		$R_{HPS}^{s_i}(k+1)=R_L$\\
		$p_{id,1}^{s_i}(k) \leftarrow 1-\hat{P}_{HPS}^{s_i}(k+1), \ \ p_{id,2}^{s_i}(k) \leftarrow \hat{P}_{HPS}^{s_i}(k+1), \ \ p_{id,3}^{s_i}(k) \leftarrow 0$
		}
		}
}
}
\end{algorithm}

\vspace{6pt}
\subsubsection{Target Detection} In the  LPS  state target detection can occur by two means: (i) using the  LPS  devices and/or (ii) by fusing the target state information received from the neighbors. If a target $\tau_{\ell}$ is located within $R_{LPS}$ of $s_i$, then $s_i$ can detect it with a probability $P_{LPS}^{\tau_{\ell},s_i}$, as in Eq.~(\ref{eq:pd}).

Next, $s_i$ checks if it has received any information from the  HPS  sensors in it's neighborhood (see Section~\ref{sec:dsc} for details). Let $\mathcal{N}_{HPS}^{s_i}\subseteq \mathcal{N}^{s_i}$ be the set of nodes in the  HPS  state in the neighborhood of $s_i$, which have broadcasted the target state information.  If $\mathcal{N}_{HPS}^{s_i} = \emptyset$, i.e., no information is received from neighbors (\textbf{Line 5, Alg.~\ref{alg:lps}}), then $s_i$ transitions to the  HPS  state solely based on its own $P_{LPS}^{\tau_{\ell},s_i}(k)$. The corresponding updates to the state transition probabilities are shown in \textbf{Line 7, Alg.~\ref{alg:lps}}. On the other hand, if $\mathcal{N}_{HPS}^{s_i} \ne \emptyset$, i.e., information is received from neighbors (\textbf{Line 13, Alg.~\ref{alg:lps}}), then $s_i$ performs distributed node collaboration (DNC) (\textbf{Line 14, Alg.~\ref{alg:lps}}) to make an informed switching decision as described below.

\vspace{6pt}
\subsubsection{Distributed Node Collaboration (DNC)} \label{sec:lps_dsc}
This consists of the following three steps below. (Full details are in Section \ref{sec:dsc}.)

\begin{itemize}
\item [i.] DUPS (\underline{D}istributed F\underline{u}sion for \underline{P}rediction of Target \underline{S}tate): In this step, $s_i$ fuses the received information to obtain a target state prediction $\hat{\mathbf{x}}^{s_i}(k+1|k)$.

\item  [ii.] DANS (\underline{D}istributed \underline{A}daptive \underline{N}ode \underline{S}election): In this step, the predicted state is used for:
\begin{itemize}
\item [a)] selecting the optimal set of nodes, $\mathcal{S}^*(k+1)$, to track the target at time $k+1$, and
\item [b)] selecting their optimal sensing ranges, $\mathcal{R}^*(k+1)=\left\{R_{HPS}^{s_j}(k+1)\right\}_{\forall s_j\in \mathcal{S}^*(k+1)}$, to maximize target coverage and minimize energy consumption.
\end{itemize}
\item [iii.] DOPS (\underline{D}istributed C\underline{o}mputation of the \underline{P}robability of \underline{S}uccess of Target Detection, $\hat{P}_{HPS}^{s_i}(k+1)$): In this step, node $s_i$ computes it's probability of successfully detecting the target at time $k+1$ considering the uncertainty in target's state prediction (details are in Eq.~(\ref{eq:phps})).
\end{itemize}

\vspace{6pt}
\subsubsection{Computation of the State Transition Probabilities after DNC} \label{sec:lps_stp}
If $s_i \in \mathcal{S}^*(k+1)$ (\textbf{Line 15, Alg.~\ref{alg:lps}}), then it uses $\hat{P}_{HPS}^{s_i}(k+1)$ to update the state transition probabilities (\textbf{Line 16, Alg.~\ref{alg:lps}}). However, if $s_i \notin \mathcal{S}^*(k+1)$ (\textbf{Line 17, Alg.~\ref{alg:lps}}), then it implies that there are other better nodes to track the target. In this case, if $s_i$ is located within $R_1$ of the target's predicted position (\textbf{Line 18, Alg.~\ref{alg:lps}}), then although it is not selected, it should still stay in the  LPS  state to  participate in node selection during the next time step to facilitate continuous tracking (\textbf{Lines 19-20, Alg.~\ref{alg:lps}}). This is important as the current selected nodes in $\mathcal{S}^*(k+1)$ may not be suitable for tracking at the next time step and thus we need other candidate nodes for the next round of node selection. (Note that sleeping nodes don't participate in node selection). On the other hand, if $s_i$ is located at a distance $>R_1$ from the target's predicted position (\textbf{Line 21, Alg.~\ref{alg:lps}}), then it computes the base coverage degree $D^{\tau_{\ell}}_{b}(k+1)$ (\textbf{Line 22, Alg.~\ref{alg:lps}}). If $D^{\tau_{\ell}}_{b}(k+1)=N_{sel}$ (\textbf{Line 23, Alg.~\ref{alg:lps}}), then $s_i$ goes to $Sleep$ with probability 1 (\textbf{Line 24, Alg.~\ref{alg:lps}}). If $D^{\tau_{\ell}}_{b}(k+1)< N_{sel}$ (\textbf{Line 25, Alg.~\ref{alg:lps}}), then $s_i$ needs to be in the  LPS  state (\textbf{Lines 26-27, Alg.~\ref{alg:lps}}). The only way $D^{\tau_{\ell}}_{b}(k+1)< N_{sel}$ is possible if there are insufficient sensors within $R_1$ of the target's predicted position, i.e., it is a low density area or a coverage gap. This implies that at least some of the selected nodes are chosen from the region lying between $R_1$ to $R_L$ of the target. These nodes must then expand their  HPS ranges to achieve $D_e^{\tau_{\ell}}(k+1)=N_{sel}$. Therefore, the nodes not selected within $R_L$ should stay in the LPS state to participate in node selection as future candidates to track the target.

\subsection{High Power Sensing State} \label{sec:hps}
The  HPS  state, $\theta_3$, is designed to track the target and estimate it's state using the measurements from  HPS  devices. In this state, the DPU, the transceiver and the  HPS  devices are enabled while the  LPS  devices are disabled. Figure~\ref{fig:hps} shows the flowchart of the algorithm, which is described below.

\begin{figure*}[t]
\centering
    \begin{subfigure}[t]{.24\textwidth}
        \centering
        \includegraphics[width=\textwidth]{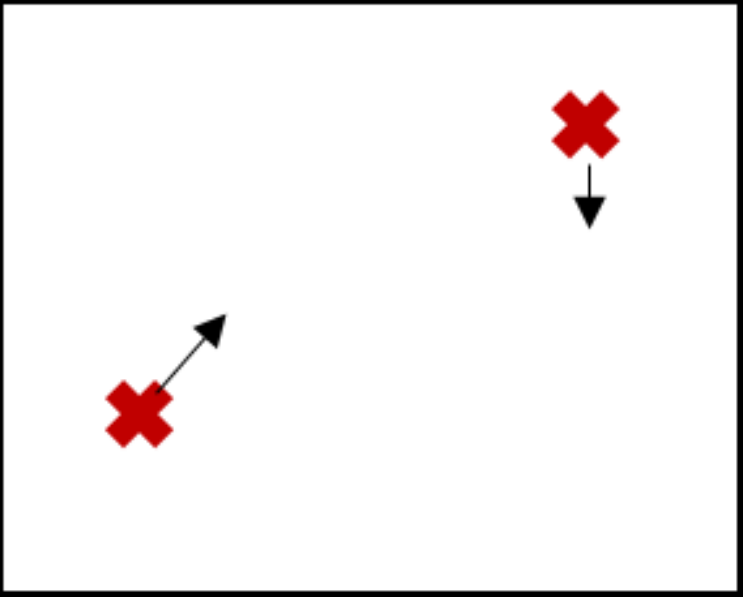}
        \caption{Previous state estimates}
    \end{subfigure}	
    \begin{subfigure}[t]{.24\textwidth}
        \centering
        \includegraphics[width=\textwidth]{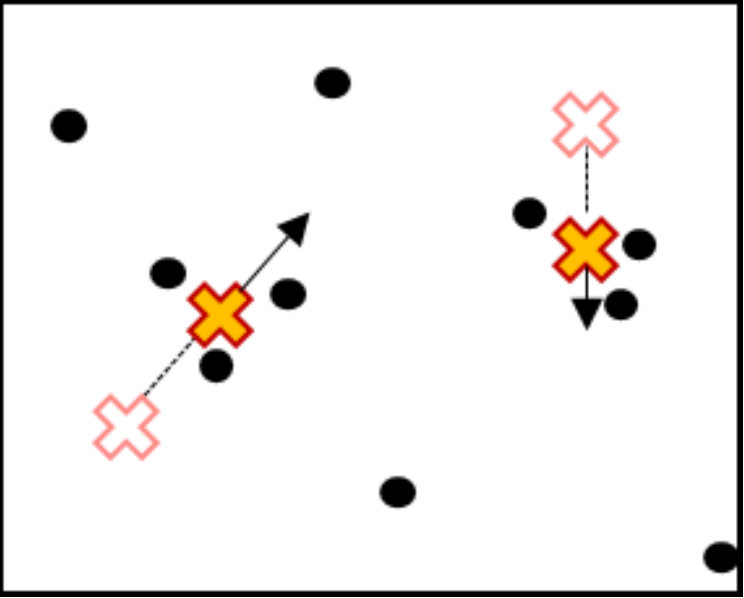}
        \caption{Predicted states and measurement set}
    \end{subfigure}	
    \begin{subfigure}[t]{.24\textwidth}
        \centering
        \includegraphics[width=\textwidth]{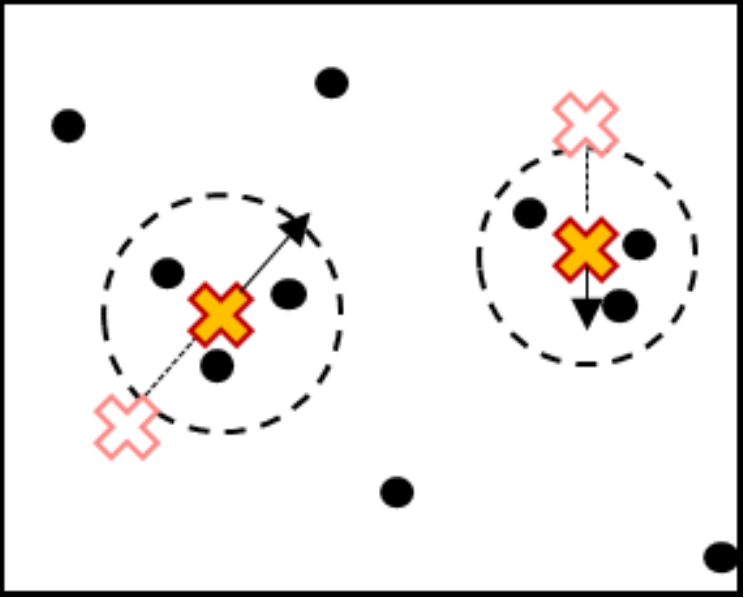}
        \caption{Associated measurements}
    \end{subfigure}	
    \begin{subfigure}[t]{.24\textwidth}
        \centering
        \includegraphics[width=\textwidth]{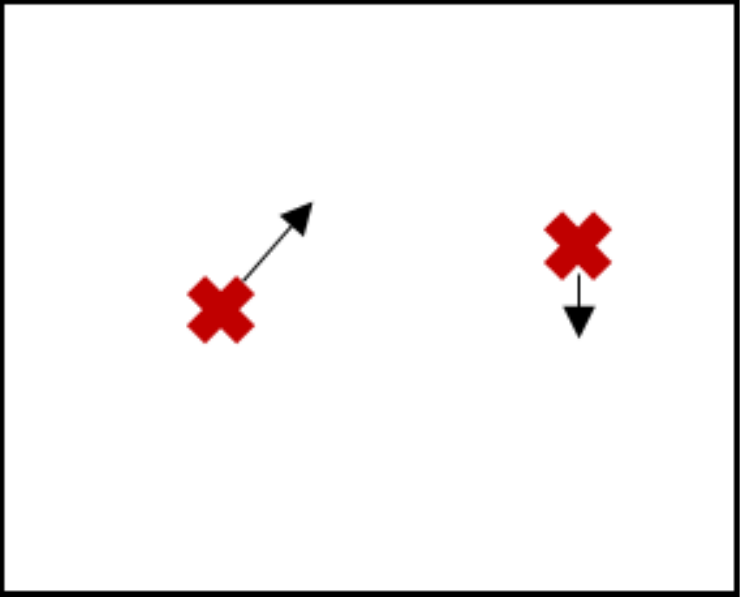}
        \caption{Updated state estimates}
    \end{subfigure}	
    \caption{Simple example of the Joint Probabilistic Data Association Filter. This process consists of first: (a) obtaining the previous state estimates; (b) predicting the state estimates during the next time step $k$ and collecting a set of measurements, as  shown by filled "x"s and black dots respectively; then, (c) the measurements are associated to the predicted state estimates based on the validation region shown in dotted circles; and finally, (d) the validated measurements are used to update the predicted state estimates.}
    \label{fig:jpda_example} \vspace{-12pt}
\end{figure*}

\vspace{6pt}
\subsubsection{Data Association and State Estimation} In the  HPS state, node $s_i$ first collects a set of measurements, $\mathbf{z}(k)$, from it's  HPS  devices with sensing range $R_{HPS}^{s_i}(k)$, where $R_{HPS}^{s_i}(k)$ was selected during the previous time step as part of the node selection process. Subsequently, the track is estimated by a Gaussian distribution with the state and covariance estimates, $\hat{\mathbf{x}}^{s_i}(k|k)$ and $\hat{\mathbf{\Sigma}}^{s_i}(k|k)$, respectively. The previous $\hat{\mathbf{x}}^{s_i}(k-1|k-1)$, $\hat{\mathbf{\Sigma}}^{s_i}(k-1|k-1)$ are updated using the \emph{Joint Probabilistic Data Association} (JPDA) method \cite{BDH2009} to generate $\hat{\mathbf{x}}^{s_i}(k|k)$ and $\hat{\mathbf{\Sigma}}^{s_i}(k|k)$.  Additionally, during the JPDA update step, the node maintains the Kalman filter gain matrix $\hat{\mathbf{W}}^{s_i}(k)$ to be utilized in the DNC algorithm. A simple example of the JPDA process is shown in Fig.~\ref{fig:jpda_example}. If the measurements do not associate to a previous state estimate, $s_i$ initializes a new state estimate \cite{BCB1989}.

\vspace{6pt}
\subsubsection{M-of-N Track Confirmation}
The HPS device measurements may contain false measurements at each time step, as discussed in Section~\ref{sec:detectnmeasure}. This can cause $s_i$ to initialize a new state estimate if a false measurement does not associate to a previous estimate. To account for false measurements and to ensure that a false track is not propagated throughout the network, $s_i$ utilizes the $M$-of-$N$ \emph{Track Confirmation Logic} \cite{CC2005} to allow the network to be robust to false measurements. This approach ensures that $M$ out of $N$ consecutive measurements are associated to a target state estimate before the node confirms that it is not a false track. Furthermore, once a target track has been confirmed, the node can only drop the track if $M$ consecutive measurements do not associate to it. Subsequently, the confirmed target's state and covariance estimates, $\hat{\mathbf{x}}^{s_i}(k|k)$ and $\hat{\mathbf{\Sigma}}^{s_i}(k|k)$, and the filter gain matrix, $\hat{\mathbf{W}}^{s_i}(k)$, are broadcasted.

Next, $s_i$ checks if it has received any information from  HPS  sensors in it's neighborhood $\mathcal{N}_{HPS}^{s_i}$. Since $s_i$ is in the  HPS  state and has broadcasted information to it's neighbors, the set of  HPS  sensors is redefined as $\mathcal{N}_{HPS}^{s_i}= \mathcal{N}_{HPS}^{s_i} \cup \{s_i\}$. However, if $s_i$ has not transmitted a confirmed track, $\mathcal{N}_{HPS}^{s_i}$ does not include $s_i$. If $\mathcal{N}^{s_i}_{HPS}=\emptyset$ (\textbf{Line 9, Alg.~\ref{alg:lps}}), i.e., no information is received, then $s_i$ relies on it's own measurement probability, $P_{HPS}^{\tau_{\ell},s_i}(k)$, to remain in the  HPS  state. The corresponding updates to the state transition probabilities are shown in \textbf{Line 11, Alg.~\ref{alg:lps}}. If $\mathcal{N}^{s_i}_{HPS}\ne \emptyset$ (\textbf{Line 13, Alg.~\ref{alg:lps}}), i.e., information is received, then $s_i$ performs distributed node collaboration (DNC) (\textbf{Line 14, Alg.~\ref{alg:lps}}) to make an informed switching decision.

\vspace{6pt}
\subsubsection{DNC and Computation of the State Transition Probabilities} \label{sec:hps_stp} Node collaboration and computation of the state transition probabilities follow the same processes as described in the  LPS  state in Sections~\ref{sec:lps_dsc} and \ref{sec:lps_stp}, respectively. Full details are available in Section~\ref{sec:dsc}.

\section{Distributed Node Collaboration} \label{sec:dsc}
This section presents the details of the DNC algorithm. Let $\mathcal{N}_{RC}=\{s_j\in \mathcal{S}: \mathcal{N}_{HPS}^{s_j}\neq \emptyset\}$ be the set of all nodes that have received the target's state information from the  HPS  sensors in their neighborhood who are currently tracking the target. Then, if $s_i \in \mathcal{N}_{RC}$, then it runs the DNC algorithm.  The three steps of the DNC algorithm are described below.

\subsection{STEP 1: Distributed Fusion for Prediction of Target State (DUPS)}
The first step in DNC consists of fusing the received target state information to obtain a fused state estimate and then a one-step prediction. Since $s_i \in \mathcal{N}_{RC}$, it could be in the  LPS  or  HPS  state. If $s_i$ is in the  HPS  state, then DUPS improves its target state prediction, and if $s_i$ is in the  LPS  state, then DUPS enables state prediction without sensing. The information ensemble received by $s_i$ is

\begin{eqnarray}
\hat{\mathbf{I}}^{s_i}(k)=\Big\{\Big[\hat{\mathbf{x}}^{s_j},\hat{\mathbf{\Sigma}}^{s_j},\hat{\mathbf{W}}^{s_j}\Big](k),\forall s_j \in \mathcal{N}_{HPS}^{s_i}\Big\}
\end{eqnarray}
where $\hat{\mathbf{x}}^{s_j}(k|k)$, $\hat{\mathbf{\Sigma}}^{s_j}(k|k)$, and $\hat{\mathbf{W}}^{s_j}(k)$ are the target state, covariance, and filter gain estimates made by node $s_j$ at time $k$. This information ensemble is used to make target state prediction as follows.

\vspace{6pt}
\subsubsection{Trustworthy Set Formation} Due to false measurements from the HPS sensor, noise, and other factors, it is possible that the information received may contain false tracks, which requires the node to first validate the information to ensure that it is accurate and reliable before processing. False measurements associated to a target track may result in a movement that differs from the target motion model. This causes the covariance of the estimate to increase above the initialized value. This increase in estimation error provides the node with an indication of whether the track information is trustworthy. Therefore, this step aims to reduce false tracks by forming a set of trustworthy neighbors $\mathcal{N}_T^{s_i}\subseteq \mathcal{N}^{s_i}_{HPS}$ by evaluating the sum of the estimated position error as follows
\begin{equation}
\mathcal{N}_T^{s_i} = \Big\{s_j \in \mathcal{N}^{s_i}_{HPS}: Trace\big(\mathbf{H}(k)\hat{\mathbf{\Sigma}}^{s_j}(k|k)\mathbf{H}(k)'\big)\leq \xi\Big\}
\end{equation}
where $\mathbf{H}(k)$ is the Jacobian of the measurement model defined in Eq. (\ref{eq:tar_mod}); $\mathbf{H}(k)\hat{\mathbf{\Sigma}}^{s_j}(k|k)\mathbf{H}(k)'$ is the target's estimated position error, obtained from a subset of the covariance matrix associated to only the position state variables;  and $\xi$ is the maximum tolerance of the estimate. In this paper, $\xi = \frac{R_{1}^2\sigma_{\phi}^2+\sigma_R^2}{2}$, where $\sigma_{\phi}$ and $\sigma_R$ are the standard deviations in the azimuth and range measurements of the HPS sensor encompassed in the measurement noise w(k). This is chosen based on the initialized state position error such that if the estimated error increases above $\xi$ the track will be discarded. Thus, node $s_i$ accumulates the following trustworthy information ensemble:
\begin{equation}
\mathbf{\hat{I}}_T^{s_i}(k)=\Big\{\Big[\hat{\mathbf{x}}^{s_j}, \hat{\mathbf{\Sigma}}^{s_j}, \hat{\mathbf{W}}^{s_j}\Big](k), \forall s_j \in \mathcal{N}_{T}^{s_i}\Big\}
\end{equation}

\begin{figure*}[t!]
    \centering
    \begin{subfigure}[t]{0.45\textwidth}
        \centering
        \includegraphics[width=\textwidth]{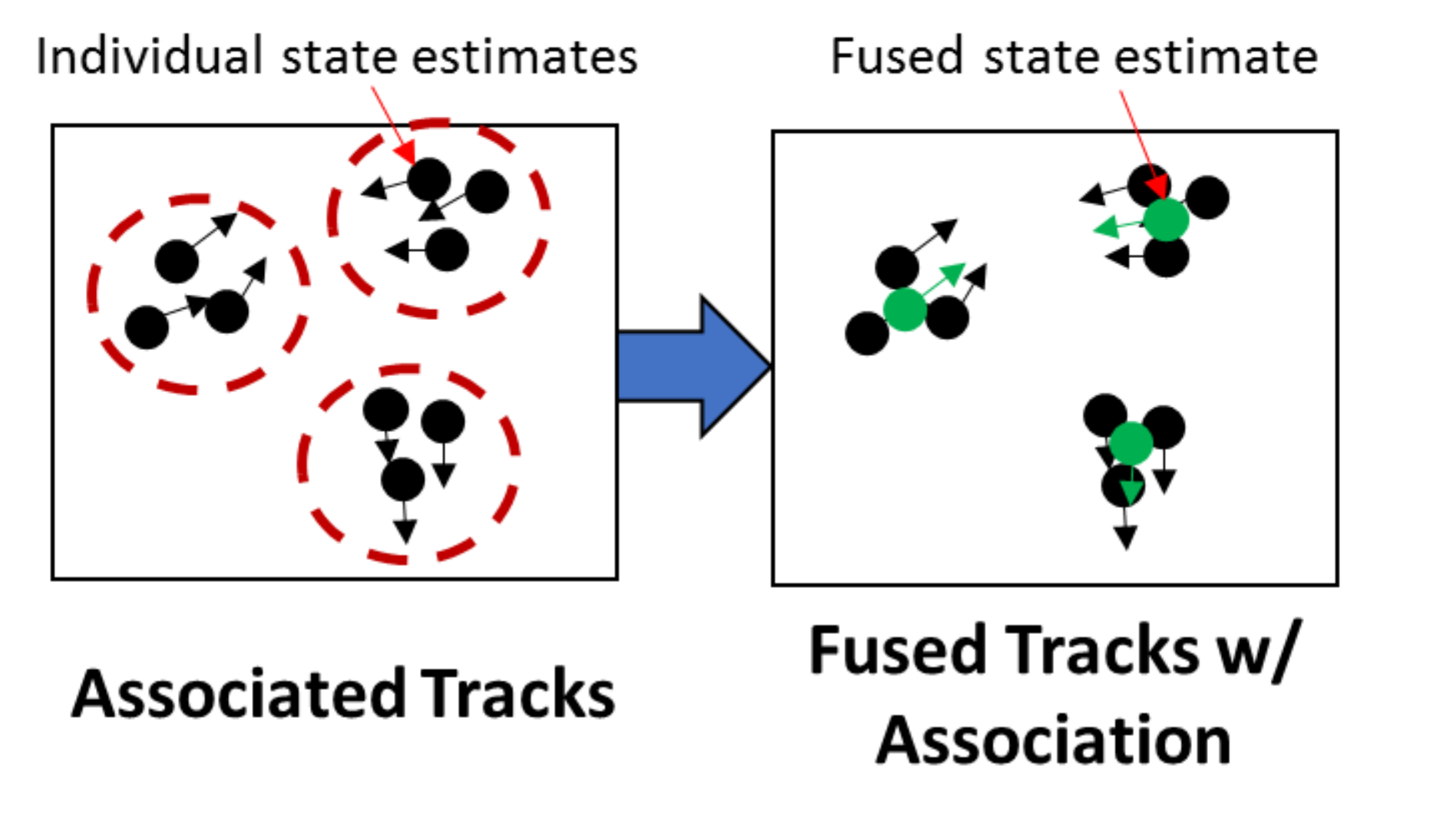}
        \caption{With T2TA}
        \label{fig:t2ta_t2tf}
    \end{subfigure}
    ~
    \begin{subfigure}[t]{0.45\textwidth}
        \centering
        \includegraphics[width=\textwidth]{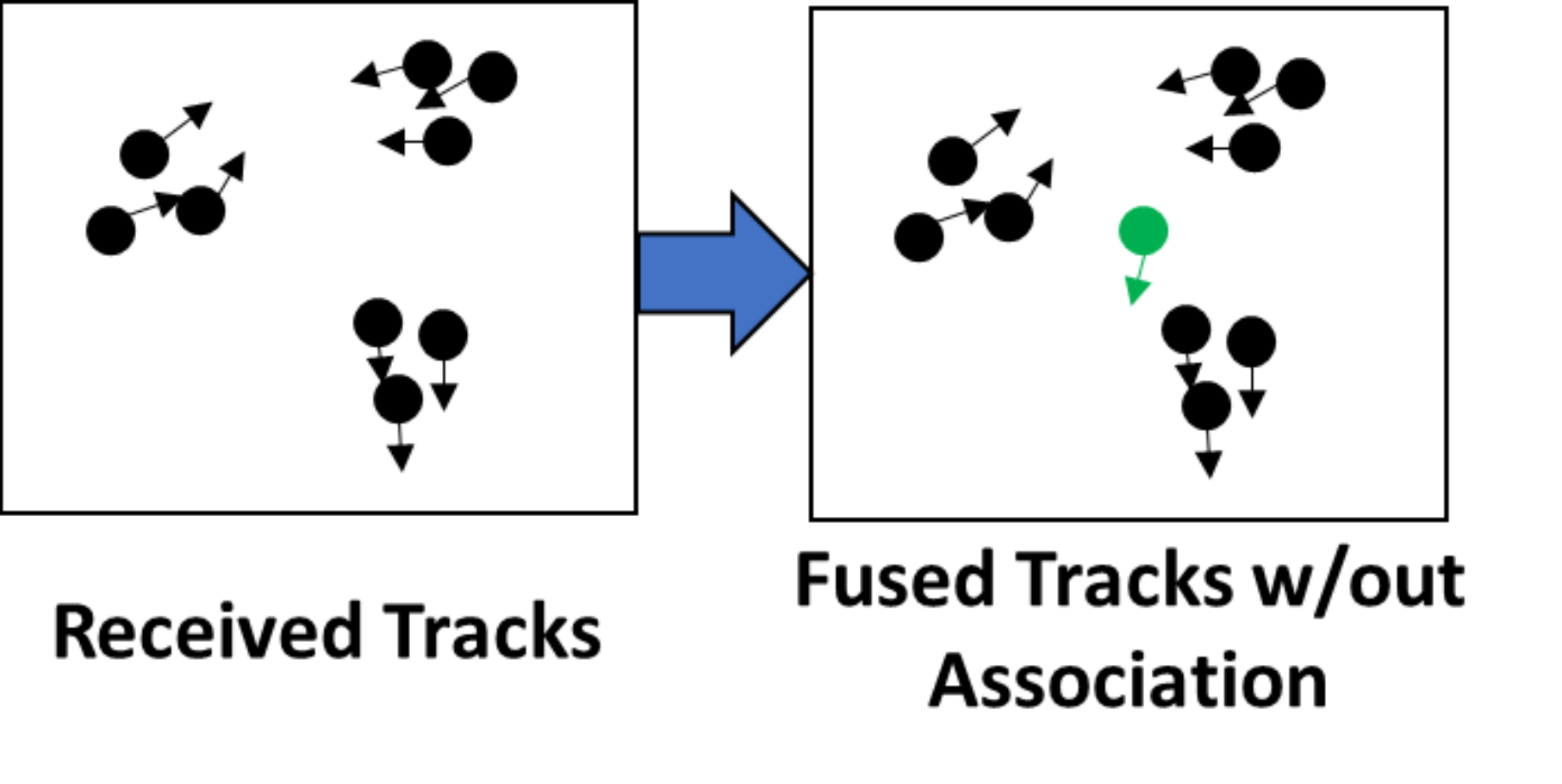}
        \caption{Without T2TA}
        \label{fig:t2tf}
    \end{subfigure} \vspace{-6pt}
    \caption{Example of the fused state estimates (a) with T2TA and (b) without T2TA. The balls and their arrows represent the position and velocity state variables, respectively. The black balls and their velocities are the different estimates received by the node at time $k$, while the green ones are the fused estimates. The association step is critical for nodes to identify the different targets located in their coverage area.} \label{fig:t2taf} \vspace{-9pt}
\end{figure*}

\vspace{6pt}
\subsubsection{Track-to-Track Association and Fusion} Next, the trustworthy information is associated to ensure that it is related to the same target to further improve fusion. In this work, the \textit{Track-to-Track Association Method} (T2TA) \cite{B1981} is used for this purpose. In this method, node $s_i$ associates the trustworthy information into $C$ different groups which correspond to the $C$ different targets that could be present within the node $s_i$'s neighborhood; thus forming the information ensembles:
\begin{equation}
\mathbf{\hat{I}}_T^{s_i,c}(k) \subseteq \mathbf{\hat{I}}_T^{s_i}(k), \ c = 1,...C.
\end{equation}
Subsequently, for each $c$, the state information in $\mathbf{\hat{I}}_T^{s_i,c}(k)$ is fused using the \textit{Track-to-Track Fusion} (T2TF) algorithm \cite{BWT2011}, to form a single state $\hat{\mathbf{x}}^{s_i,c}(k|k)$ and covariance $\hat{\mathbf{\Sigma}}^{s_i,c}(k|k)$ estimate. Fig.~\ref{fig:t2taf} shows an example of the advantage of association on the fused estimates.

\vspace{0pt}
\subsubsection{Target State Prediction} Once the fused estimates $\hat{\mathbf{x}}^{s_i,c}(k|k)$ and $\hat{\mathbf{\Sigma}}^{s_i,c}(k|k)$ are computed, node $s_i$ performs a one-step prediction using the Extended Kalman Filter to construct an estimate of the target states at time $k+1$ \cite{BDH2009}, as follows:
\begin{align}
&\mathbf{\hat{\mathbf{x}}}^{s_i,c}(k+1|k) =  \mathbf{f}(\hat{\mathbf{x}}^{s_i,c}(k|k),k) \nonumber \\
&\mathbf{\hat{\mathbf{\Sigma}}}^{s_i,c}(k+1|k)  =  \mathbf{F}(k)\mathbf{\hat{\mathbf{\Sigma}}}^{s_i,c}(k|k)\mathbf{F}(k)' + \mathbf{Q}   \label{predstate}
\end{align}
where $\mathbf{F}(k)$ is the Jacobian of the state transition matrix $\mathbf{f}(\cdot)$ evaluated at $\hat{\mathbf{x}}^{s_i,c}(k|k)$ and $\mathbf{Q}$ is the covariance matrix of the process noise $\boldsymbol{\upsilon}(k)$. The predicted state estimates are not communicated by the nodes; however, due to the fusion step, the predictions are the same for all neighbors.

Note: For simplicity, we drop the superscript $c$ in the remaining paper for all variables computed for each $c$. We will describe the content therein as necessary.

\vspace{0pt}
\subsection{STEP 2: Distributed Adaptive Node Selection (DANS)} \label{sec:ass}
After obtaining the target state prediction, the second step of DNC is distributed adaptive node selection for target tracking. Here, a node $s_i \in \mathcal{N}_{RC}$ determines if it belongs to the set of optimal nodes to track the target during the next time step. For this purpose, the predicted state of each target from Eq.~(\ref{predstate}) is used for selection of the optimal node set, $S^*(k+1)$, where $|S^*(k+1)|=N_{sel}$, with $N_{sel}>1$ to ensure robustness and to improve state estimate via distributed fusion and geometric diversity. Along with the optimal node selection, the sensing ranges of the selected nodes are optimized for maximizing coverage and minimizing energy consumption, to output $\mathcal{R}^*(k+1)=\left\{R_{HPS}^{s_j}(k+1)\in\{R_{1},...R_{L}\}, \forall s_j \in \mathcal{S}^*(k+1)\right\}$. As stated earlier, the base sensing range ($R_1$) is enough in high node density areas, while the extended sensing ranges ($>R_1$) are needed for resilience, i.e., to ensure target coverage in coverage gaps or low node density areas. Fig.~\ref{fig:dans} shows the flowchart of the DANS algorithm, whose details are in Sections~\ref{sec:ICN}-\ref{sec:CFD} below.

\vspace{6pt}
\subsubsection{Identification of Candidate Nodes}\label{sec:ICN}
To begin the process of DANS, node $s_i$ first uses the target's predicted state and covariance estimates, $\hat{\mathbf{x}}^{s_i}(k+1|k)$ and $\hat{\mathbf{\Sigma}}^{s_i}(k+1|k)$ from Eq.~(\ref{predstate}), to identify the set of candidate nodes that can completely cover the uncertainty region around the target's predicted position.

\begin{figure}[t]
\centering
        \includegraphics[width=\textwidth]{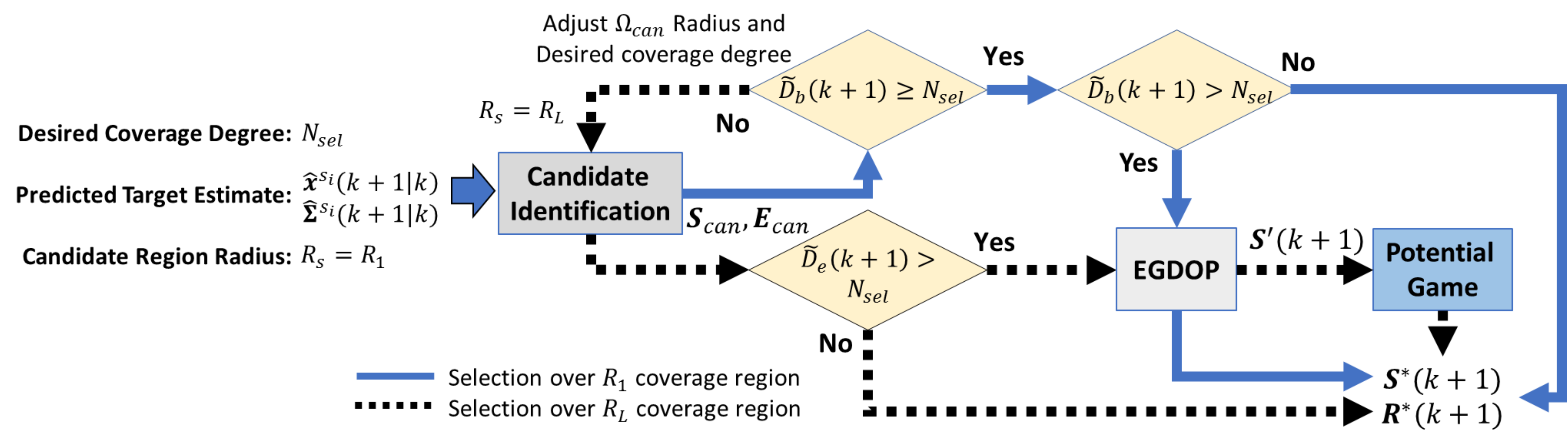}
        \caption{Flowchart of the DANS algorithm.} \label{fig:dans}
	\vspace{-12pt}
\end{figure}

Consider a sensing range parameter $R_{s}\in \{R_1,R_L\}$; by default $R_{s}=R_1$. Let $\Omega^{R_s}_{can}(k+1)\subset \Omega$, be the region such that any node lying within $\Omega^{R_s}_{can}(k+1)$ can cover the $6\sigma$ uncertainty region around the target's predicted position. Then, $\Omega^{R_s}_{can}(k+1)$ forms an elliptical region as follows
\begin{eqnarray}\label{canreg}
\left(\frac{x -\hat{x}^{s_i}(k+1|k)}{R_{s}-3\sigma_x}\right)^2+\left(\frac{y -\hat{y}^{s_i,}(k+1|k)}{R_{s}-3\sigma_y}\right)^2\le 1
\end{eqnarray}
where $\{\hat{x}^{s_i}(k+1|k),\hat{y}^{s_i}(k+1|k)\}$ is the predicted position estimate of the target; and $\sigma_x$ and $\sigma_y$ are the corresponding standard deviations of the uncertainty estimate. Note that $\Omega^{R_s}_{can}(k+1)$ lies inside a circle with center at $\hat{x}^{s_i}(k+1|k)$ and $\hat{y}^{s_i}(k+1|k)$ and radius $R_{s}$. The set of candidate nodes capable of tracking the target is defined as:
\begin{eqnarray}
\mathcal{S}^{R_s}_{can}(k+1) = \big\{s_j \in \mathcal{N}_{RC}: \mathbf{u}^{s_j}\in \Omega^{R_s}_{can}(k+1) \big\}
\end{eqnarray}
where the nodes that do not belong to $\mathcal{S}^{R_s}_{can}(k+1)$ are considered ineligible.

Next, if $s_i \in \mathcal{S}^{R_s}_{can}(k+1)$, then it broadcasts it's energy remaining, $E_{rem}^{s_i}(k)=1-\frac{E^{s_i}(k)}{E_0}$, to indicate that it is available for tracking, where $E_0$ is the node's initial energy and $E^{s_i}(k)$ is the total energy consumed, as defined in Section~\ref{sec:pf}. Similarly, $s_i$ receives the energy information from the other nodes in $\mathcal{S}^{R_s}_{can}(k+1)$ and forms the set of remaining energies of the candidate nodes
\begin{eqnarray}
E^{R_s}_{can}(k)=\big\{E_{rem}^{s_j}(k), \ \forall s_j \in \mathcal{S}^{R_s}_{can}(k+1)\big\}
\end{eqnarray}
which will be used for optimal node selection later. Note that the nodes in the sleep state do not transmit their energies; thus only the nodes in the LPS or HPS state are considered as candidates.

\vspace{6pt}
\subsubsection{Coverage Degree Identification} First, $s_i$ finds $\mathcal{S}^{R_1}_{can}(k+1)$. Then it determines the base coverage degree at time $k+1$ considering the uncertainty in the target's predicted position. This is defined as $\widetilde{D}_{b}(k+1)=|\mathcal{S}^{R_1}_{can}(k+1)|$. Following the flowchart in Fig.~\ref{fig:dans}, two situations can arise:

\begin{itemize}
\item \textit{Base coverage degree is sufficient} (i.e., $\widetilde{D}_b(k+1)\ge N_{sel}$): In this case, node $s_i$ can select a set of optimal nodes $\mathcal{S}^*(k+1)\subseteq \mathcal{S}^{R_1}_{can}(k+1)$ to track the target during the next time step, s.t. $|\mathcal{S}^*(k+1)|=N_{sel}$. Since $\Omega^{R_1}_{can}(k+1)$ lies within a circle of radius $R_1$, the optimal sensing ranges of sensors in $\mathcal{S}^*(k+1)$ can be simply chosen as $\mathcal{R}^*(k+1)=\{R^{s_j}_{HPS}(k+1)=R_1:\forall s_j\in\mathcal{S}^*(k+1)\}$.
Specifically, if $\widetilde{D}_b(k+1)=N_{sel}$, then $\mathcal{S}^*(k+1)=\mathcal{S}^{R_1}_{can}(k+1)$. On the other hand, if $\widetilde{D}_b(k+1)>N_{sel}$, then $\mathcal{S}^*(k+1)\subset \mathcal{S}^{R_1}_{can}(k+1)$ is obtained using the Energy-based Geometric Dilution of Precision (EGDOP), described in Section~\ref{sec:egdop}.
\vspace{3pt}
\item \textit{Base coverage degree is insufficient}  (i.e., $\widetilde{D}_b(k+1)<N_{sel}$): This implies that the target is located either in a low node density region (i.e., $0<\widetilde{D}_b(k+1)<N_{sel}$) or in a coverage gap (i.e., $\widetilde{D}_b(k+1)=0$). This scenario represents real world situations where the sensor deployment is biased (e.g. due to physical obstacles or air deployment). Additional, this can occur when a group of spatially co-located nodes fail (e.g. an attack on a particular sector of the network). In either case, in order to find sufficient nodes for tracking, node $s_i$ expands the candidate region to $\Omega^{R_L}_{can}(k+1)$ by setting the sensing range parameter $R_s = R_L$ in Eq.~(\ref{canreg}). This results in a larger candidate set $\mathcal{S}^{R_L}_{can}(k+1)$ that includes nodes that can detect the target with optimal sensing ranges chosen from the set $\{R_1,...R_L\}$. The extended coverage degree is then defined as $\widetilde{D}_e(k+1)=|\mathcal{S}^{R_L}_{can}(k+1)|$.

\begin{itemize}
\item if $\widetilde{D}_e(k+1)\leq N_{sel}$, then even after expansion to $\Omega^{R_L}_{can}(k+1)$, less than or equal to $N_{sel}$ nodes have been found. Thus, the optimal node set is obtained as  $\mathcal{S}^*(k+1)=\mathcal{S}^{R_L}_{can}(k+1)$.
\item If $\widetilde{D}_e(k+1)> N_{sel}$, then several new nodes have been added to the candidate pool. Thus, the following two steps are conducted: i) Filter a set of healthy nodes with high energies and that are geometrically diverse using the EGDOP measure (details are in Section~\ref{sec:egdop}), and (ii) select the optimal node set $\mathcal{S}^*(k+1)\subset \mathcal{S}^{R_L}_{can}(k+1)$ and their optimal range set $\mathcal{R}^*(k+1)$ using network potential games (details are in Section~\ref{sec:pg}).
\end{itemize}
\end{itemize}

\begin{figure*}
    \centering
    \includegraphics[width=0.45\textwidth]{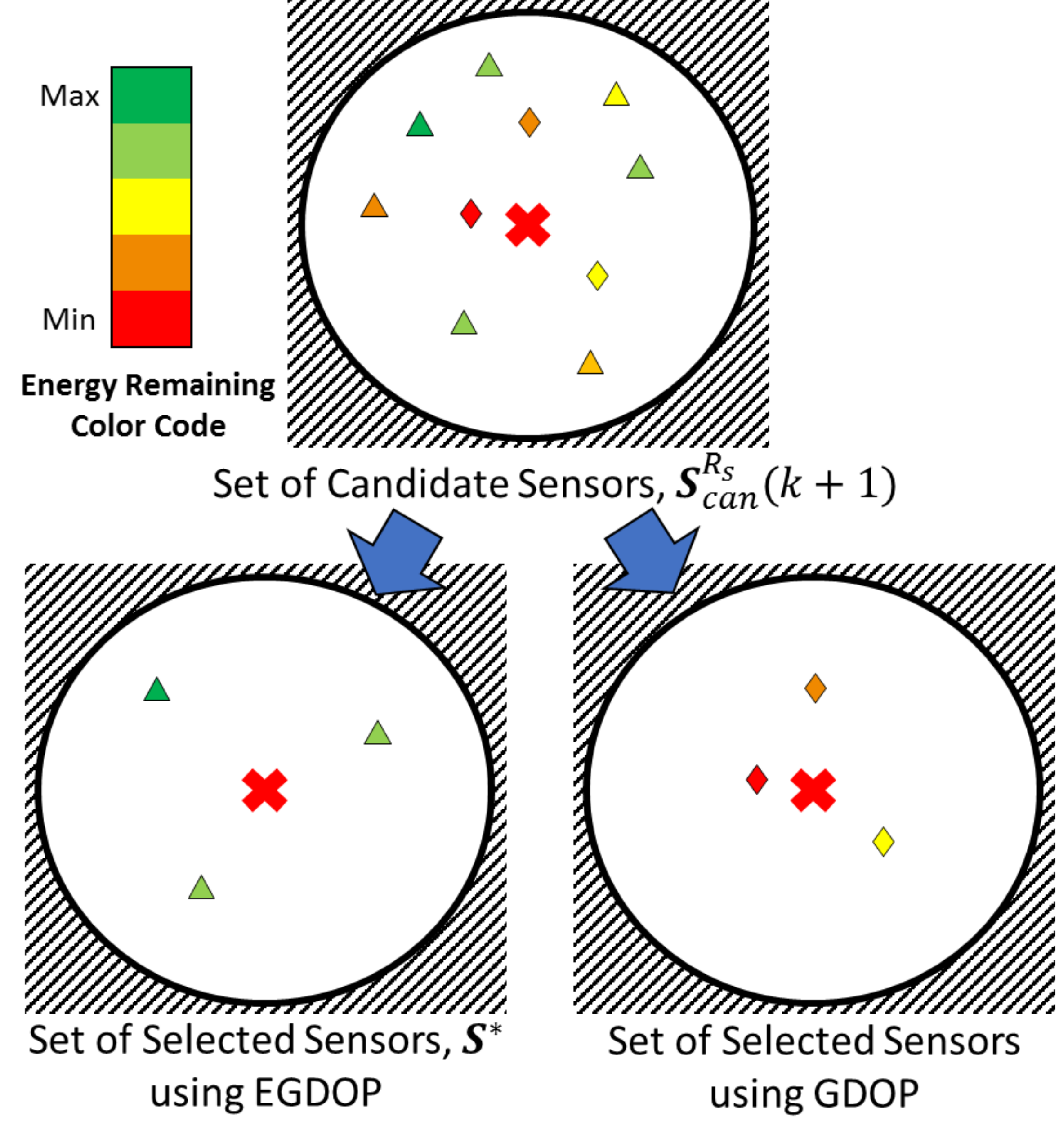} \vspace{-6pt}
    \caption{Example of sensors selected using Energy-based Geometric Dilution of Precision (EGDOP) and Geometric Dilution of Precision (GDOP). As seen, EGDOP selects geometrically diverse sensors with high energy, while GDOP selects nodes that are geometrically diverse and close to the target with low energy.}
    \label{fig:EGDOP}\vspace{-12pt}
\end{figure*}

\vspace{6pt}
\subsubsection{Energy-based Geometric Dilution of Precision (EGDOP)}\label{sec:egdop}
Typically, it is observed that the nodes with the largest energy remaining may not achieve the minimum mean squared estimation error due to their relative locations. In contrast, the nodes selected to minimize the mean squared estimation error may not maximize the energy remaining. Therefore, to jointly optimize these two criteria, this paper proposes a measure, called EGDOP, whose objective is to compute the optimal set of nodes that maximizes the energy remaining while minimizing the mean squared error of the target estimate. An example of this process is shown in Fig.~\ref{fig:EGDOP}. The nodes selected by EGDOP are geometrically distributed around the target's predicted position with high remaining energies. Thus, these nodes are reliable and produce accurate fused estimates. Formally, EGDOP is the {\emph{Geometric Dilution of Precision} (GDOP)} \cite{K2006} measure weighted by the remaining energy. This is computed as

\begin{eqnarray}
\mu(\widetilde{\mathcal{S}}) = \frac{det (\mathbf{J(\widetilde{\mathcal{S}})})}{trace(\mathbf{J(\widetilde{\mathcal{S}})})},
\end{eqnarray} \vspace{-6pt}
\begin{eqnarray}
\mathbf{J}(\widetilde{\mathcal{S}}) = \sum_{s_j\in \widetilde{\mathcal{S}}} \frac{E_R^{s_j}(k)}{\sigma_{\phi,n}^2 r_{s_j,n}^2} \left[ \begin{array}{cc} \sin^2(\phi_{s_j}) & -\sin(\phi_{s_j})\cos(\phi_{s_j}) \\ -\sin(\phi_{s_j})\cos(\phi_{s_j}) & \cos^2(\phi_{s_j}) \end{array} \right], \nonumber
\end{eqnarray}
where $\phi_{s_j}$ is the azimuth angle between sensor $s_j$ and the target's predicted position; $r^2_{s_j,n}=\left(\frac{x -\hat{x}^{s_i,c}(k+1|k)}{R_{s}-3\sigma_x}\right)^2+\left(\frac{y -\hat{y}^{s_i,c}(k+1|k)}{R_{s}-3\sigma_y}\right)^2$ is the normalized range of sensor $s_j$ to the target's predicted position; $\sigma_{\phi,n}=\frac{\sigma_{\phi}}{2\pi}$ is the  normalized measurement angle standard deviation; and $\widetilde{\mathcal{S}}\subseteq \mathcal{S}^{R_s}_{can}(k+1)$.

As described in the previous subsection, node $s_i$ runs the EGDOP algorithm under two conditions:
\begin{itemize}
\item [i)]  $\widetilde{D}_b(k+1)>N_{sel}$:

In this case, $\mathcal{S}^*(k+1)\subset \mathcal{S}^{R_1}_{can}(k+1)$. Then, the sets $\mathcal{S}^*(k+1)$  and $\mathcal{R}^*(k+1)$ are computed as
\begin{eqnarray}
\mathcal{S}^*(k+1) & = & \underset{\widetilde{\mathcal{S}}\subseteq \mathcal{S}^{R_1}_{can}(k+1)}{\arg\max} (\mu(\widetilde{\mathcal{S}})), \ \text{s.t.} \ |\widetilde{\mathcal{S}}|=N_{sel} \nonumber \\
\mathcal{R}^*(k+1) & = & \{R^{s_j}_{HPS}(k+1)=R_1:\forall s_j\in\mathcal{S}^*(k+1)\}
\end{eqnarray}
\item [ii)] $\widetilde{D}_b(k+1) < N_{sel}$ $\&$ $\widetilde{D}_e(k+1)>N_{sel}$:

In this case, $\mathcal{S}^*(k+1)\subset \mathcal{S}^{R_L}_{can}(k+1)$. However, in this case some nodes will lie at ranges greater than $R_1$, thus the node selection process should  optimize for the  HPS  sensing ranges of nodes to maximize coverage under uncertainty, as well as their energy remaining and geometric diversity. Since the EGDOP cost function does not account for range selection for maximizing target coverage, it alone cannot be used to identify $\mathcal{S}^*(k+1)$ and $\mathcal{R}^*(k+1)$. Also, the new candidate set of sensors $\mathcal{S}^{R_L}_{can}(k+1)$ could be very large, which can make the joint range selection computationally expensive to be performed in real time. Therefore, it is necessary to filter the candidate set $\mathcal{S}^{R_L}_{can}(k+1)$ to reduce complexity. Due to the above reasons, a two step node selection process is followed:
\begin{itemize}
\item  First, node $s_i$ uses the EGDOP cost function to identify a candidate set, $\mathcal{S}'(k+1)\subseteq \mathcal{S}^{R_L}_{can}(k+1)$, consisting of good (i.e., energetic and diverse) nodes, as follows
\begin{eqnarray}
\mathcal{S}'(k+1) = \underset{\widetilde{\mathcal{S}}\subseteq \mathcal{S}^{R_L}_{can}(k+1)}{\arg\max} (\mu(\widetilde{\mathcal{S}})); \ \text{s.t.} \ |\widetilde{\mathcal{S}}|=N'_{sel}>N_{sel}.
\end{eqnarray}
\item Subsequently, if $s_i\in \mathcal{S}'(k+1)$, then it utilizes a game-theoretic framework consisting of potential games (Section~\ref{sec:pg}), to jointly optimize for the sensing ranges of the candidate set. Whereas, if $s_i\notin\mathcal{S}'(k+1)$, sensor selection is complete and node $s_i$ computes its state transition probabilities described in Sections~\ref{sec:lps_stp} and \ref{sec:hps_stp}.
\end{itemize}
\end{itemize}

To validate the performance of the EGDOP metric, we computed the energy remaining and predicted covariance error of the target achieved using the EGDOP metric and compared them against the ones achieved by the classical GDOP and selection based on maximum energy remaining. For continuity of reading, these results are presented in Appendix~\ref{app:EGDOP}.

\vspace{6pt}
\subsubsection{Potential Games for Optimal Range Selection} \label{sec:pg}
After obtaining the candidate set $\mathcal{S}'(k+1)$ by filtering $\mathcal{S}^{R_L}_{can}(k+1)$ using EGDOP, the nodes in $\mathcal{S}'(k+1)$ must collaborate to jointly optimize their sensing ranges to a) maximize target coverage considering uncertainty in it's predicted state, and b) minimize total energy consumption in the extended sensing range. For this purpose, this paper develops a game-theoretic approach as described below.

A game $G$ in strategic form \cite{M13} is formulated  to consist of the following:
\begin{itemize}
\renewcommand{\labelitemi}{\textendash}
\item A finite set of players, $\mathcal{S}'(k+1)$.
\item A non-empty set of actions $\mathcal{A}_i$ associated to each player $s_i \in \mathcal{S}'(k+1)$. In this paper, each action $a_i \in \mathcal{A}_i$ indicates a different sensing range. Specifically, the action set $\mathcal{A}_i = \{0, R_{1},...R_L\}$, where action $0$ implies that the node is not selected to track the target during the next time step and will transition to either the  LPS  or $Sleep$ state. The action set is assumed to be identical for all players, i.e., $\mathcal{A}_i = \mathcal{A}_j$, $\forall s_i, s_j \in \mathcal{S}'(k+1)$.

\item The \textit{utility function} associated with each player $s_i$, defined as $\mathcal{U}_i: \mathcal{A}_{\mathcal{S}'(k+1)} \rightarrow \mathbb{R}$, where $\mathcal{A}_{\mathcal{S}'(k+1)} = \mathcal{A}_1 \times \ldots \times \mathcal{A}_{|\mathcal{S}'(k+1)|}$, denotes the set of joint actions for all players. The utility function computes the payoff that a node $s_i \in \mathcal{S}'(k+1)$ can expect by taking an action $a_i \in \mathcal{A}_i$, given that the rest of the players jointly select $a_{-i} \in \mathcal{A}_{-i}$, where $\mathcal{A}_{-i} := \mathcal{A}_1 \times \ldots \times \mathcal{A}_{i-1} \times \mathcal{A}_{i+1} \times \ldots \times \mathcal{A}_{|\mathcal{S}'(k+1)|}$. In this paper, the utility function is designed to jointly maximize target coverage and minimize the total predicted energy consumption.
\end{itemize}

A joint action of all players $a \in \mathcal{A}_{\mathcal{S}'(k+1)}$ is often written as $a = (a_i, a_{-i})$.

\begin{defn}[\textbf{Nash Equlibrium}]\label{defn:ne}
A joint action $a = (a_i^\star, a_{-i}^\star) \in \mathcal{A}_{\mathcal{S}'}$ is called a pure Nash Equilibrium if
\begin{align*}\label{eq:ne}
\mathcal{U}_i(a_i^\star, a_{-i}^\star) = \underset{a_i \in \mathcal{A}_i}{\max} \ \mathcal{U}_i(a_i, a_{-i}^\star), \  \forall s_i \in \mathcal{S}'(k+1)
\end{align*}
\end{defn}

Specifically, in this paper, the game-theoretic framework is built using Potential games~\cite{MAS2009}.

\begin{defn}[\textbf{Potential Game}]\label{defn:potentialgames}
A game $G$ in strategic form with action sets $\{\mathcal{A}_i\}_{i=1}^{|\mathcal{S}'(k+1)|}$ together with utility functions $\{\mathcal{U}_i\}_{i=1}^{|\mathcal{S}'(k+1)|}$ is a potential game if and only if, a potential function $\Phi: \mathcal{A}_{\mathcal{S}'(k+1)} \rightarrow \mathbb{R}$ exists, s.t. $\forall$ $s_i \in \mathcal{S}'(k+1)$
\begin{equation*}\label{eq:potentialgames}
\mathcal{U}_i(a_i', a_{-i}) - \mathcal{U}_i(a_i'', a_{-i}) = \Phi(a_i', a_{-i}) - \Phi(a_i'', a_{-i})
\end{equation*}
$\forall$ $a_i', a_i'' \in \mathcal{A}_i$ and $\forall$ $a_{-i} \in \mathcal{A}_{-i}$.
\end{defn}

A potential game requires the perfect alignment between the utility of an individual player and a globally shared objective function, called the \textit{potential function} $\Phi$, for all players. That is, the change in $\mathcal{U}_i$ by unilaterally deviating the action of player $s_i$ is equal to the amount of change in the potential function $\Phi$. In this regard, as the players negotiate towards maximizing their individual utilities, the global objective is also optimized.

The use of potential games has these advantages:
(i) at least one pure Nash Equilibrium is guaranteed to exist, which represents the optimal set of sensing ranges; (ii) there exist learning algorithms that can asymptotically converge to the optimal equilibrium with a fast convergence rate (e.g., the Max-Logit algorithm ~\cite{DHY2015}) to allow for real-time implementation; and (iii) the utility of each player is perfectly aligned with a global objective function, this implies that when the players negotiate to maximize their own utilities, the potential function is simultaneously maximized upon reaching the optimal equilibrium.

\vspace{6pt}
$\bullet$ \emph{\textbf{Leader Identification}}: Before the game is started, a node in $\mathcal{S}'(k+1)$ is identified as a group leader to compute the optimal sensing ranges for the whole group $\mathcal{S}'(k+1)$. This enables reduction of the communication overhead and energy consumption. The criteria for leader selection is the maximum available energy. Thus, the leader is selected as
\begin{eqnarray}
s_{Lead} = \argmax{s_j\in\mathcal{S}'(k+1)}\left(E_{rem}^{s_j}(k)\right).
\end{eqnarray}
If node $s_i=s_{Lead}$, then it continues to the next step, while if $s_i \ne s_{Lead}$, then it waits until $s_{Lead}$ computes the optimal ranges for $\mathcal{S}'(k+1)$ and transmits the result.

\vspace{6pt}

$\bullet$ \emph{\textbf{Partitioning of the Uncertainty Zone Around the Target's Predicted Position}}:
If $s_i=s_{Lead}$, then it partitions the uncertainty zone consisting of the $6\sigma$ confidence region around the target's predicted position at time $k+1$. Let $\Omega_{u}\subset \Omega$ be the rectangular area that contains the $6\sigma$ uncertainty zone of the target's predicted position, as shown in Fig. \ref{fig:tcpr}. Then,

\begin{eqnarray}
\Omega_{u} =  \left\{(x,y)\in \Omega ; -3\sigma_x \le ||x-\hat{x}^{s_i}(k+1|k)|| \le 3\sigma_x, -3\sigma_y \le ||y-\hat{y}^{s_i}(k+1|k)|| \le 3\sigma_y \right\}.
\end{eqnarray}

\begin{figure*}[t!]
	\centering
	\includegraphics[width=\textwidth]{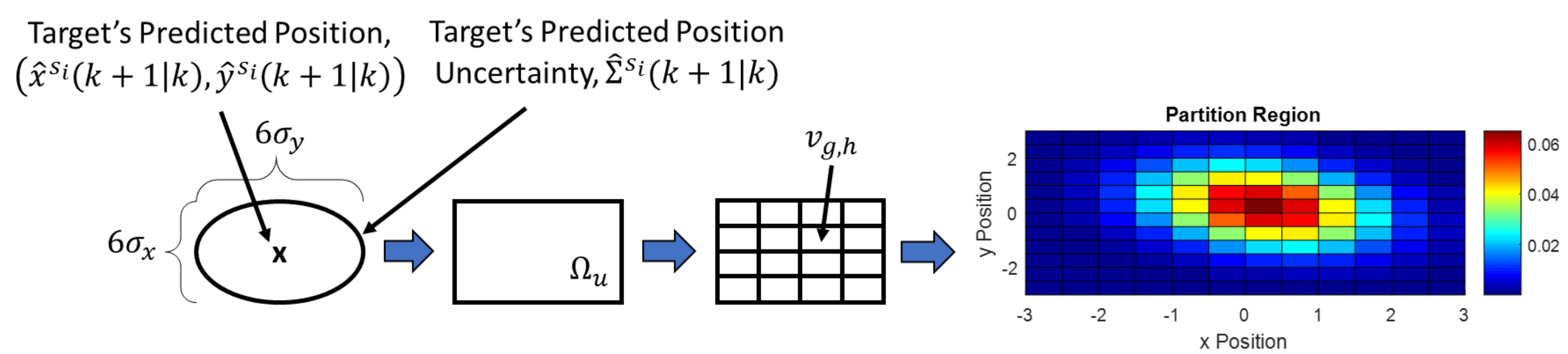}
	\caption{$6\sigma$ uncertainty region around the target's predicted position $\rightarrow$ rectangular region that covers the $6\sigma$ uncertainty region  $\rightarrow$ partition of the rectangular region $\rightarrow$ worth distribution over the partition.} \label{fig:tcpr}
	\vspace{-8pt}
\end{figure*}

\vspace{3pt}
Next, $\Omega_{u}$ is partitioned into $U\times V$ cells to form a grid, where each cell is denoted as $v_{g,h}$, $g=1,...,U$; $h=1,...,V$. Next, each cell $v_{g,h}$ is assigned a worth $\omega_{g,h}$ which represents the probability that the target is found in $v_{g,h}$, at time $k+1$. This is computed using the multivariate normal probability density function as follows:
\begin{eqnarray} \label{eq:weight}
\omega_{g,h}=\frac{1}{\Delta} \iint_{v_{g,h}} \mathcal{N}\left( \big[\hat{x}^{s_i}(k+1|k), \hat{y}^{s_i}(k+1|k)\big],\hat{\mathbf{\Sigma}}^{s_i}_z(k+1|k)\right)dxdy,
\end{eqnarray}
where $\Delta$ is a normalization constant s.t. $\sum_{g=1}^{U}\sum_{h=1}^{V} \omega_{g,h} =1$ and $\hat{\Sigma}^{s_i}_z(k+1|k) = \mathbf{H}(k)\hat{\Sigma}^{s_i}(k+1|k)\mathbf{H}(k)'$ is the target's predicted position uncertainty. In practice,~\eqref{eq:weight} is computed by numerically estimating the multivariate Gaussian cumulative density function \cite{G2004}.

\vspace{6pt}
$\bullet$ \emph{\textbf{Construction of the Potential Function}}: As stated earlier, the potential function must jointly maximize the overall coverage of the uncertainty zone around the target's predicted position, and minimize the predicted energy consumption. Thus, the potential function is designed as
\begin{eqnarray} \label{eq:potential_function}
\Phi(a)= \sum_{g=1}^{U} \sum_{h=1}^{V} \omega_{g,h} B_{g,h}\big(J_{g,h}(a)\big) - \frac{1}{N'_{sel} E_{c}(R_L)}\sum_{s_j\in \mathcal{S}'(k+1)} E_{c}(a_j),
\end{eqnarray}
where $J_{g,h}(a)$ is the number of nodes that can cover cell $v_{g,h}$ given the joint action $a$ of players; $B_{g,h}\big(J_{g,h}(a)\big)$ is the coverage function that depends on $J_{g,h}(a)$; $N'_{sel}=|\mathcal{S}'(k+1)|$; and $E_{c}(a_j)$ is the predicted energy consumption of sensor $s_j\in \mathcal{S}'(k+1)$ at time $k+1$, which is defined as
\begin{eqnarray} \label{eq:energy_game}
E_{c}(a_j)= \left\{ \begin{array}{ll}
e^{s_j}_{HPS}(a_j)\Delta T &  \text{if } a_j \ne 0 \\
e_{LPS}\Delta T & \text{if } a_j = 0 \end{array} \right.
\end{eqnarray}
Note that the potential function does not consider energy remaining because the players have been already selected with high energy remaining using EGDOP. Thus, the objective now is to select the sensing ranges of players to ensure coverage, while minimizing predicted energy consumption.

\vspace{6pt}
$\bullet$ \emph{\textbf{Details of Coverage Function Design:}} The coverage function $B_{g,h}\big(J_{g,h}(a)\big)$, $g=1,...U$, $h=1,...V$, is designed as a piece-wise linear function such that
\begin{eqnarray}\label{covfunction}
B_{g,h}\big(J_{g,h}(a)\big)=\left\{
\begin{array}{ll}
\Delta b_1J_{g,h}(a) & \text{if } J_{g,h}(a) \le N_{sel} \\
 \Delta b_1N_{sel} - \Delta b_2 (J_{g,h}(a)-N_{sel}) & \text{if } J_{g,h}(a) > N_{sel}
\end{array} \right.
\end{eqnarray}
where $\Delta b_1$ and $\Delta b_2$ are chosen to ensure that the game's equilibrium solution achieves an overall target coverage degree of $D(k+1)=N_{sel}$. In particular, $\Delta b_1$ is designed to incentivize the node to take action $a_i$ if $J_{g,h}(a)<=N_{sel}$ by increasing the potential function~\eqref{eq:potential_function}; while $\Delta b_2$ is chosen to ensure that the potential function decreases when $>N_{sel}$ nodes are covering the cells. An example of the coverage function $B_{g,h}\big(J_{g,h}(a)\big)$ is shown in Fig.~\ref{fig:cov_fun}, where for simplicity we chose a symmetric shape about $J_{g,h}(a) = N_{sel}$. Below, we present a theorem that allows the network designer to choose the slopes $\Delta b_1$ and $\Delta b_2$ to meet their specifications.

\begin{figure} [t]
\centering
\begin{tikzpicture}[
    scale = 0.75, declare function={
    func(\x)= (\x <= 3) * (0.5*\x)  + (\x>3)*(3-0.5*\x)
    ;
  }
]
\begin{axis}[
  axis x line=middle, axis y line=middle,
  x label style={at={(axis description cs:0.5,-0.1)},anchor=north},
  y label style={at={(axis description cs:-0.1,.5)},rotate=90,anchor=south},
  ymin=0, ymax=2.5, ytick={0,0.5,...,2.5}, ylabel=$B_{g,h}(J_{g,h})$,
  xmin=0, xmax=7, xtick={0,...,7}, xlabel=$J_{g,h}$,
  domain=0:7,samples=100
]

\addplot [blue,thick] {func(x)};
\addplot[domain=0:3] [red,thick,dashed] {(5/(5*3*60*0.0115))*x};
\addplot +[mark=none]  [black,thick,dotted] coordinates {(3, 0) (3, 2.5)};
\addplot[domain=3:6] [red,thick,dashed] {(5/(5*3*60*0.0115))*3 };

\addlegendentry{$\Delta b_1 = 0.5$ and $\Delta b_2 = \Delta b_1$}
\addlegendentry{$\Delta b_1$ and $\Delta b_2$ defined in (\ref{eq:b_bounds})}
\addlegendentry{$J_{g,h} = N_{sel}$}
\end{axis}
\end{tikzpicture}
\caption{Example of a Coverage function (\ref{covfunction}) with $\Delta b_1=0.5$ and $\Delta b_2 = \Delta b_1$ that satisfies the conditions (\ref{eq:b_bounds}) of Theorem \ref{thm:main} for $\delta = 0.035$, $\Delta R = 5m$, $N'_{sel}=5$, $N_{sel}=3$ and $R_L=60m$.} \label{fig:cov_fun}\vspace{-12pt}
\end{figure}
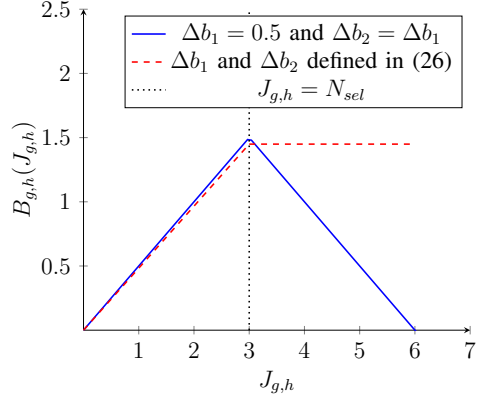

\begin{assum} \label{assum:uncertain_size}
The uncertainty in the target's predicted position is small enough, i.e., $R_L\ge \max(6\sigma_x,6\sigma_y)$, and there are sufficient available nodes, i.e., $|\mathcal{S}^{R_L}_{can}(k+1)|>N_{sel}$, such that there exists an action set $a^*$ that allows for at least $N_{sel}$ nodes to cover the entire uncertainty region.
\end{assum}

\vspace{6pt}
\begin{thm} \label{thm:main}
Given that Assumption~\ref{assum:uncertain_size} holds, the Nash equilibrium $a^*=(a_i^*,a_{-i}^*)$ achieves a coverage degree of $D(k+1) = N_{sel}$ with probability $Pr(D(k+1)=N_{sel}|a^*)\ge 1-\delta$, \ $0< \delta < 1$, if the slopes $\Delta b_1$ and $\Delta b_2$ of the coverage function $B_{g,h}\big(J_{g,h}(a)\big)$ in Eq. (\ref{covfunction}) satisfy the following
\begin{eqnarray} \label{eq:b_bounds}
\Delta b_1 &>& \frac{\Delta R}{N'_{sel} R_L \delta}, \nonumber \\
\Delta b_2 &>& 0,
\end{eqnarray}
where $\Delta R$ is the increment between any two consecutive sensing ranges.
\end{thm}
\begin{proof}
Please see Appendix~\ref{app:proof}.
\end{proof}

Note that Assumption~\ref{assum:uncertain_size} is only necessary to make a theoretical claim on the probability that the coverage degree will be equal to $N_{sel}$ when sufficient nodes are available. When Assumption~\ref{assum:uncertain_size} is violated, the algorithm will still work and use the available nodes for tracking the target.

To ensure that the game is a potential game, the utility function is designed based on the concept of \emph{Marginal Contribution} \cite{MW13}. Marginal contribution has each player compute their utility based on the amount of worth that the agent contributes to the group by selecting an action as opposed to selecting the null action. In this work, the null action is the sensing range $0$. Thus, the utility function is designed as follows,
\begin{eqnarray}\label{eq:utility}
\mathcal{U}(a_{i},a_{-i}) &= &\Phi(a_{i},a_{-i})-\Phi(\emptyset,a_{-i}) \nonumber \\
&= & \sum_{g=1}^{U} \sum_{h=1}^{V} \omega_{g,h} \left(B_{g,h}(J_{g,h}(a_{i},a_{-i})) - B_{g,h}(J_{g,h}(\emptyset,a_{-i})\right)) - \frac{E_c(a_i)-E_c(\emptyset)}{N'_{sel} E_c(R_L)}
\end{eqnarray}
where $\emptyset$ represents player $s_i$'s null action.

\begin{thm}\label{theorem:potentialgames}
The game $G$ with potential function $\Phi$ of Eq. (\ref{eq:potential_function}) and the utility function $\mathcal{U}_i$ of Eq. (\ref{eq:utility}) is a potential game.
\end{thm}
\begin{proof}
Given a joint action $a_{-i}$, the difference in $\Phi$ for sensor node $s_i \in \mathcal{S}'(k+1)$ to deviate its action from $a_i'$ to $a_i''$ is:
\begin{eqnarray}
\Phi(a_i', a_{-i}) - \Phi(a_i'', a_{-i}) &=& \left(\Phi(a_i', a_{-i}) - \Phi(\emptyset, a_{-i})\right) - \left(\Phi(a_i'', a_{-i}) - \Phi(\emptyset, a_{-i})\right) \nonumber \\
 &=&  \mathcal{U}_i(a_i', a_{-i}) - \mathcal{U}_i(a_i'', a_{-i})
\end{eqnarray}
Thus, game $G$ satisfies Defn. \ref{eq:potentialgames} and is a potential game.
\end{proof}

\subsubsection{Obtaining Game Equilibrium using Maxlogit Learning}\label{sec:CFD}
The leader $s_{Lead}$ identifies the optimal sensing ranges for all players in the game using the Maxlogit Learning algorithm~\cite{DHY2015}, which can converge fast to the optimal
equilibrium. The goal of the Maxlogit learning algorithm is to identify the Nash equilibrium of the potential function. Therefore, $s_{Lead}$ utilizes the utility function of Eq. (\ref{eq:utility}) in the Maxlogit learning algorithm to find the best joint action. The Maxlogit algorithm adopts a repeated learning framework where at each iteration $\kappa \in \mathbb{N}^+$, $s_{Lead}$ randomly selects one player $s_j \in \mathcal{S}'(k+1)$ and randomly selects a new action $\hat{a}_j(\kappa)$, while keeping the actions of remaining players, $a_{-j}(\kappa)$, the same. Then, $s_{Lead}$ computes the utility function $\mathcal{U}_j(\hat{a_j}(\kappa), a_{-j}(\kappa))$ and updates $s_j$'s action in a probabilistic manner~\cite{SWL11} as follows:
\[
 a_j(\kappa+1) = \begin{cases}
              \hat{a_j}(\kappa),& \text{with probability} \  \mu  \\
              a_j(\kappa),& \text{with probability} \  1 - \mu
            \end{cases}
\]
where $\mu = \frac{\psi(\hat{a_j})}{\max\{\psi(a_j), \psi(\hat{a_j})\}}$, $\psi(\hat{a_j}) = e^{\mathcal{U}_j(\hat{a_j}, a_{-j})/\tau}$, and $\tau > 0$.  The learning process stops when a predefined maximum number of learning steps are reached. Once the equilibrium $a^\star \in \mathcal{A}$ is reached, the joint action $\mathcal{R}^*(k+1)=a^\star$ is distributed to all the players.

To validate the performance of the potential games for optimal range selection, we compared the game efficiency against the optimal solution for various values of $N'_{sel}$. For continuity of reading, these results are presented in Appendix~\ref{app:Game}.

\begin{rem}
This paper assumes reliable communication as stated in Remark~\ref{rem:com}. However, if packets are dropped throughout the process, then the distributed sensor selection algorithm will still continue to operate, but the number of nodes selected to track the target using their HPS devices may vary from $N_{sel}$. If in the worst case scenario $>N_{sel}$ nodes are activated, then it will result in slightly more energy consumption. Furthermore, if target state estimates are dropped, then the fused states and target predictions may vary among the nodes and may result in an increased root mean squared error.  The effects of communication problems on the network will be studied in future work.
\end{rem}

\subsection{STEP 3: Distributed Computation of the Probability of Success of Target Detection (DOPS)}
\begin{wrapfigure}{r}{0.5\textwidth} \vspace{-15pt}
        \includegraphics[width=0.45\textwidth]{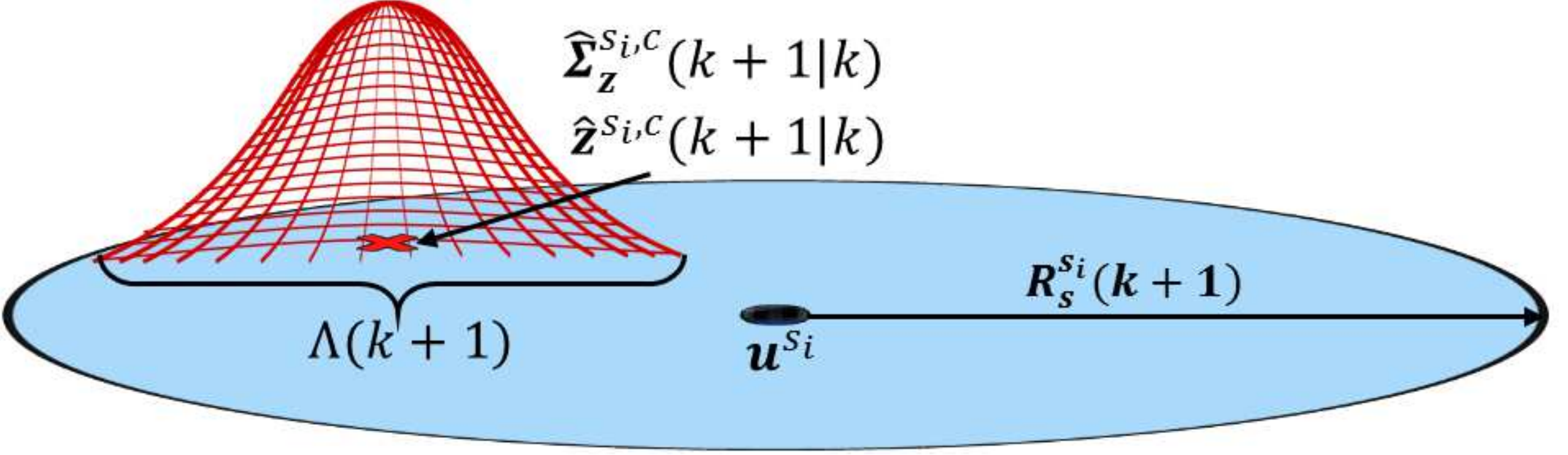}\vspace{-3pt}
        \caption{Computation of $\hat{P}_{HPS}^{s_i}(k+1)$.} \label{fig:Pd}
	\vspace{-15pt}
\end{wrapfigure}Finally, if $s_i \in \mathcal{S}^*(k+1)$ for any target track, then it should transition to the  HPS  state with a sensing range $R_{HPS}^{s_i}(k+1)\in \mathcal{R}^*(k+1)$ to track the target during the next time step. In order to make this transition, it computes its probability of success $\hat{P}_{HPS}^{s_i}(k+1)$ in detecting the target based on the target's predicted position, as shown in Fig.~\ref{fig:Pd}. Let
\vspace{-0pt}
\begin{eqnarray} \nonumber
\Lambda^{s_i,c}(k+1)=p_d\iint_G{N\Big(\hat{\mathbf{z}}^{s_i,c}(k+1|k),\mathbf{\hat{\mathbf{\Sigma}}}^{s_i,c}_{\mathbf{z}}(k+1|k)\Big)dxdy}
\end{eqnarray}
represent the scaled cumulative distribution function of the target $c$'s predicted position $\hat{\mathbf{z}}^{s_i,c}(k+1|k)$ over the coverage area of node $s_i$, where $p_d$ is the probability of detection of an HPS device and $G=\{(x,y):||(x,y)-\textbf{u}^{s_i}|| \le R_{HPS}^{s_i}(k+1)\}$. Then, the maximum probability of success of target detection over all tracks is given as
\begin{eqnarray} \label{eq:phps}
\hat{P}_{HPS}^{s_i}(k+1)  =  \max_c \left\{\Lambda^{s_i,c}(k+1)\right\},
\end{eqnarray}
which is used to transition to the  HPS  state as described in Section~\ref{sec:poser}.

\begin{rem} The optimal set $\mathcal{S}^*$ is chosen in a distributed manner and is unique if all nodes in $\mathcal{S}_{can}$ are connected. This is guaranteed when $R_c\geq 2R_L$. Note that $\mathcal{S}^*$ is computed for each track $c$.
\end{rem}

\begin{table}[b]
\centering
\caption{Simulation Parameters}
\label{tb:sim_params}\vspace{-6pt}
\begin{tabular}{ccccc}
\hline
$e_{clock}=0.01 W$  & $e_{LPS} = 115 mW$ &  $R_r = 15 m$  &   $\alpha = 0.95$  &  $\beta = 0.0036$ \\
$e_{HPS} = 0.2 \frac{W}{m}$ &  $e_{TX} = 1.26 W$ &  $R_{LPS} = 30 m$    &  $N_{sel}=3$ &  $N'_{sel} = 5$ \\
$e_{RX} = 0.63 W$ &  $e_{DPU} = 1 W$  &   $R_c = 120 m$   &  $\sigma_{\phi}= 0.25^{\circ}$ &  $\sigma_R = 0.075 m$ \\
$E_0 = 137592 J$ & $\Delta T = 0.5 s$ & $R_1= 30 m$   & $\sigma_{\boldsymbol{\upsilon}, x}=\sigma_{\boldsymbol{\upsilon},y} = 0.1 m$ & $\sigma_{\boldsymbol{\upsilon},\psi} = 0.1^{\circ}$  \\
$\mu_{cl}=0.025$ & $p_{fa} = 0.01$  &  $R_L= 60 m$   &  $\Delta R= 6 m$ &  $\chi=0.1$ \\
\hline
\end{tabular}
\end{table}

\section{Results and Discussion} \label{sec:results}

This section presents the results of the POSE.R algorithm in comparison with other methods to validate its effectiveness in providing resilient and efficient target tracking even in the presence of coverage gaps. First, we present the characteristics of the POSE.R network. For this purpose, the POSE.R algorithm was simulated in a $500m \times 500m$ deployment region generated in the Matlab environment. For validation, $500$ Monte Carlo runs were conducted, where the distribution of sensor nodes was regenerated in each run according to a uniform distribution (In Section~\ref{rescomp} we also intentionally created coverage gaps). This paper assumes that the nodes are deployed into an underwater environment where each heterogeneous sensor node has a hydrophone array \cite{lps_sense} as the LPS device and an active sonar \cite{hps_sense} as the HPS device. This work assumes that the amount of power applied to the active sonar device allows the node to adjust it's HPS sensing range. Table~\ref{tb:sim_params} lists the energy costs, sensing ranges, process noises ($\sigma_{\boldsymbol{\upsilon}, x}, \sigma_{\boldsymbol{\upsilon},y}, \sigma_{\boldsymbol{\upsilon},\psi}$), measurement noises ($\sigma_R, \sigma_{\phi}$), and sensor selection parameters.

\subsection{POSE.R Characteristics}
Fig.~\ref{fig:POSER_Char} presents the performance characteristics of the POSE.R algorithm in terms of: i) missed detection rates, ii) network lifetime, and iii) the number of active  HPS  nodes. The network density was varied as $\rho = [0.6,0.7,...,1.4]\times 10^{-3} (\frac{nodes}{m^2})$ and the probability of sleeping was varied as $p_{sleep}=[0, 0.25, 0.5, 0.75]$. For low network densities, coverage gaps could be created for fixed range sensing networks.

\vspace{10pt}
\subsubsection{Missed Detection Characteristics}  Figure~\ref{fig:POSER_PM} presents the probability of missed detection $P_m$ vs network density $\rho$ for various $p_{sleep}$ values. These characteristics indicate that the POSE.R algorithm achieves quite low missed detection rates using the DANS method even for less dense networks. This demonstrates resilience, i.e., the power of POSE.R algorithm in maintaining the detection capability for low density networks, which could result from sparse initial deployment or node failures.  Furthermore, it can be seen that as the value of $p_{sleep}$ increases, the missed detection probability increases as well. This is because as $p_{sleep}$ increases there is a higher probability that the nodes are sleeping around the target's position. Thus, there is a trade off between $p_{sleep}$ and $P_m$, especially for low densities.

\begin{figure*}[t!]
    \centering
    \begin{subfigure}[t]{0.32\textwidth}
        \centering
        \includegraphics[width=\textwidth]{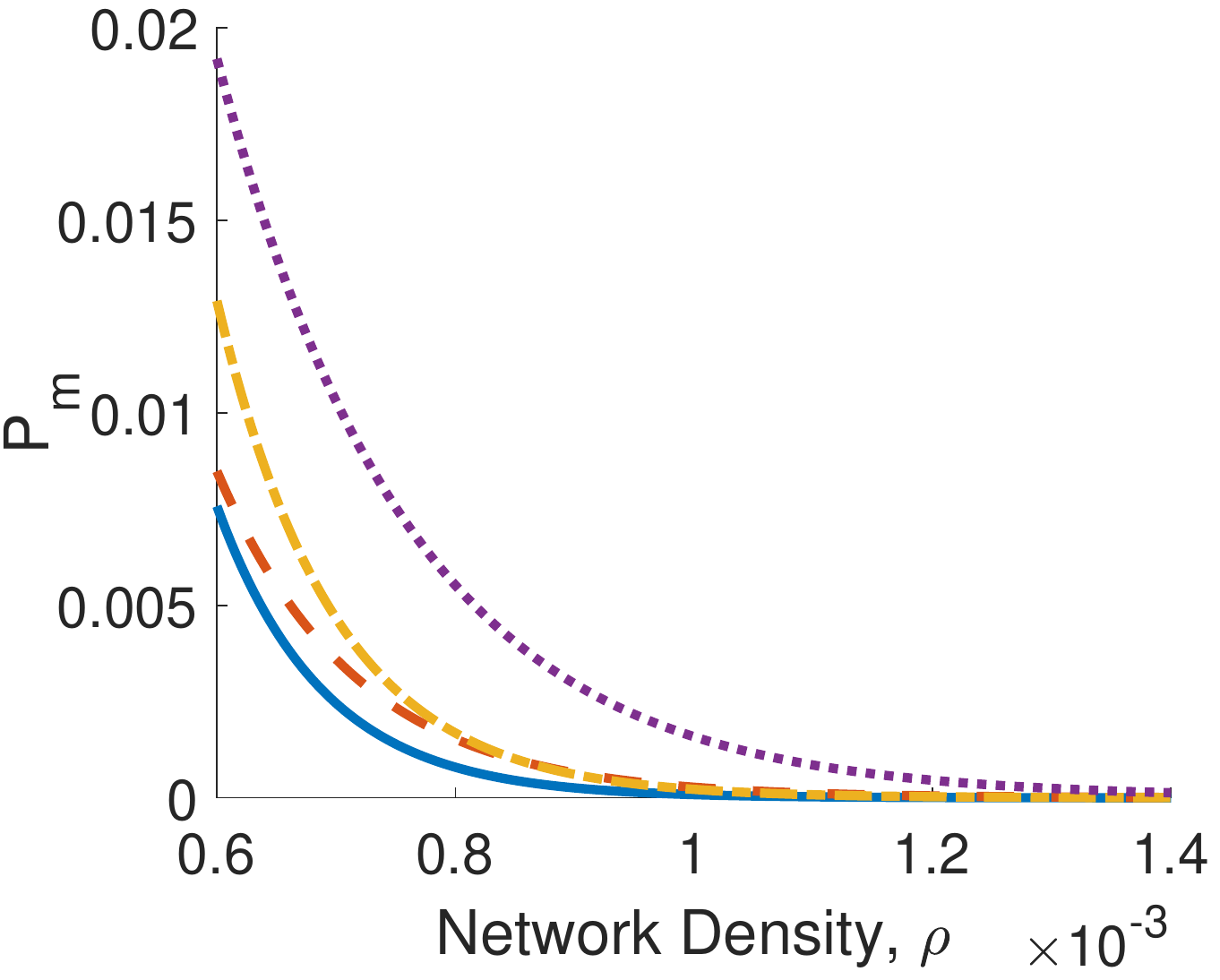}
        \caption{Missed Detection Rate}
        \label{fig:POSER_PM}
    \end{subfigure}
    ~
    \begin{subfigure}[t]{0.32\textwidth}
        \centering
        \includegraphics[width=\textwidth]{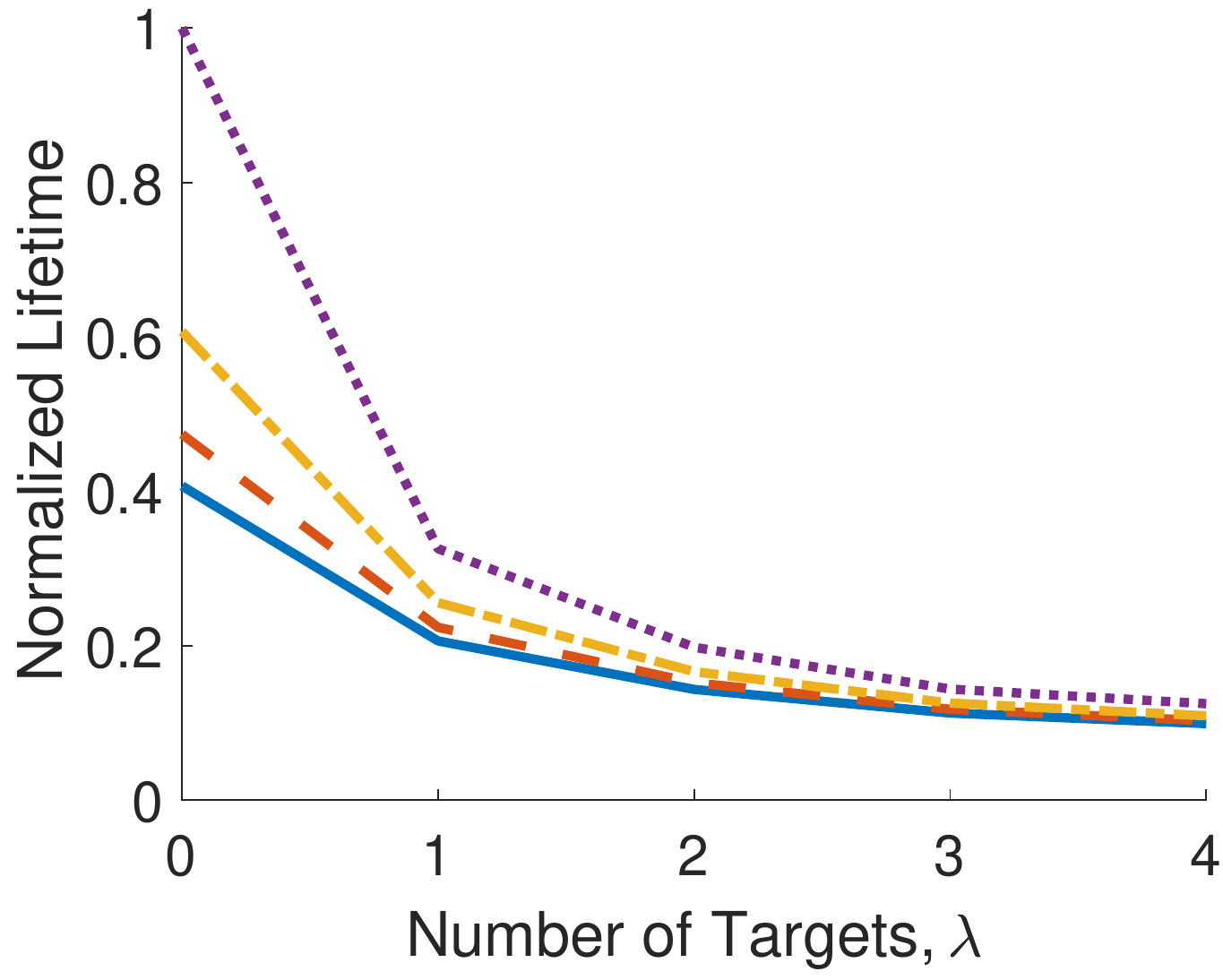}
        \caption{Normalized Network lifetime}
        \label{fig:POSER_Lifetime}
    \end{subfigure}%
    ~
    \begin{subfigure}[t]{0.32\textwidth}
        \centering
        \includegraphics[width=\textwidth]{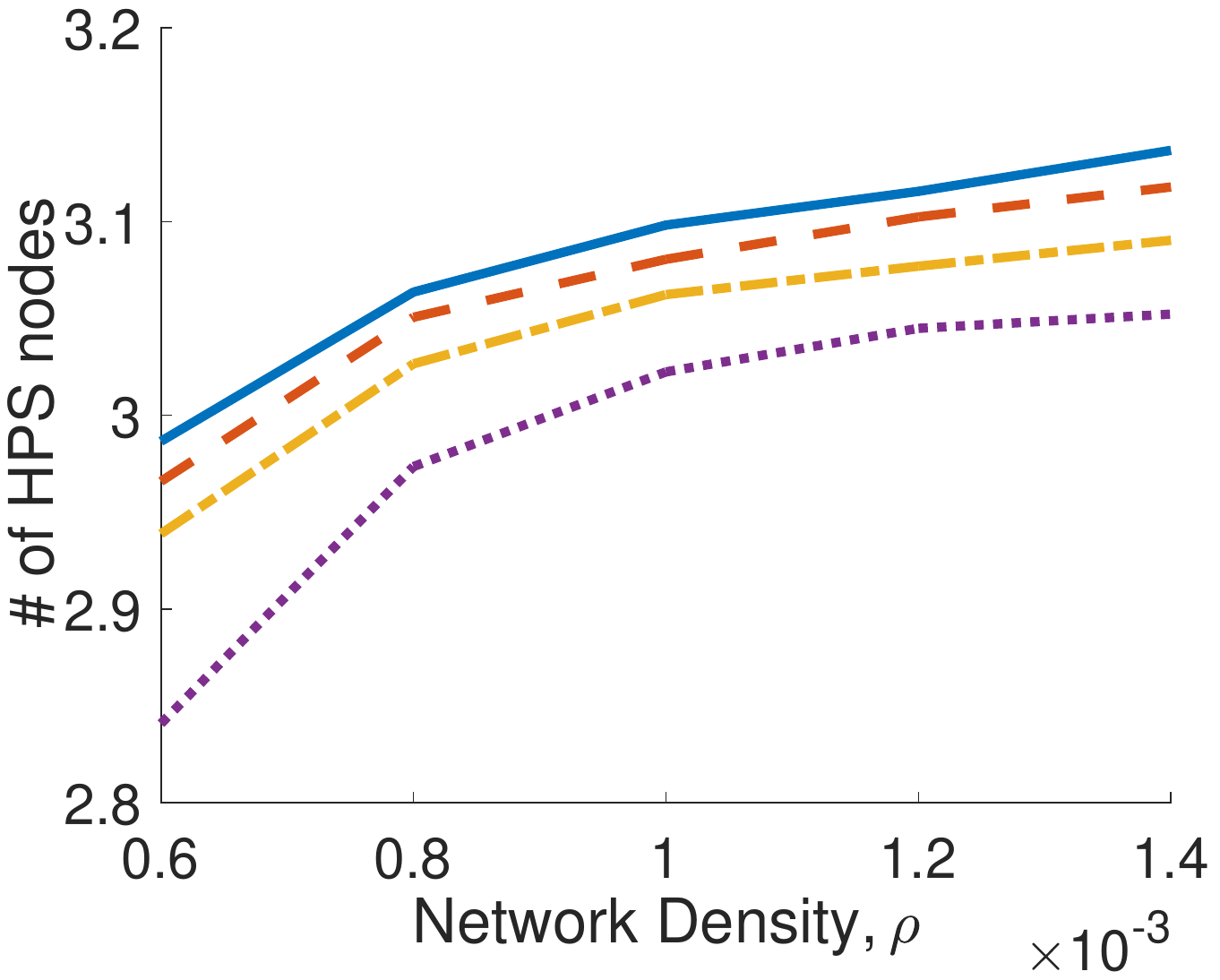}
        \caption{Average Number of  HPS  Nodes}
        \label{fig:POSER_NHPS}
    \end{subfigure}%

    \begin{subfigure}[t]{\textwidth}
    	\centering
        \includegraphics[width=.8\textwidth]{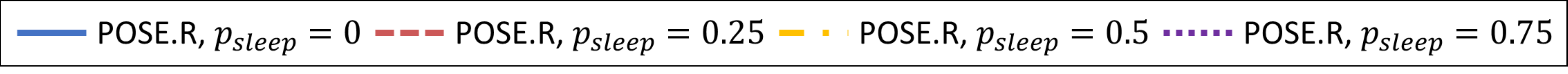}
    \end{subfigure}%

    \caption{POSE.R network characteristics.} \label{fig:POSER_Char} \vspace{-8pt}

\end{figure*}

\vspace{10pt}
\subsubsection{Network Lifetime Characteristics} \label{sec:poser_char_netlife} Fig.~\ref{fig:POSER_Lifetime} presents the network lifetime (Defn.~\ref{defn:nlife})  characteristics of the POSE.R network. The network lifetime is normalized with the lifetime of a network with no targets, i.e. $\lambda=0$, and for  $p_{sleep}$=0.75. For simulations, a tube of size $2R_{L} \times 600m$ was considered with targets traveling in a straight line through the center of the tube. The total lifetime of the network is computed when all of the nodes within $R_{LPS}$ of the targets' trajectories have zero remaining energy. The  number of targets $\lambda$, passing through the tube at a given time are varied between $\lambda=[0,4]$ and the network lifetime is computed for different values of $p_{sleep}$. The results indicate that as $\lambda$ increases the network lifetime decreases, because more nodes are needed to track more number of targets. Furthermore, it is seen that the effect of parameter $p_{sleep}$ is predominant for lower number of targets, that is higher $p_{sleep}$ results in higher network lifetime. However, as $\lambda$ increases, more number of nodes are triggered to track the target by the $DANS$ method, hence the effect of  $p_{sleep}$ diminishes.

\vspace{10pt}
\subsubsection{Number of Active  HPS  Nodes for Tracking a Target} Fig.~\ref{fig:POSER_NHPS} shows the average number of nodes activated in the  HPS  state to track a single target. The desired number was $N_{sel}=3$ during each time step. The results shows that for low density networks, i.e., $\rho < 0.8\times 10^{-3}$, the number of   HPS  nodes is slightly below $N_{sel}$. This is because for low density networks the number of available nodes within $R_L$ distance of the target could be less that $N_{sel}$. Furthermore, as the value of $p_{sleep}$ increases, the number of available nodes decreases; hence reducing the number of  HPS  nodes. For higher density networks, i.e., $\rho\ge 0.8\times 10^{-3}$, the number of  HPS  nodes is slightly larger than $N_{sel}$. This is due to the false alarm probability $p_{fa}$ causing nodes in the LPS state away from the target to transition to the  HPS  state. This effect is minimized for higher values of $p_{sleep}$.

\subsection{POSE.R vs. Existing Methods}
In this section we compare the performance of the POSE.R algorithm with existing scheduling methods. Specifically, POSE.R is compared against three distributed scheduling methods: (1) Autonomous Node Selection (ANS), (2) LPS-HPS Scheduling, and (3) Random Scheduling.
\begin{wrapfigure}{R}{0.45\textwidth}
        \includegraphics[width=.45\textwidth]{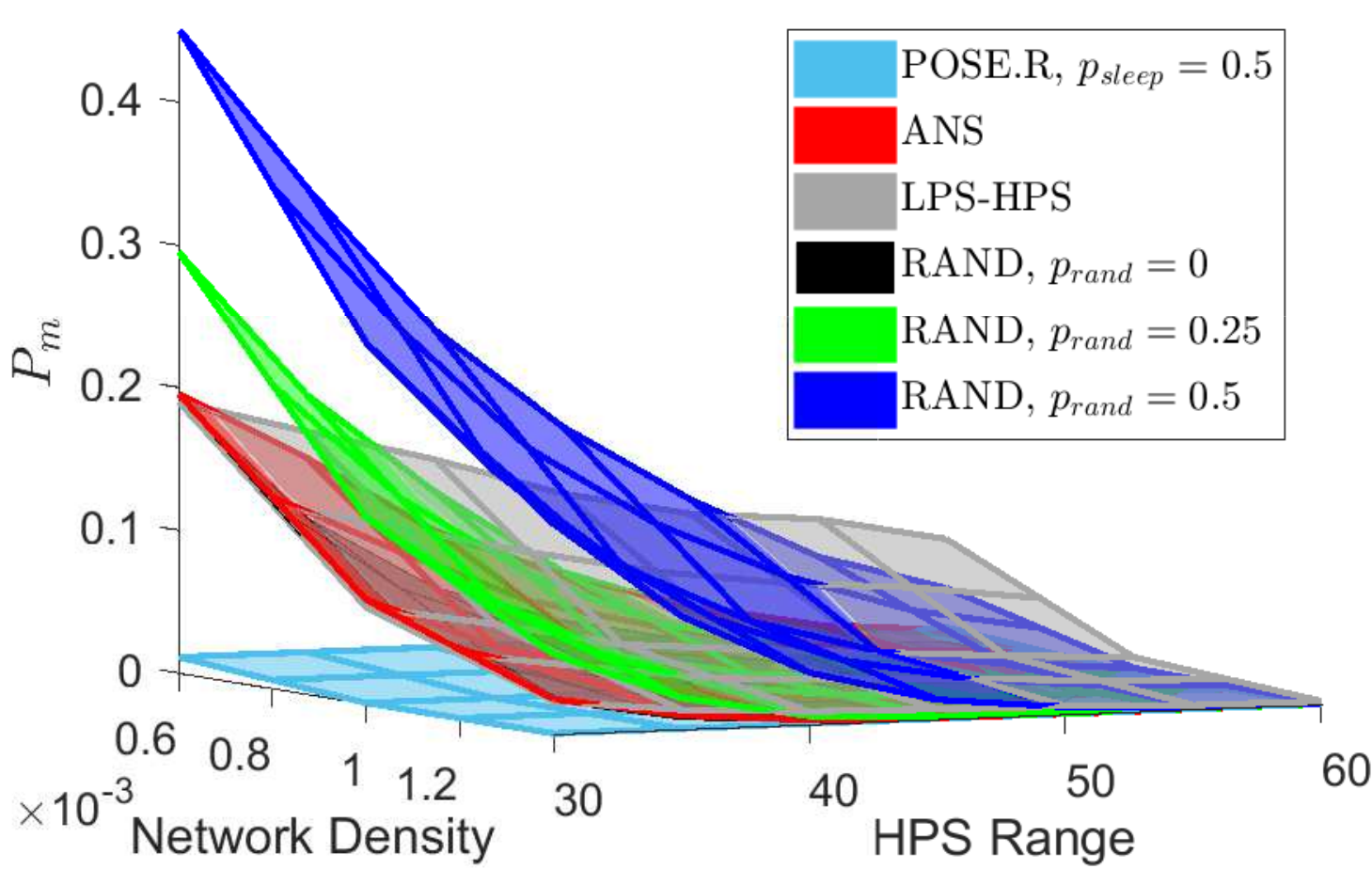}
        \caption{Probability of Missed Detection comparison with other methods.}
        \label{fig:3D_PM} \vspace{-0pt}
\end{wrapfigure}

The ANS method \cite{K2006} is a distributed node selection method that utilizes GDOP to select the optimal nodes to track the target. Here, the nodes collaborate in a distributed manner to make scheduling decisions. However, the ANS method considers passive sensors and does not include multi-modal sensor nodes. Therefore, to ensure an apple to apple comparison, the ANS method is adapted to include multi-modal operating conditions, where the selected nodes  track the target in active (HPS) state, while the others stay in the passive (LPS) state with their receivers on. As compared to ANS, the POSE.R algorithm incurs additional energy cost only when the target travels into a low density region. In this case, the nodes in a POSE.R network expand their sensing ranges to ensure coverage. Therefore, the energy costs of transmission and HPS increases. However, since POSE.R algorithm includes a Sleep state, the total energy cost is reduced overall.

The LPS-HPS Scheduling method is a distributed trigger-based activation method that utilizes two operating states, passive (LPS) and active (HPS). The nodes remain passive until a target is detected. Once a target is detected, the node remains active until the target passes out of the node's detection range.

 The Random Scheduling method is a distributed probabilistic method where the nodes randomly switch between sleeping and actively sensing (HPS). During each time step, a node sleeps with a probability $p_{rand}$ and senses with a probability $1-p_{rand}$. Thus, for $p_{rand}=1$ the network is always sleeping, while for $p_{rand}=0$ the network is always sensing. Note that the LPS-HPS and Random Scheduling methods do not facilitate node collaboration.

 As compared to the LPS-HPS and Random Scheduling methods, the only additional energy cost in POSE.R is the cost of exchanging messages. However, this is compensated by the significant energy savings of the POSE.R algorithm using the sleep state and through efficient node scheduling. The additional complexity in POSE.R arises in the state association/fusion, sensor selection, and target prediction steps. However, association/fusion along with sensor selection improve the accuracy of the state estimate. Furthermore, the sensor selection and target prediction steps minimize the number of sensors active around the target, which reduces the overall energy consumption.

Additionally, we compared the performance of the POSE.R algorithm with our prior work, POSE and POSE.3C. However, this section strictly focuses on the results comparing POSE.R with the above methods. The comparison of POSE, POSE.3C, and POSE.R are presented in Appendix~\ref{app:net_comp}.

\vspace{10pt}
\subsubsection{Missed Detection Comparison} Fig.~\ref{fig:3D_PM} shows the comparison of the missed detection characteristics of the POSE.R algorithm with the other distributed scheduling methods. While POSE.R assumes adaptive sensing range, each of the other scheduling methods were simulated with a fixed  HPS  sensing range chosen from $\{R_1,... R_L\}$. As seen in Fig.~\ref{fig:3D_PM}, the POSE.R algorithm achieves a significantly lower missed detection rate than the other methods for low network density, thus demonstrating resilience. The missed detection probability $P_m$ of the other distributed methods approach that of the POSE.R algorithm only for high network density and large  HPS  sensing ranges. Therefore, in order for the other methods to achieve similar characteristics as POSE.R, the network must contain a high density of sensor nodes that are utilizing a large  HPS  sensing range. In other words, the missed detection performance of the POSE.R network supersedes all other networks.

\begin{figure*}[t!]
    \centering
    \begin{subfigure}[t]{0.48\textwidth}
        \centering
        \includegraphics[width=\textwidth]{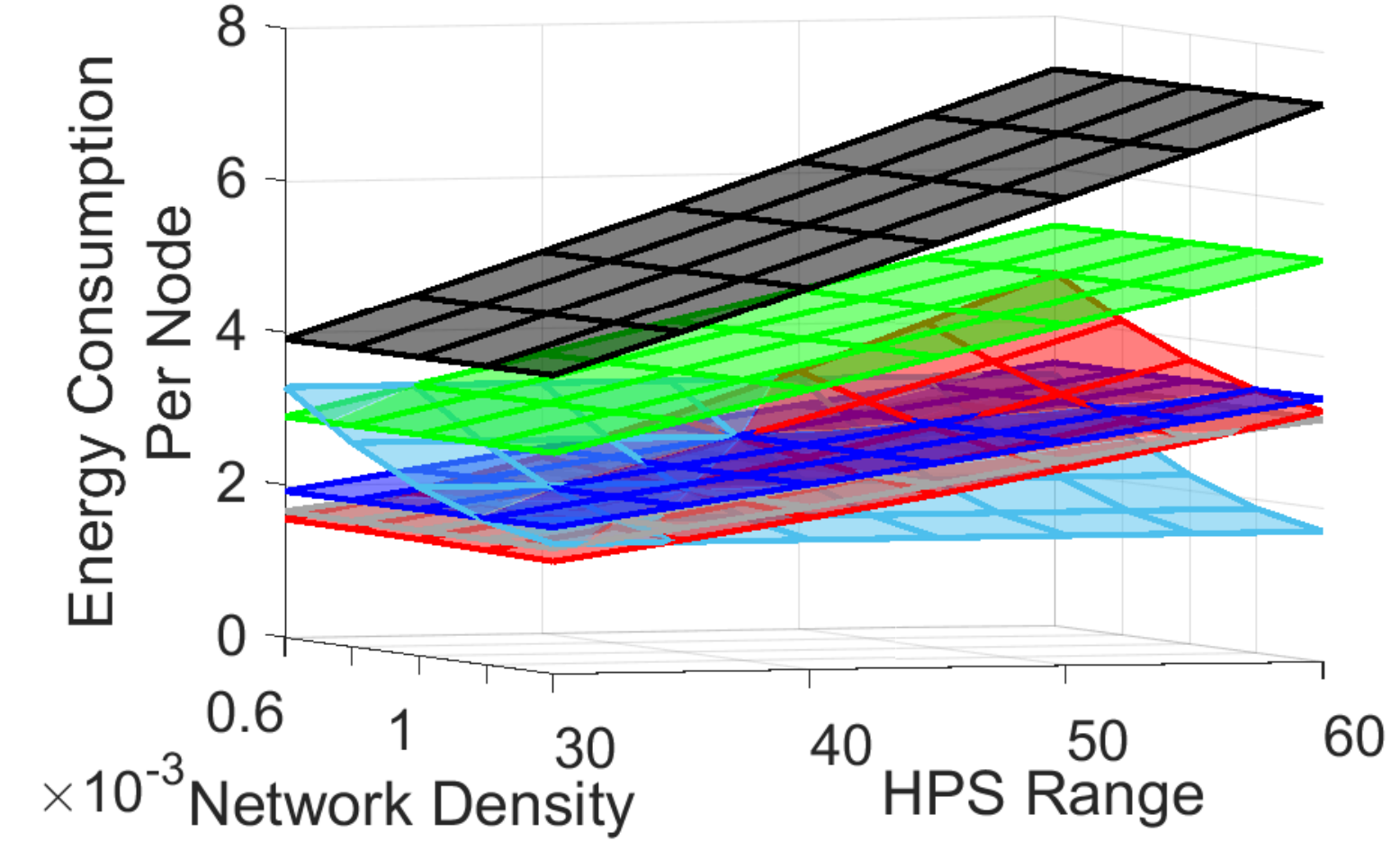}
        \caption{Average energy consumed per node located within $R_L$ of the target's position.}
        \label{fig:3D_AVG_E_IN}
    \end{subfigure} \hspace{6pt}%
    \begin{subfigure}[t]{0.48\textwidth}
        \centering
        \includegraphics[width=\textwidth]{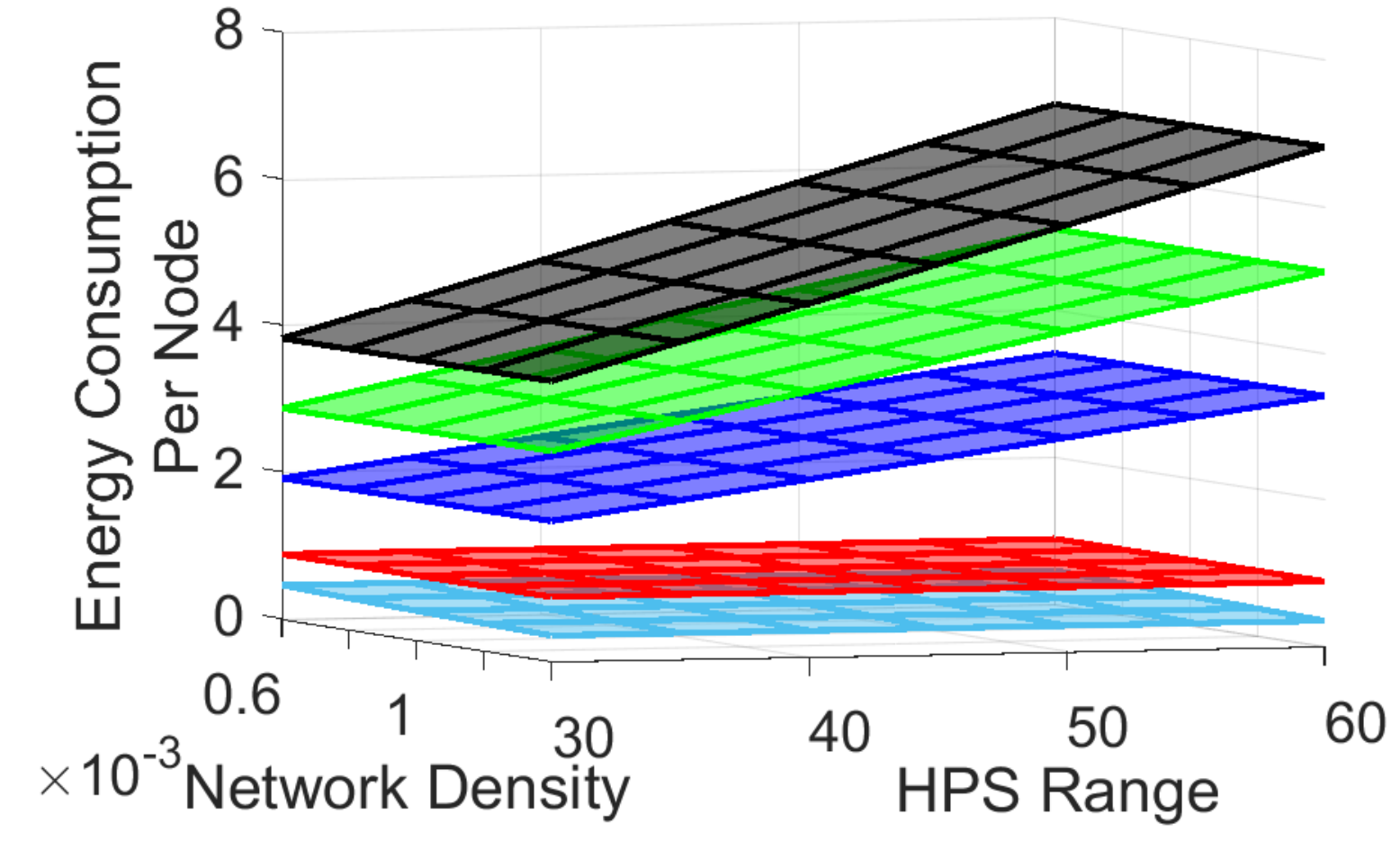}
        \caption{Average energy consumed per node located outside $R_L$ of the target's position.}
        \label{fig:3D_AVG_E_OUT}
    \end{subfigure} \vspace{6pt}

	\begin{subfigure}[t]{0.48\textwidth}
        \centering
        \includegraphics[width=\textwidth]{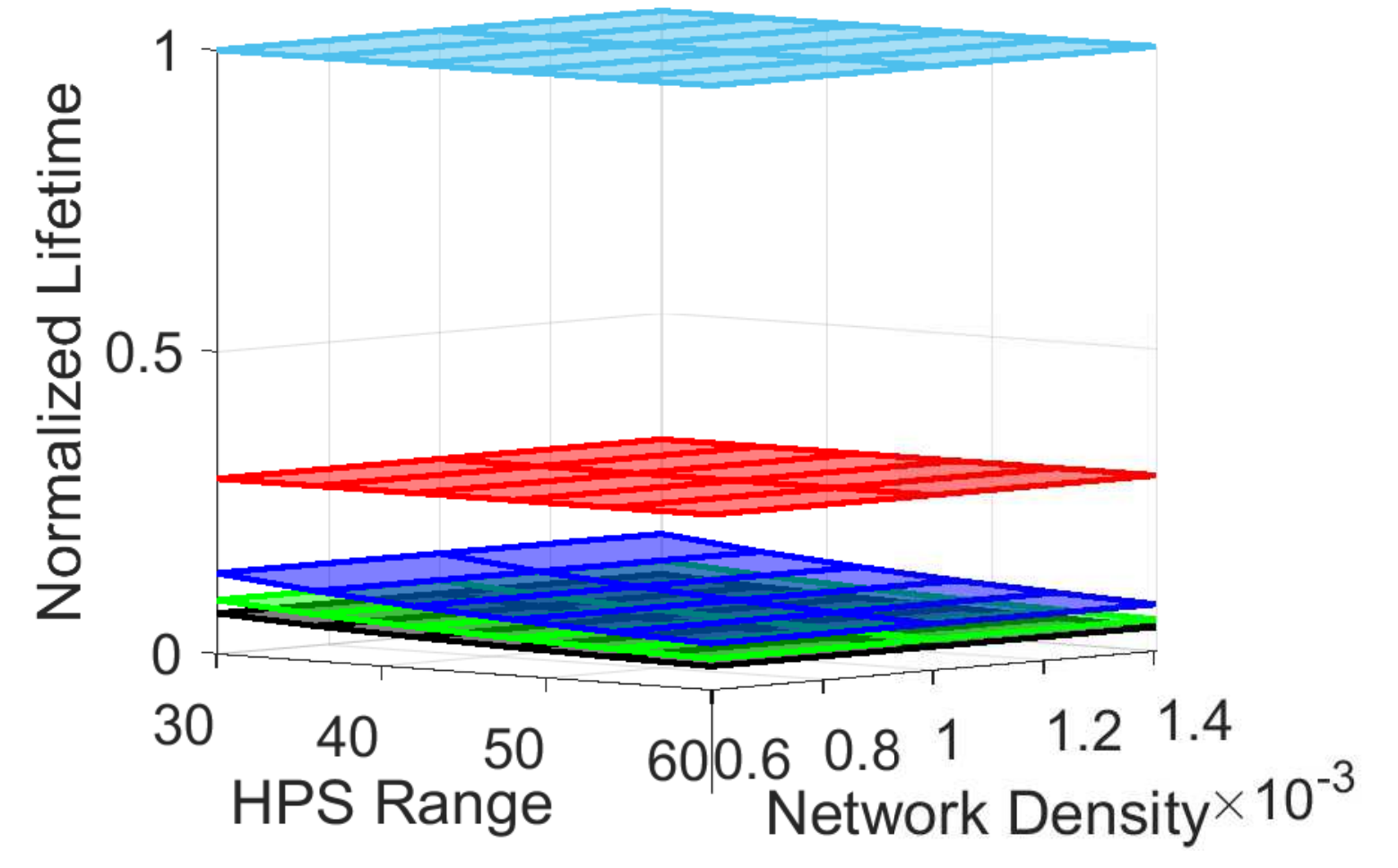}
        \caption{Network Lifetime with $\lambda=0$ targets in $\Omega_\gamma$}
        \label{fig:3D_LIFE_0}
    \end{subfigure}\hspace{6pt}%
    \begin{subfigure}[t]{0.48\textwidth}
        \centering
        \includegraphics[width=\textwidth]{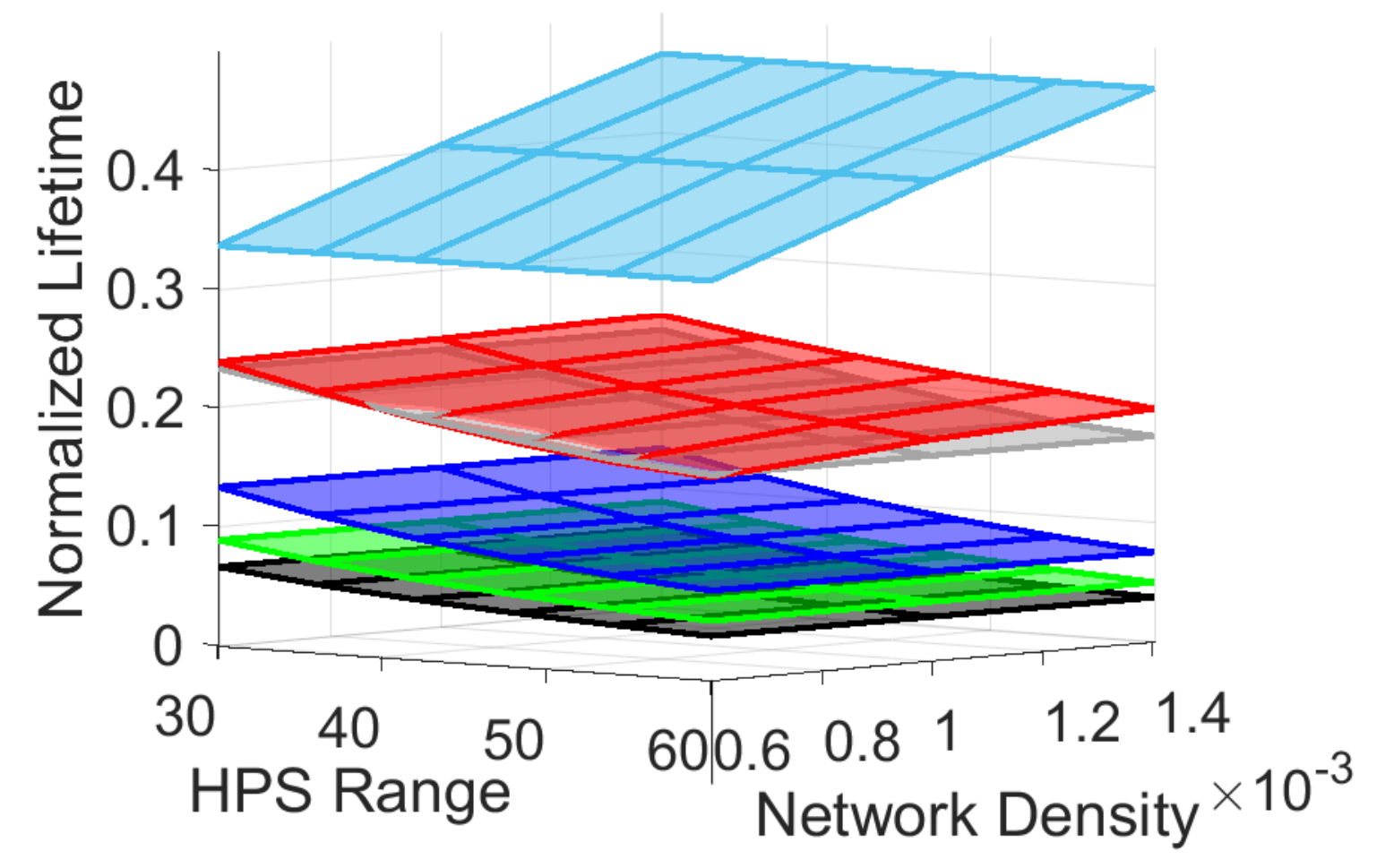}
        \caption{Network Lifetime with $\lambda=1$ targets in $\Omega_\gamma$}
        \label{fig:3D_LIFE_1}
    \end{subfigure}\vspace{6pt}%

	\begin{subfigure}[t]{0.48\textwidth}
        \centering
        \includegraphics[width=\textwidth]{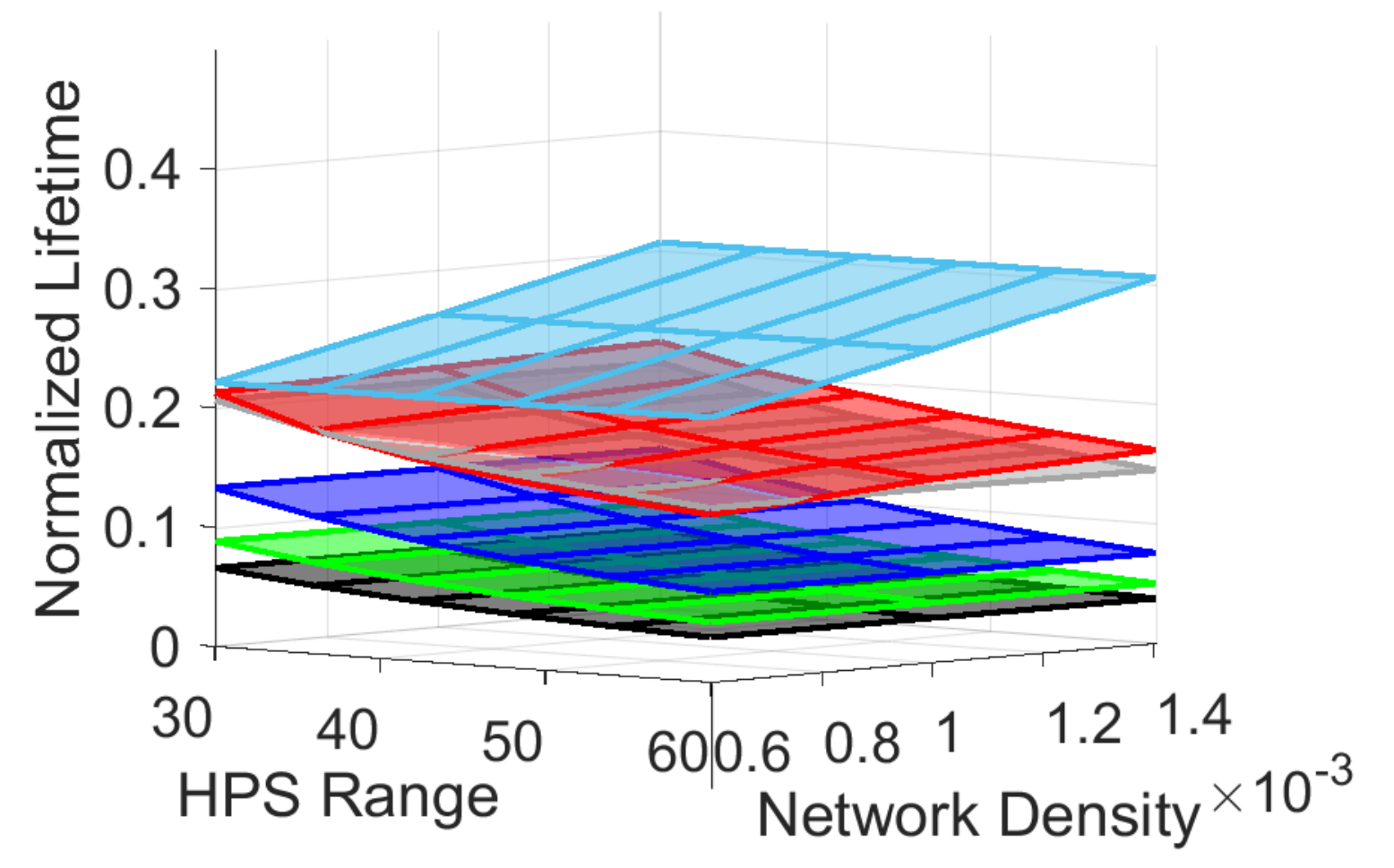}
        \caption{Network Lifetime with $\lambda=2$ targets in $\Omega_\gamma$}
        \label{fig:3D_LIFE_2}
    \end{subfigure}\hspace{12pt}%
    \begin{subfigure}[t]{0.48\textwidth}
        \centering
        \includegraphics[width=\textwidth]{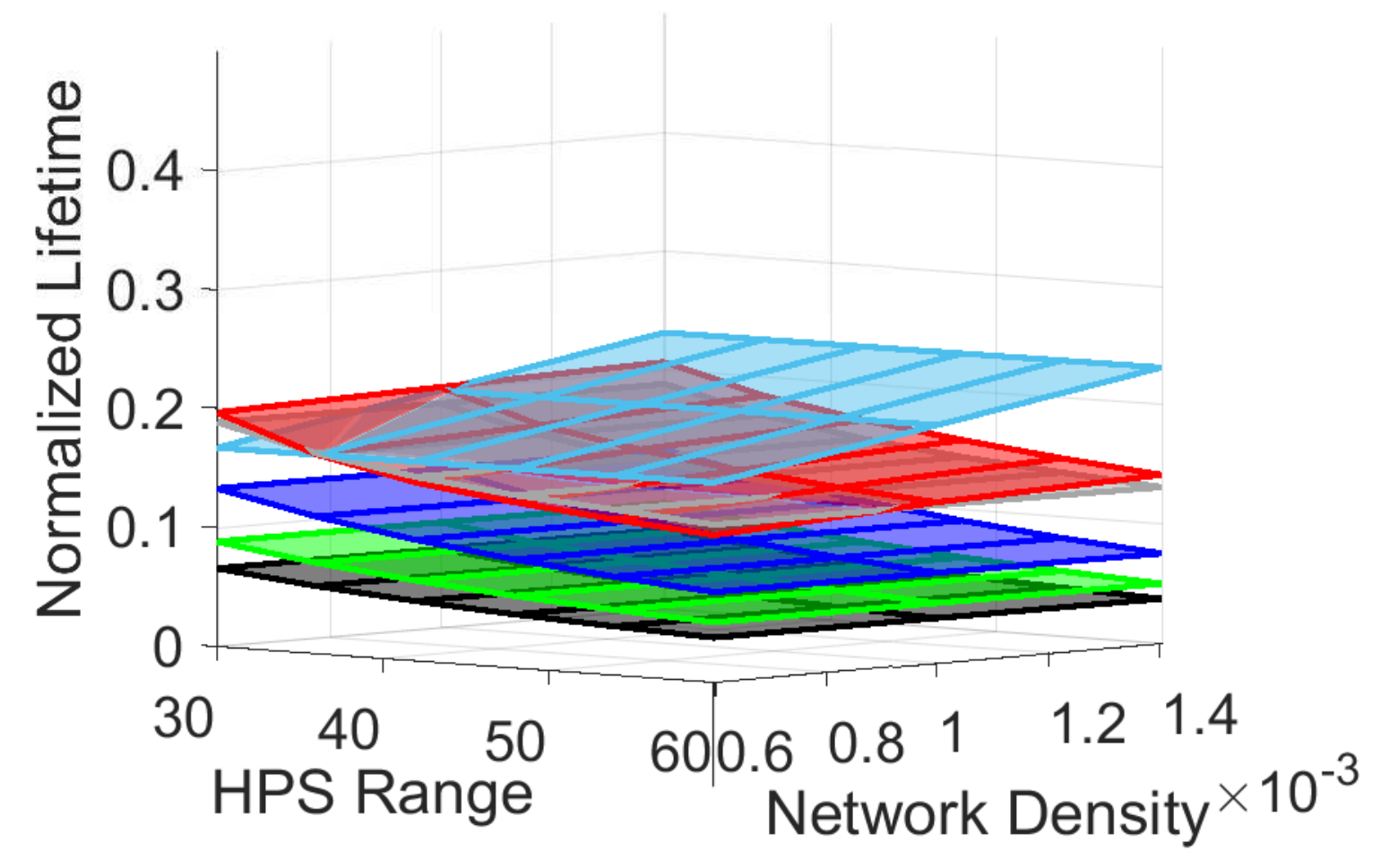}
        \caption{Network Lifetime with $\lambda=3$ targets in $\Omega_\gamma$}
        \label{fig:3D_LIFE_3}
    \end{subfigure}\vspace{6pt}

    \begin{subfigure}[t]{\textwidth}
    	\centering
        \includegraphics[width=0.9\textwidth]{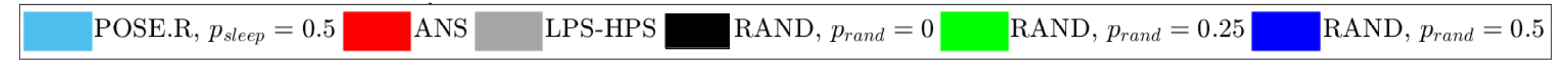}
    \end{subfigure}%

    \caption{POSE.R algorithm energy characteristics compared to existing methods.} \label{fig:3D_E_comp} \vspace{0pt}
\end{figure*}

\vspace{10pt}
\subsubsection{Energy Consumption and Network Lifetime Comparison}  While the POSE.R network achieves lower missed detection rates as compared to the other methods, it also consumes significantly less energy.  Specifically, Fig.~\ref{fig:3D_AVG_E_IN} shows the average energy consumption per node located within a distance of $R_L$ from the target's position. This result shows that for low network densities and when the other methods utilize small  HPS  ranges, the POSE.R network consumes slightly more energy. This is because for low network densities, $D_{b}^{\tau_\ell}<N_{sel}$, which requires POSE.R to select nodes outside of $R_1$ with larger sensing ranges to maintain the tracking performance, while the other methods are using a fixed small  HPS  range (yielding poor detection performance, as shown in Fig.~\ref{fig:3D_PM}). However, as the network density increases while the other methods use a small  HPS  range, the energy consumption of POSE.R decreases and approaches that of the ANS algorithm. This is because as the network density increases, it is likely that POSE.R is able to select $N_{sel}$ nodes within the $R_1$ distance of the target's position. Also, as seen in Fig.~\ref{fig:3D_AVG_E_IN}, when the other methods use larger  HPS  ranges, then they consume more energy than the POSE.R network. This is because the POSE.R algorithm opportunistically selects the optimal sensing range to track the target, thus highlighting the benefits of the DANS algorithm.

Fig.~\ref{fig:3D_AVG_E_OUT} shows the average energy consumption per node located at a distance greater than $R_L$ from the target's position. It is clearly seen that POSE.R consumes less energy than all the other methods. Since POSE.R, LPS-HPS, and ANS algorithms are opportunistic sensing methods, they consume less energy than the random methods. However, by virtue of incorporating a \emph{Sleep} state, POSE.R is the most energy-efficient algorithm.

Figs.~\ref{fig:3D_LIFE_0}, \ref{fig:3D_LIFE_1}, \ref{fig:3D_LIFE_2}, \ref{fig:3D_LIFE_3} compare the lifetime of the POSE.R network with the other networks for $\lambda=0,1,2,$ and $3$ targets, respectively. The network lifetime is normalized with the lifetime of a network with no targets, i.e. $\lambda=0$, and for  $p_{sleep}$=0.75. Each network was simulated in a $2R_{L} \times 600m$ tube with $\lambda$ targets traveling through its center in a straight line. The total life of the network is computed when all of the nodes located within $R_{LPS}$ of the targets' trajectories have no remaining energy. For $\lambda =0, 1$ and $2$ targets, as seen in Figs.~\ref{fig:3D_LIFE_0}, \ref{fig:3D_LIFE_1}, \ref{fig:3D_LIFE_2},  respectively, POSE.R achieves a significantly larger network lifetime as compared to the other methods. As $\lambda$ becomes large, i.e., $\lambda = 3$, as seen in Fig.~\ref{fig:3D_LIFE_3}, the lifetime of POSE.R method is still higher than all methods; however, the margin is less. This is because the tube $\Omega_\gamma$ becomes completely occupied with targets and almost all of the nodes are either in the LPS state or HPS state and are consuming more energy. Specifically, for very low network densities and when the other methods use low  HPS  ranges, the POSE.R network has slightly less lifetime as compared to the ANS network. This because the DANS algorithm in POSE.R opportunistically increases the  HPS  range of the selected nodes to ensure target tracking at the expense of energy consumption, while the ANS network conserves energy but is not always able to track the target with a low HPS range.

\begin{figure*}[t!]
    \centering
    \begin{subfigure}[t]{0.32\textwidth}
        \centering
        \includegraphics[width=\textwidth]{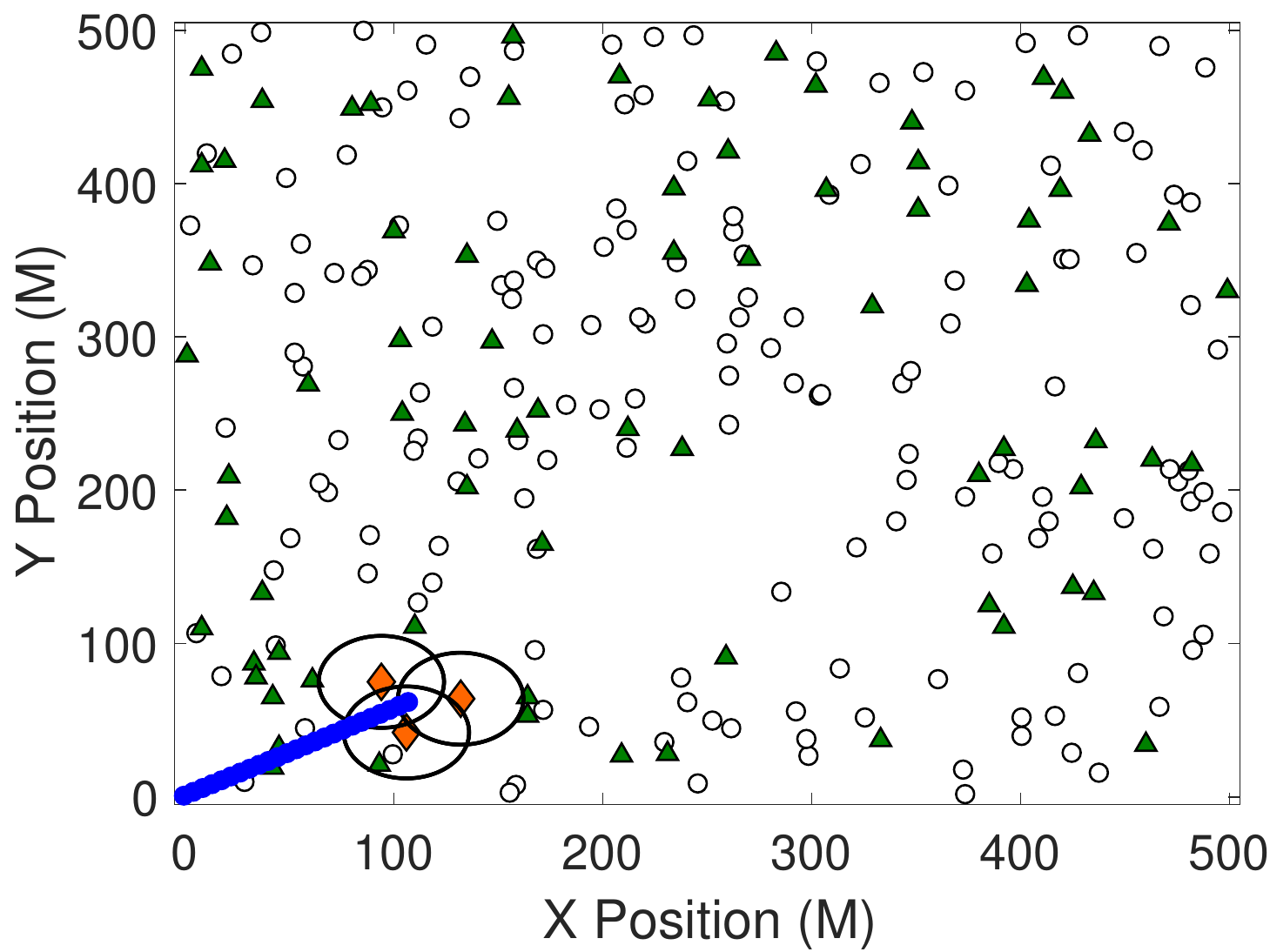}
        \caption{Time $k=25$: Target travelling in a high density region, where the selected sensors use their default smallest  HPS  range.}
        \label{fig:snap_25}
    \end{subfigure}%
	\begin{subfigure}[t]{0.32\textwidth}
        \centering
        \includegraphics[width=\textwidth]{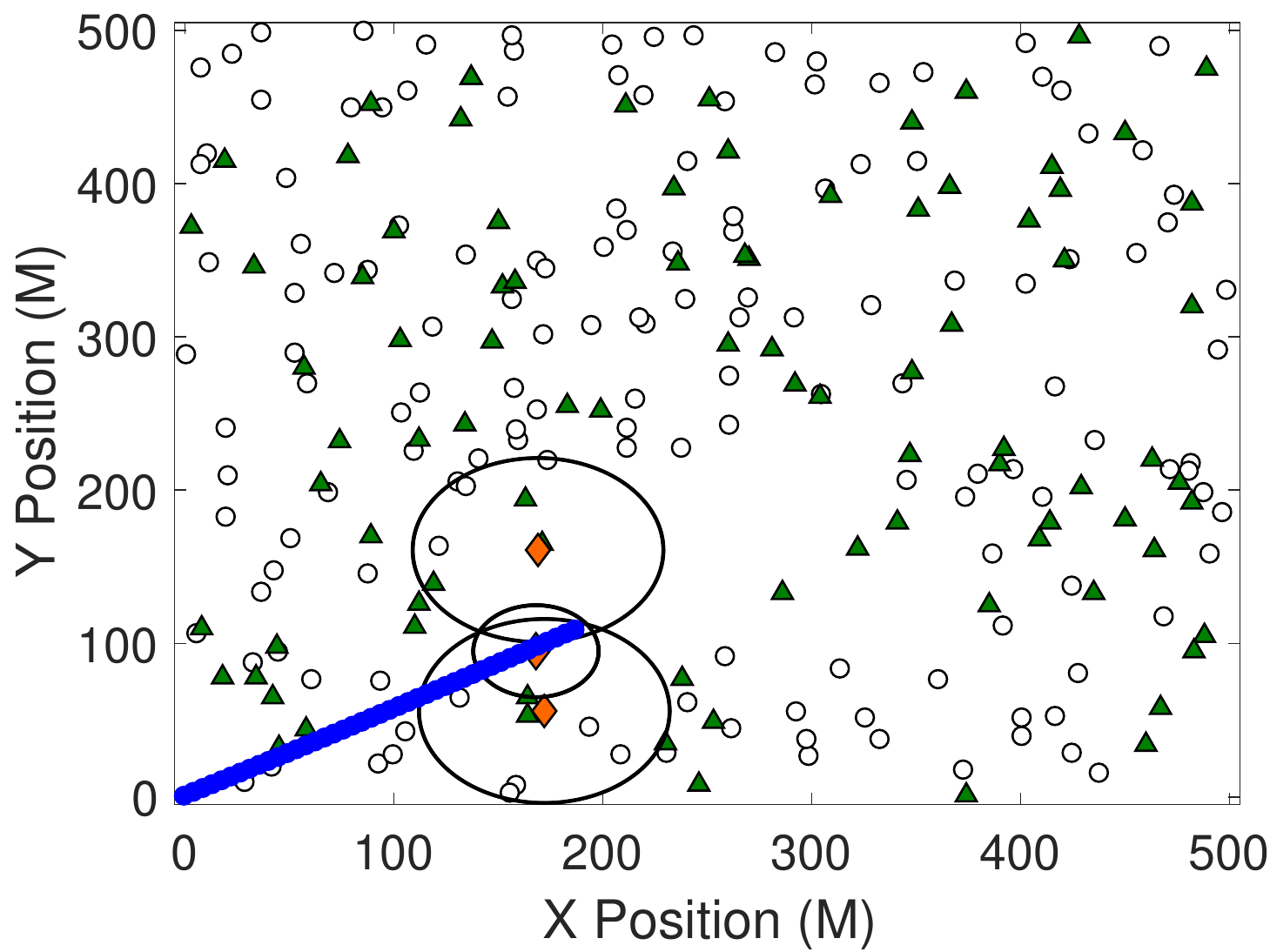}
        \caption{Time $k=43$: Target travelling in a low density region requiring adaptive sensor range selection.}
        \label{fig:snap_43}
    \end{subfigure}%
    \begin{subfigure}[t]{0.32\textwidth}
        \centering
        \includegraphics[width=\textwidth]{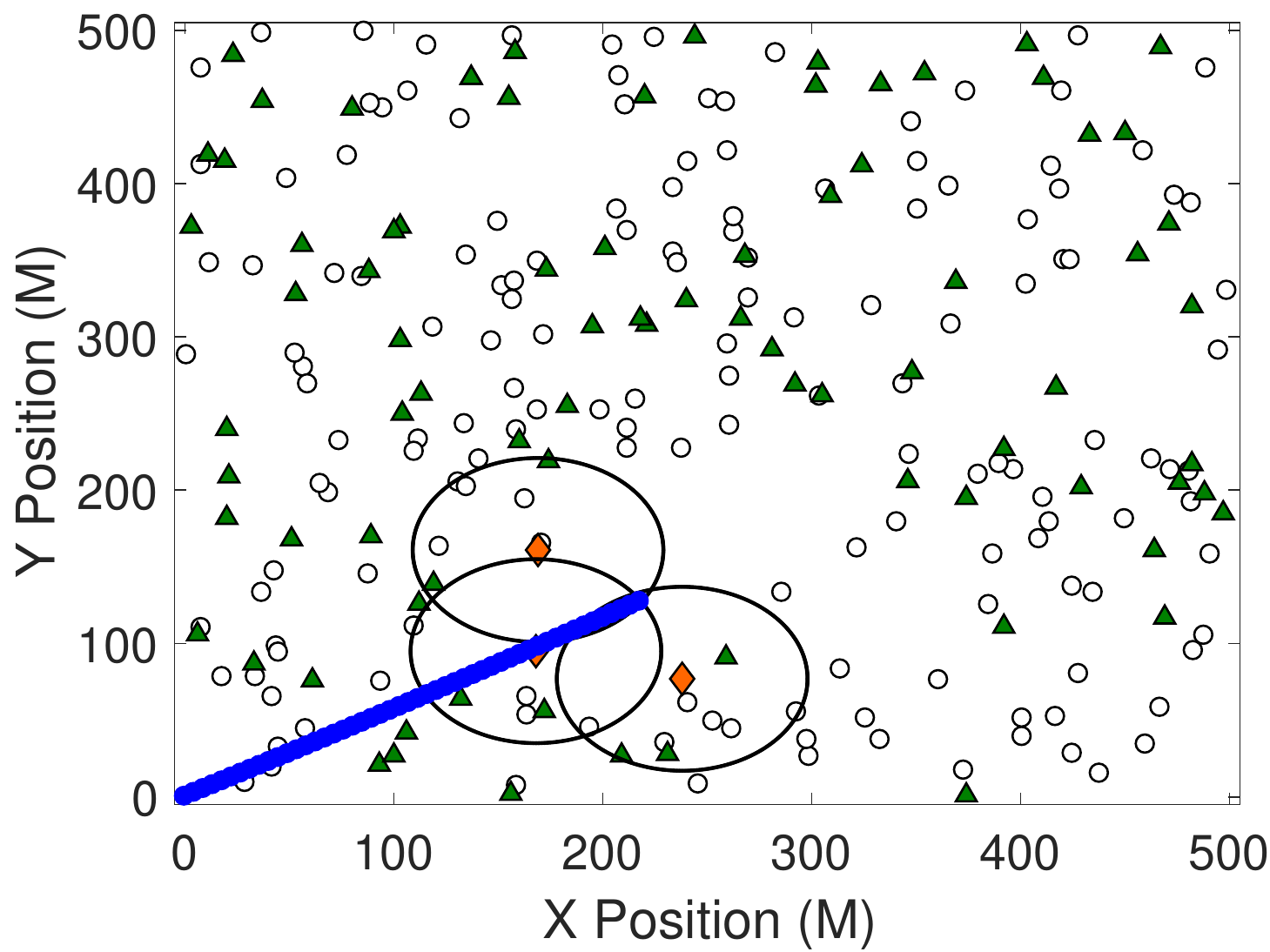}
        \caption{Time $k=50$: Target enters a coverage gap and all the selected sensors use their largest HPS sensing range.}
        \label{fig:snap_50}
    \end{subfigure}\vspace{6pt}

    \begin{subfigure}[t]{0.32\textwidth}
        \centering
        \includegraphics[width=\textwidth]{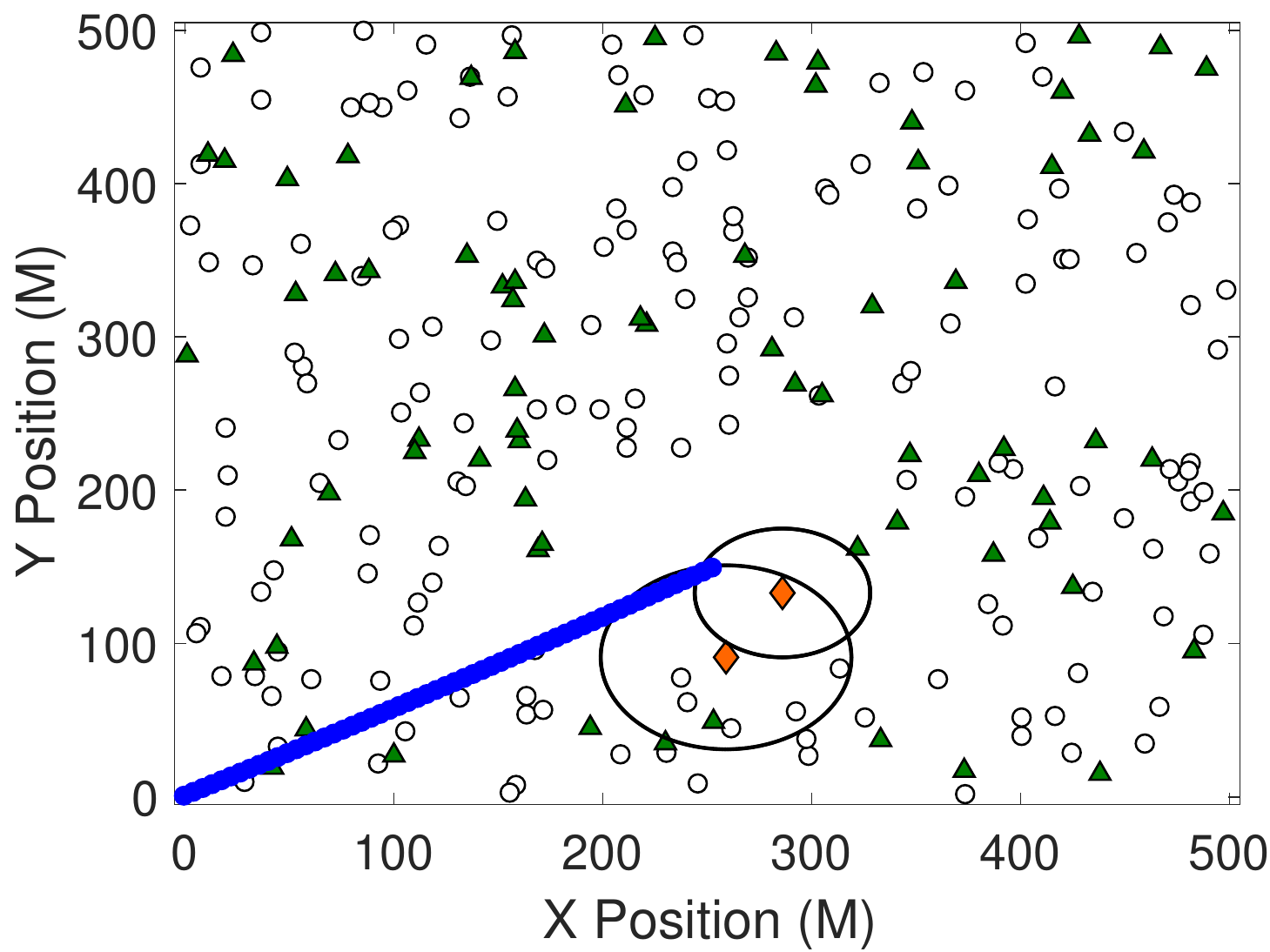}
        \caption{Time $k=58$: Target travelling in a low density region requiring adaptive sensor range selection.}
        \label{fig:snap_58}
    \end{subfigure}%
	\begin{subfigure}[t]{0.32\textwidth}
        \centering
        \includegraphics[width=\textwidth]{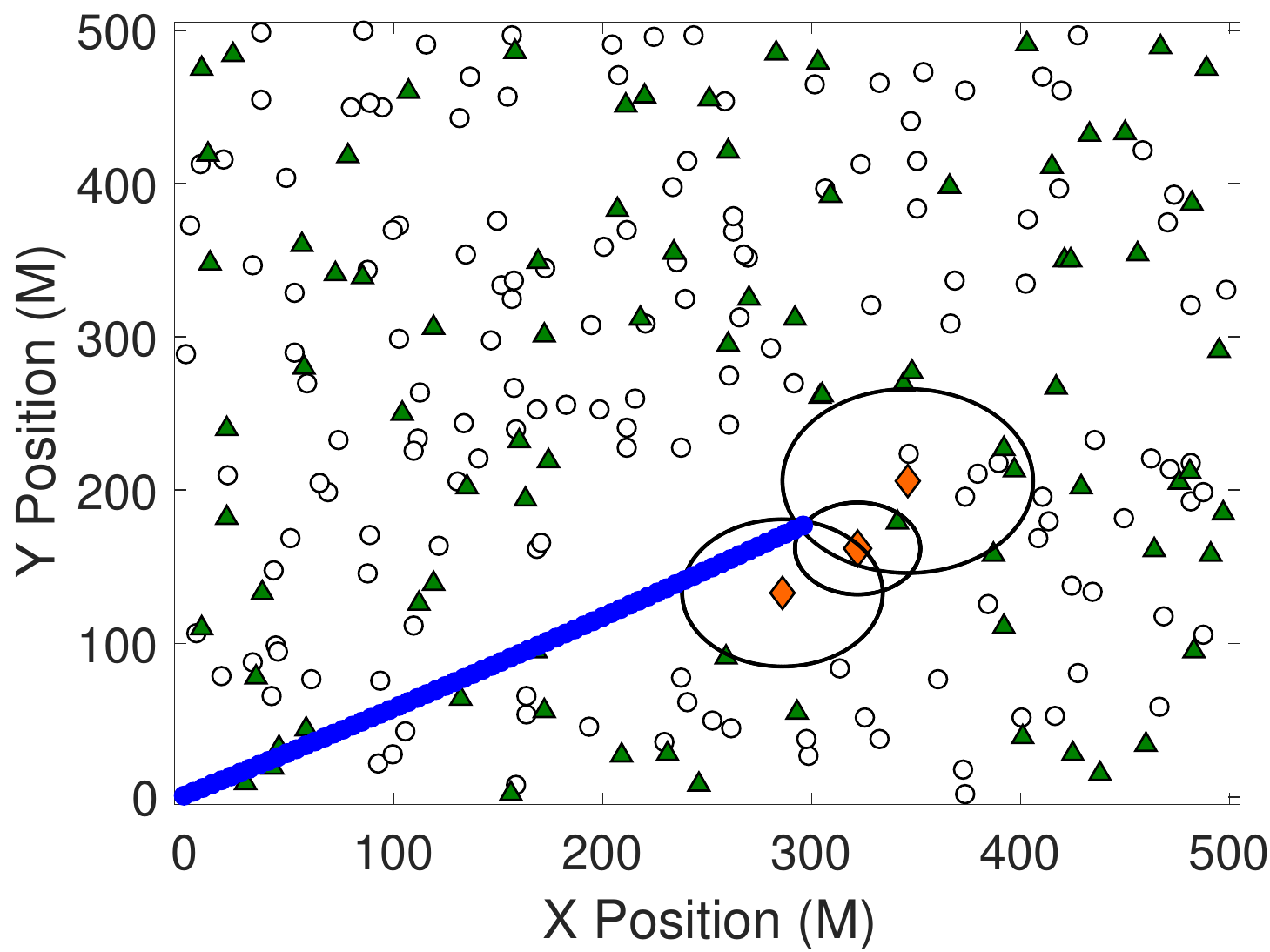}
        \caption{Time $k=68$: Target travelling in a low density region requiring adaptive sensor range selection.}
        \label{fig:snap_68}
    \end{subfigure}%
    \begin{subfigure}[t]{0.32\textwidth}
        \centering
        \includegraphics[width=\textwidth]{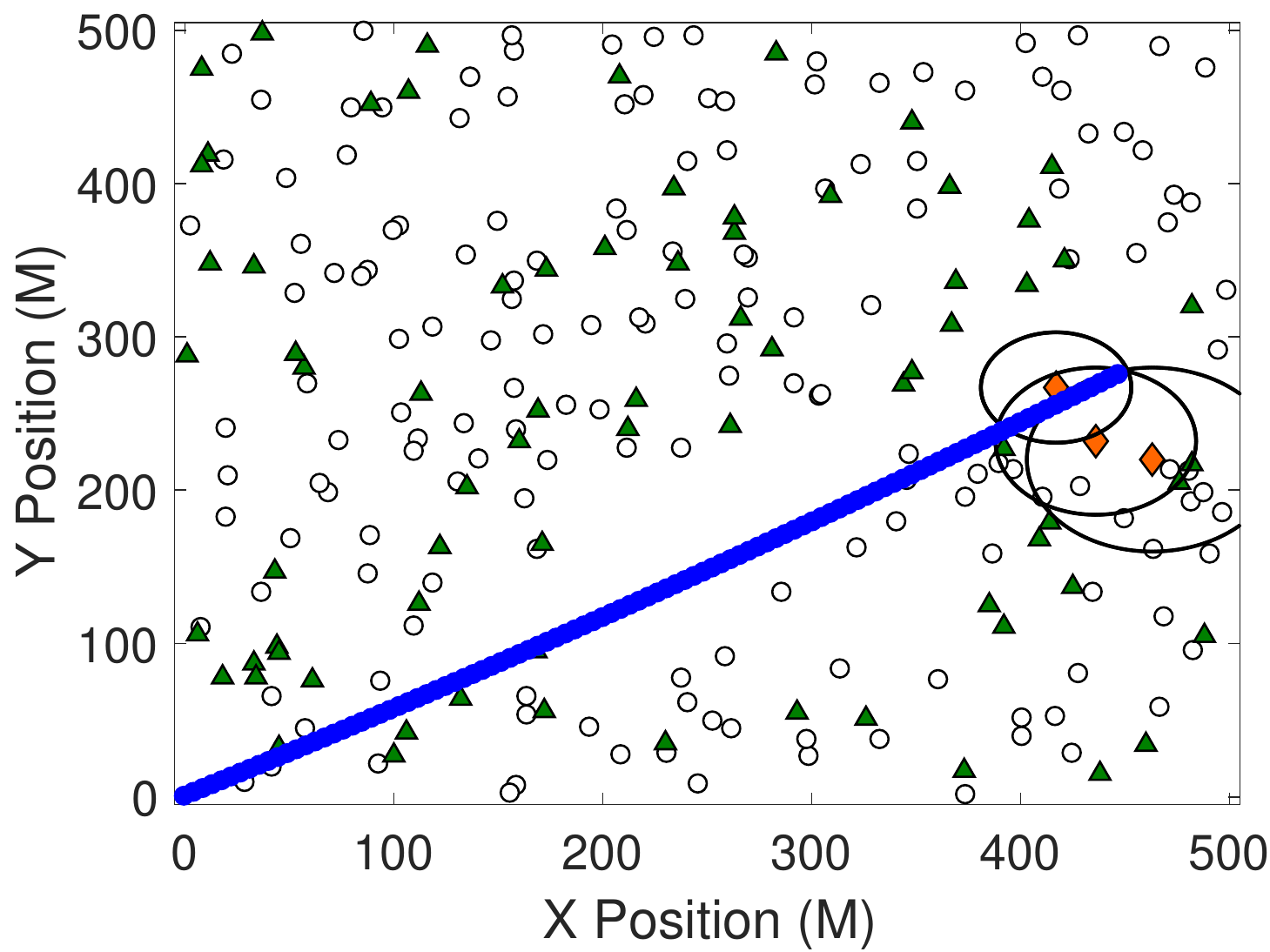}
        \caption{Time $k=103$: Target travelling in a low density region requiring adaptive sensor range selection.}
        \label{fig:snap_103}
    \end{subfigure}\vspace{6pt}

    \begin{subfigure}[t]{\textwidth}
    	\centering
        \includegraphics[width=0.9\textwidth]{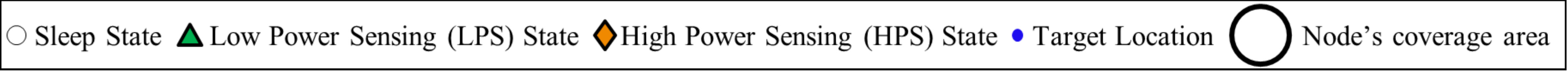}
    \end{subfigure}%

    \caption{Illustration of how POSE.R provides resilient tracking by adapting  HPS  sensing ranges of selected nodes when the target travels through low density regions or a coverage gap.} \label{fig:Snap_Shots} \vspace{6pt}
\end{figure*}

The choice of stopping after $3$ targets is due to the length of the simulated tube, for which the POSE.R lifetime characteristic is  saturating, as shown in Fig~\ref{fig:POSER_Char}. Once the number of targets increases above $3$, the number of nodes in the Sleep state decreases significantly and the POSE.R network acts as an LPS-HPS network due to majority of the tube being covered. Eventually, in the limiting case where a constant procession of targets are traveling through the tube, the network will act as an all on network (i.e., Random Scheduling network with $p_{sleep}=0$) to ensure that every target is tracked. The baseline lifetime performance is seen by the black planes in Figs.~\ref{fig:3D_LIFE_0}, \ref{fig:3D_LIFE_1}, \ref{fig:3D_LIFE_2}, \ref{fig:3D_LIFE_3}. For bigger networks, POSE.R will show significant energy savings for more number of targets.

\begin{figure*}[t!]
    \centering
    \begin{subfigure}[t]{0.32\textwidth}
        \centering
        \includegraphics[width=\textwidth]{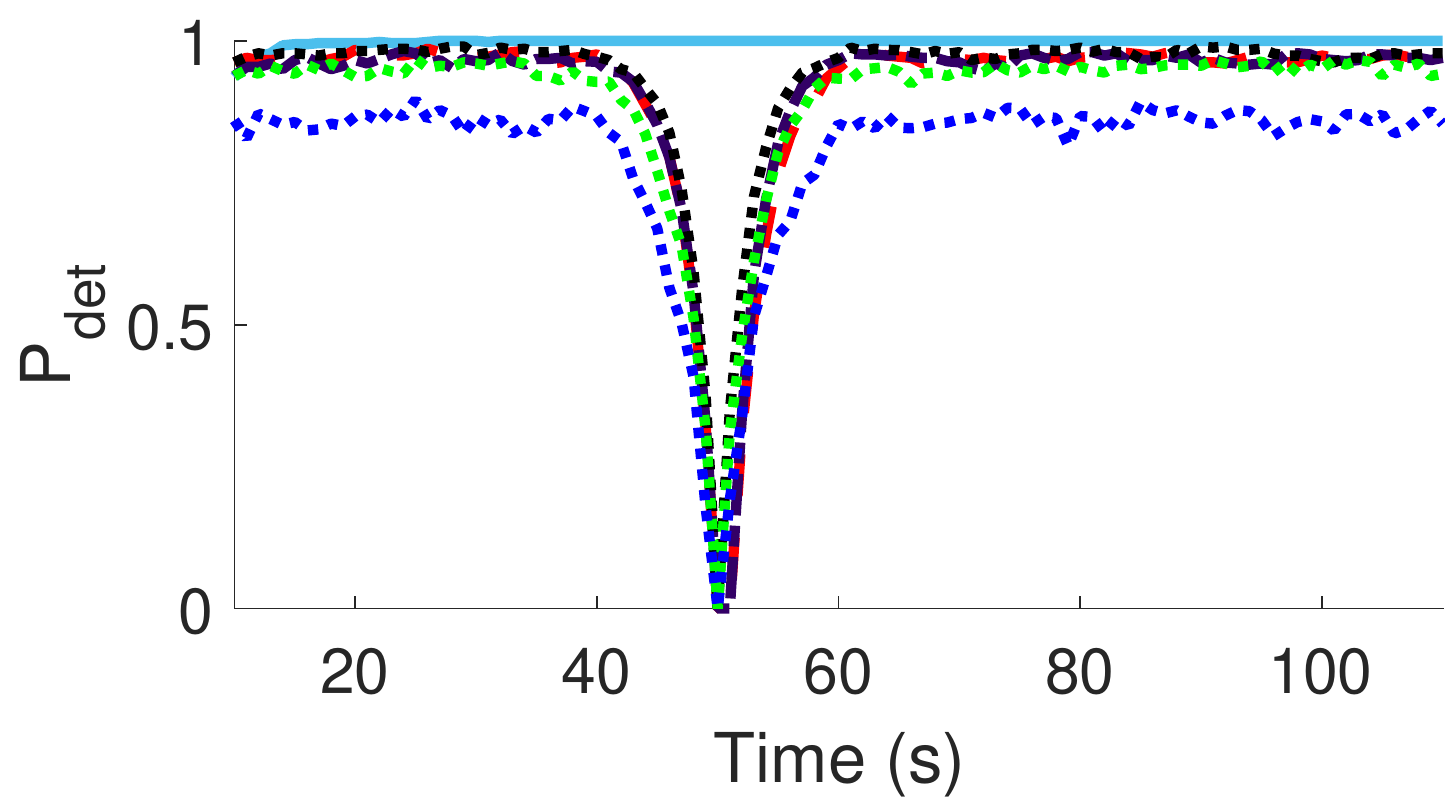} \vspace{-6pt}
     \caption{$R_{gap}=30, R_{HPS} = 30$}
    \end{subfigure}%
    ~
    \begin{subfigure}[t]{0.32\textwidth}
        \centering
        \includegraphics[width=\textwidth]{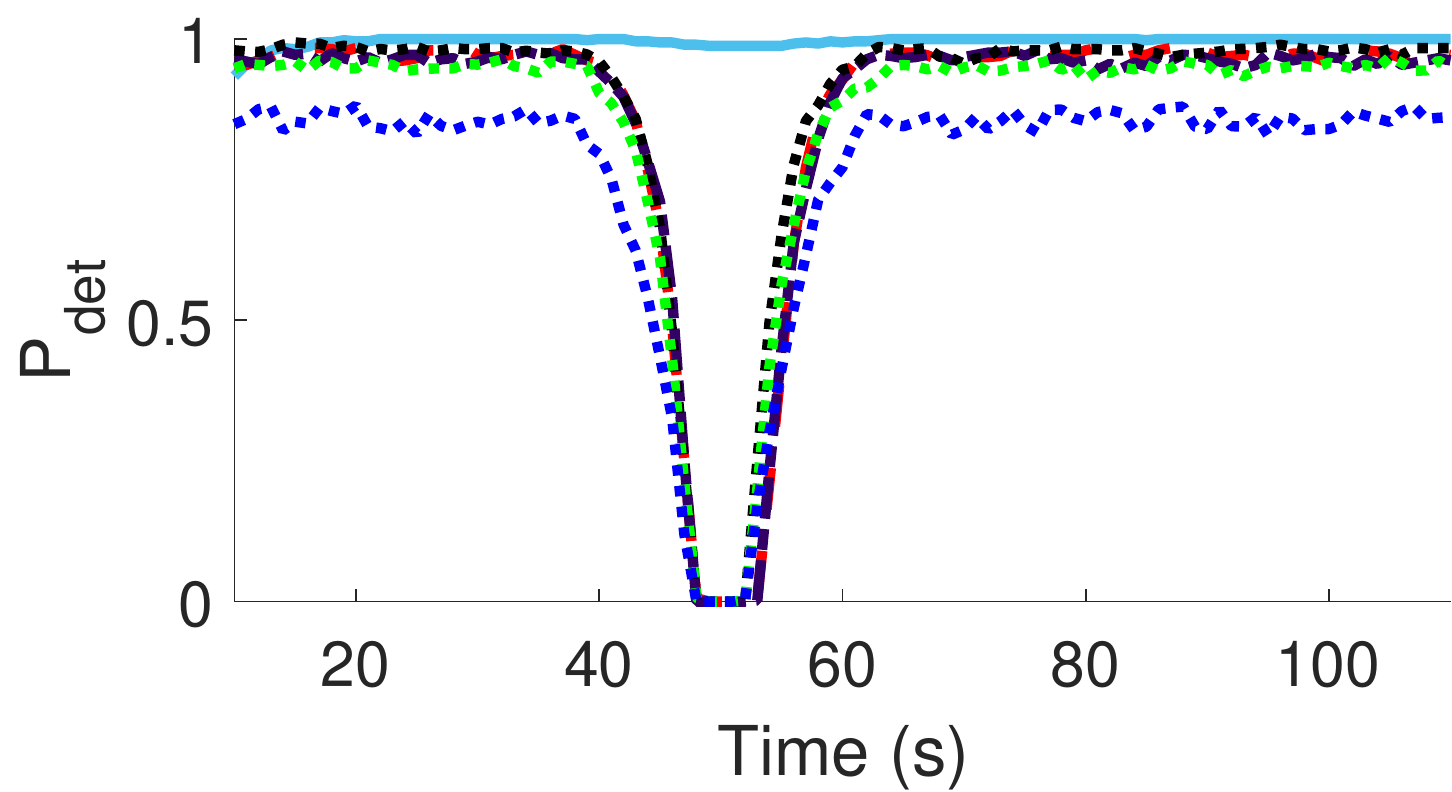}  \vspace{-6pt}
     \caption{$R_{gap}=40, R_{HPS} = 30$}
    \end{subfigure}%
     ~
    \begin{subfigure}[t]{0.32\textwidth}
        \centering
        \includegraphics[width=\textwidth]{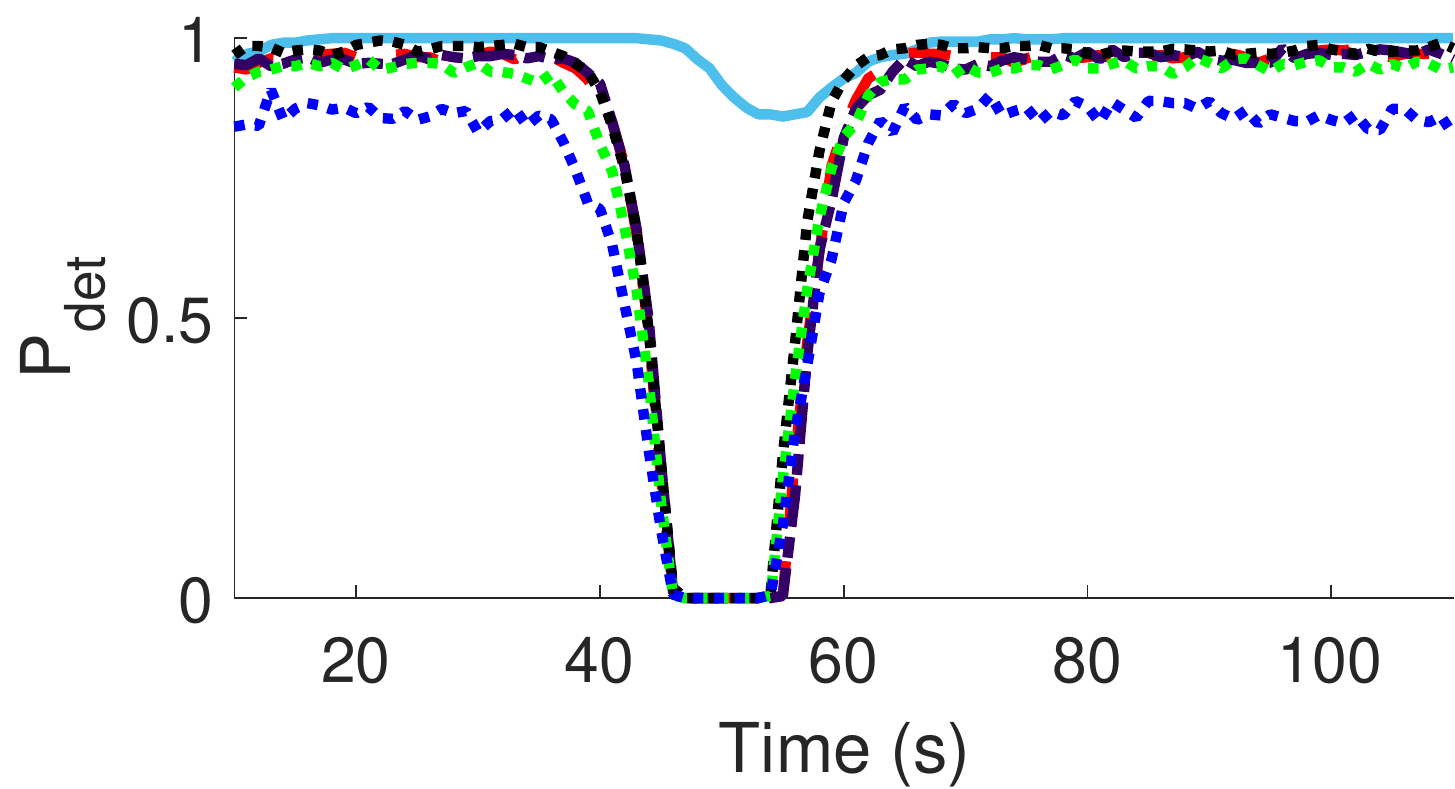}  \vspace{-6pt}
     \caption{$R_{gap}=50, R_{HPS} = 30$}
    \end{subfigure}\vspace{6pt}

    \begin{subfigure}[t]{0.32\textwidth}
        \centering
        \includegraphics[width=\textwidth]{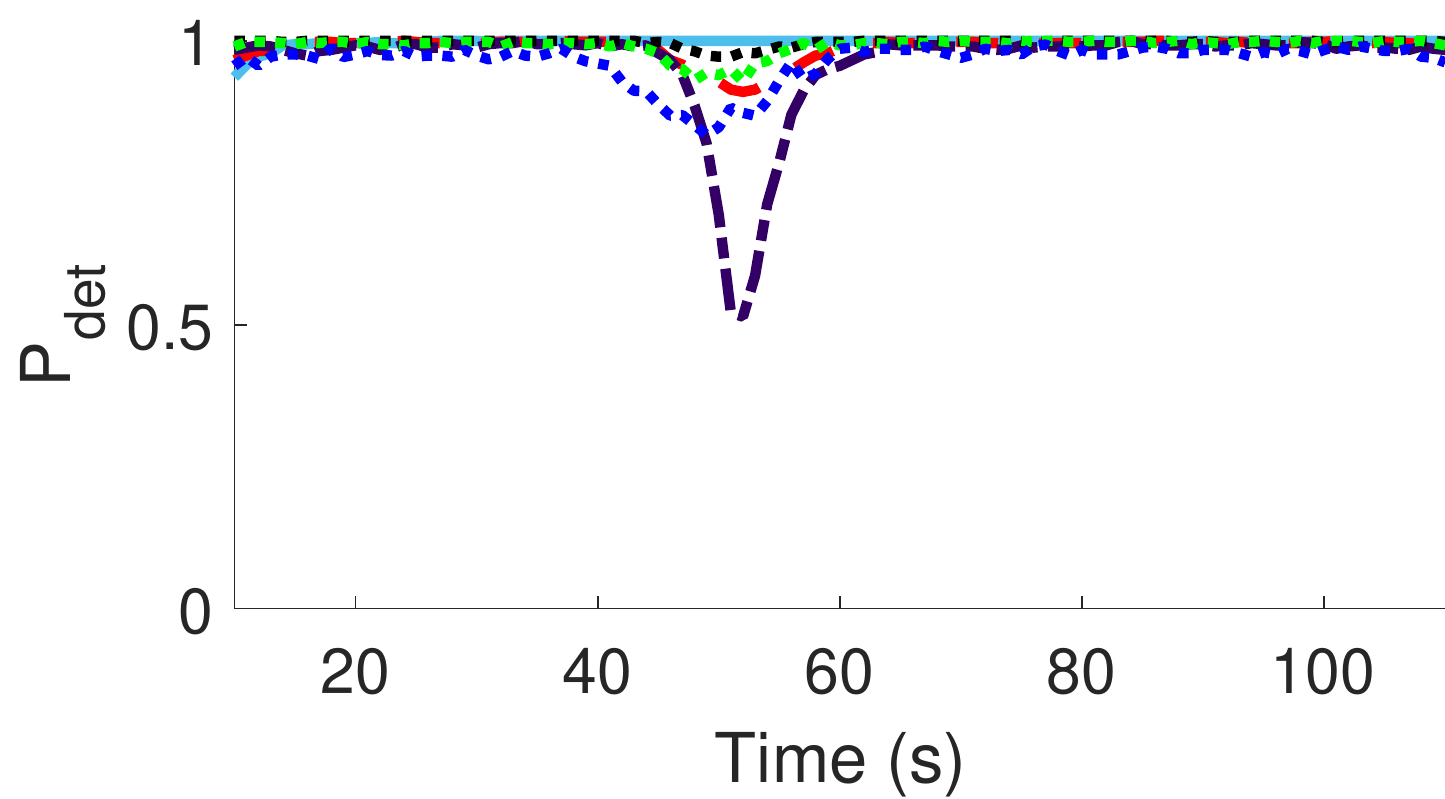}
     \caption{$R_{gap}=30, R_{HPS} = 42$}
    \end{subfigure}%
    ~
    \begin{subfigure}[t]{0.32\textwidth}
        \centering
        \includegraphics[width=\textwidth]{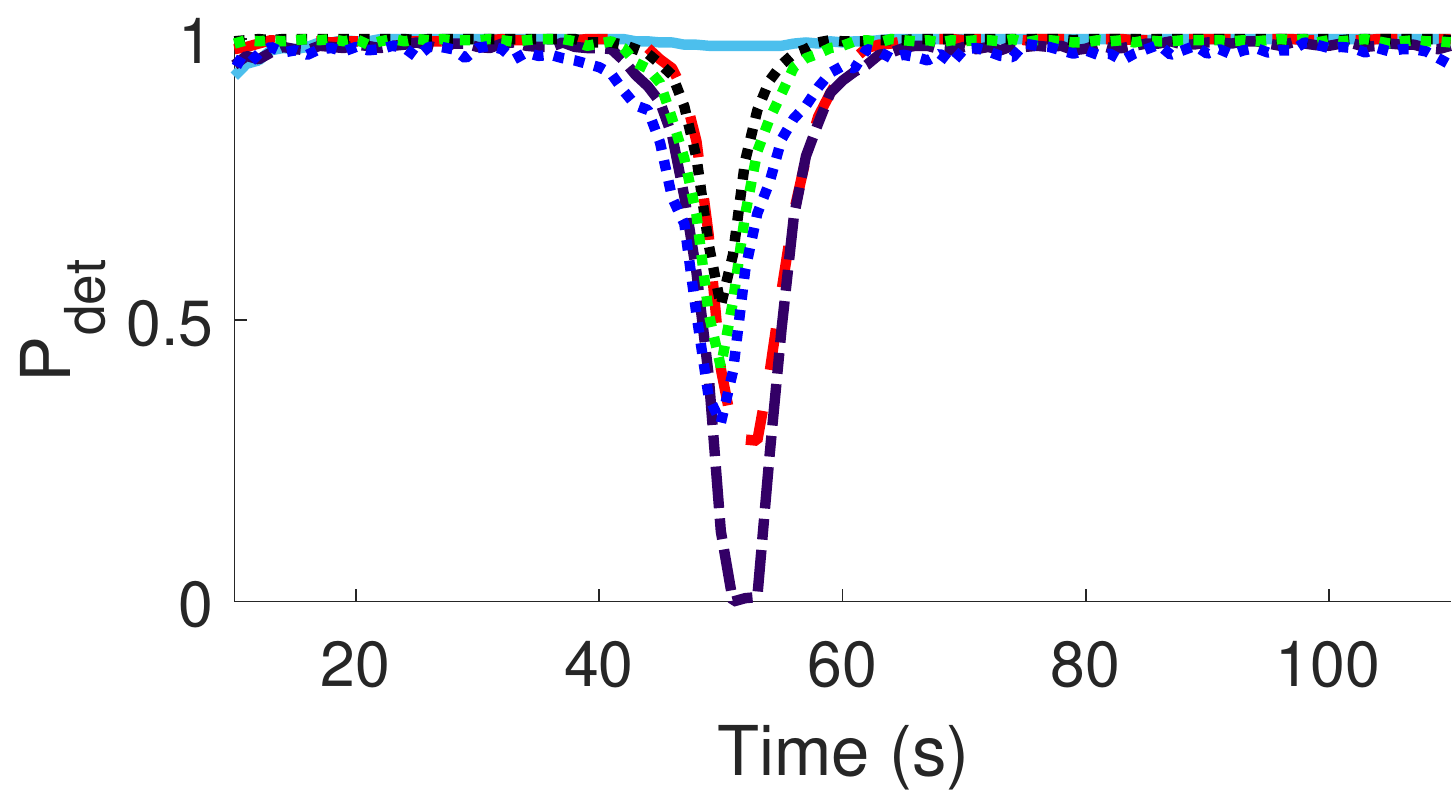}
     \caption{$R_{gap}=40, R_{HPS} = 42$}
    \end{subfigure}%
    ~
    \begin{subfigure}[t]{0.32\textwidth}
        \centering
        \includegraphics[width=\textwidth]{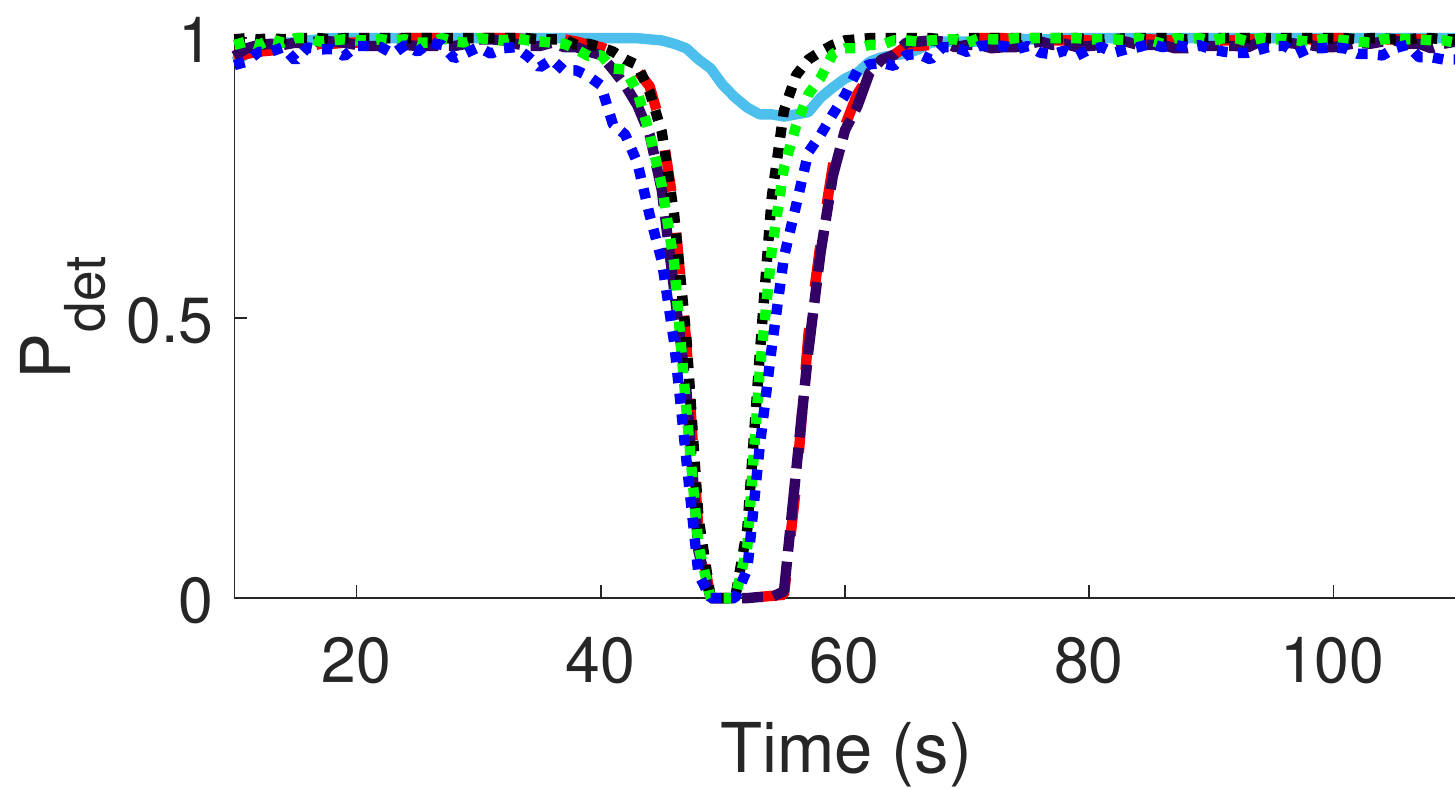}
    \caption{$R_{gap}=50, R_{HPS} = 42$}
    \end{subfigure}\vspace{6pt}

	\begin{subfigure}[t]{0.32\textwidth}
        \centering
        \includegraphics[width=\textwidth]{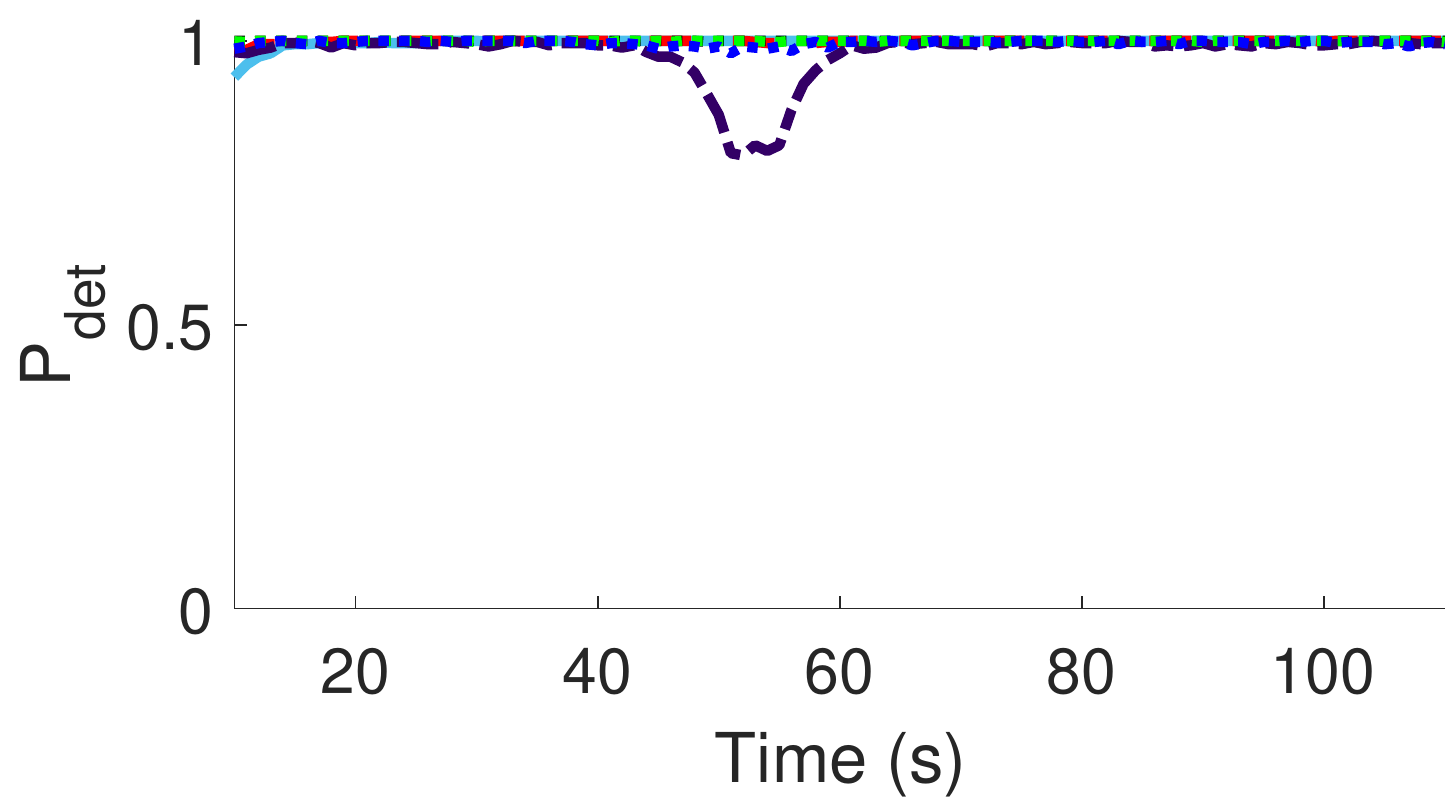}
     \caption{$R_{gap}=30, R_{HPS} = 54$}
    \end{subfigure}%
    ~
    \begin{subfigure}[t]{0.32\textwidth}
        \centering
        \includegraphics[width=\textwidth]{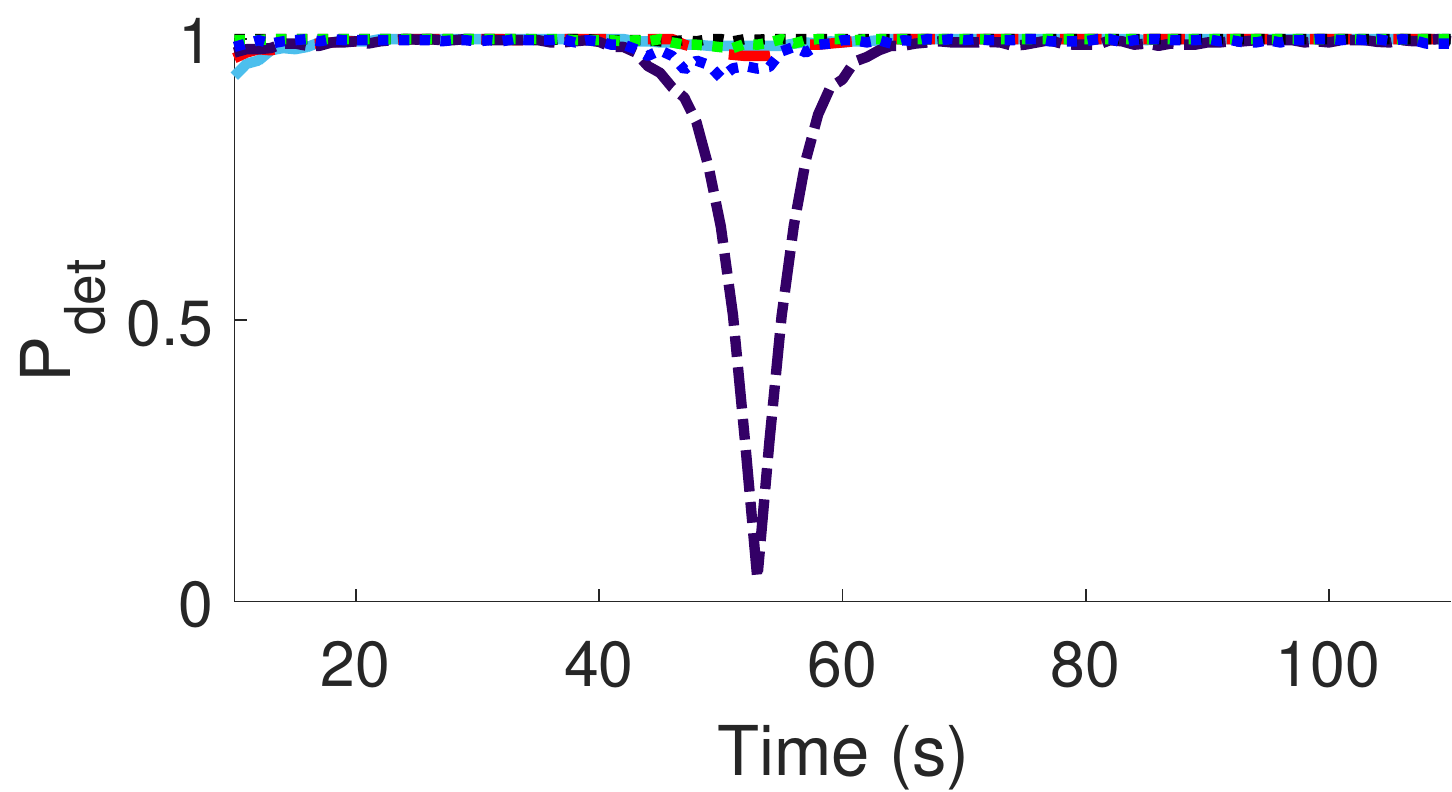}
     \caption{$R_{gap}=40, R_{HPS} = 54$}
    \end{subfigure}%
    ~
    \begin{subfigure}[t]{0.32\textwidth}
        \centering
        \includegraphics[width=\textwidth]{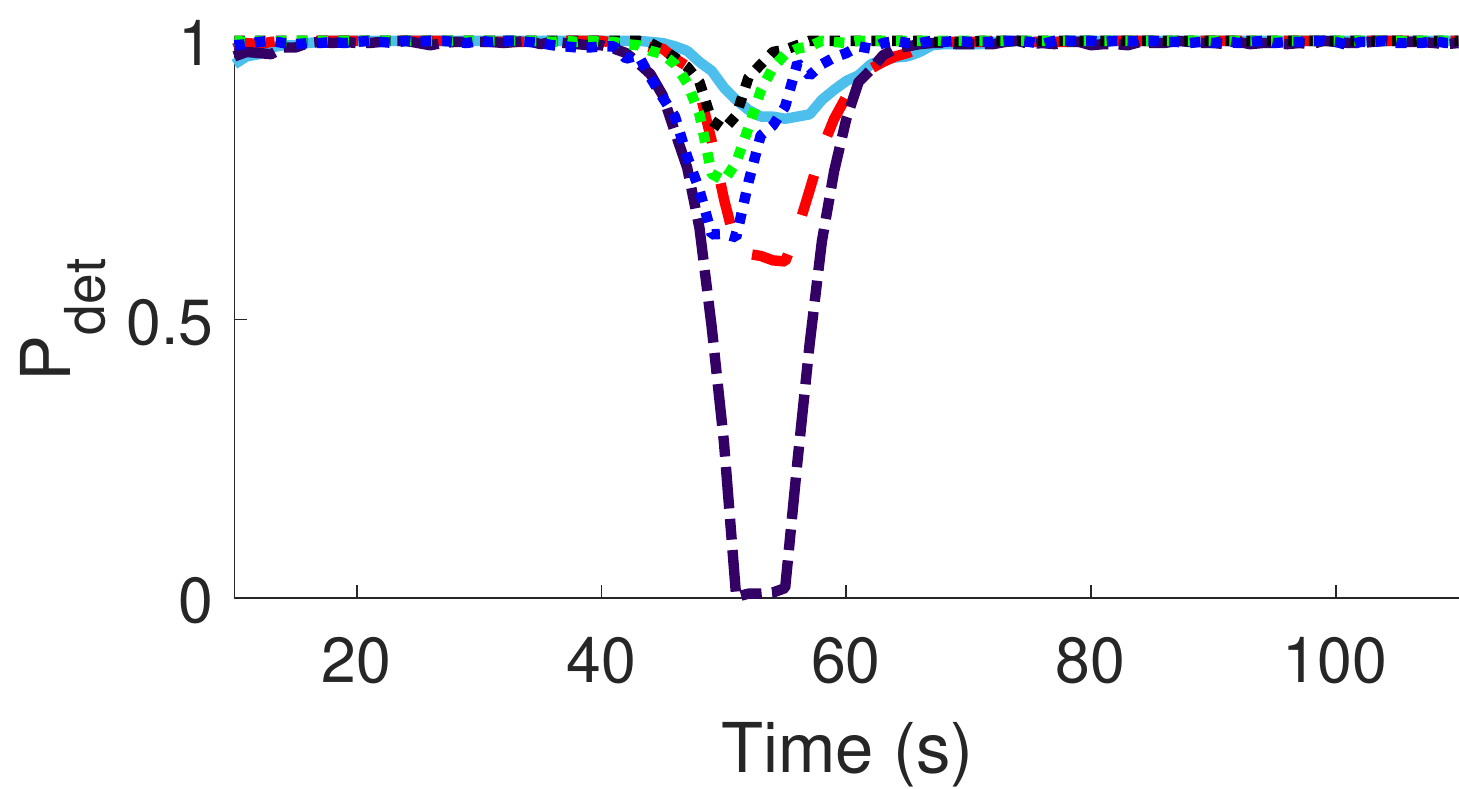}
    \caption{$R_{gap}=50, R_{HPS} = 54$}
    \end{subfigure}\vspace{6pt}

    \begin{subfigure}[t]{0.32\textwidth}
        \centering
        \includegraphics[width=\textwidth]{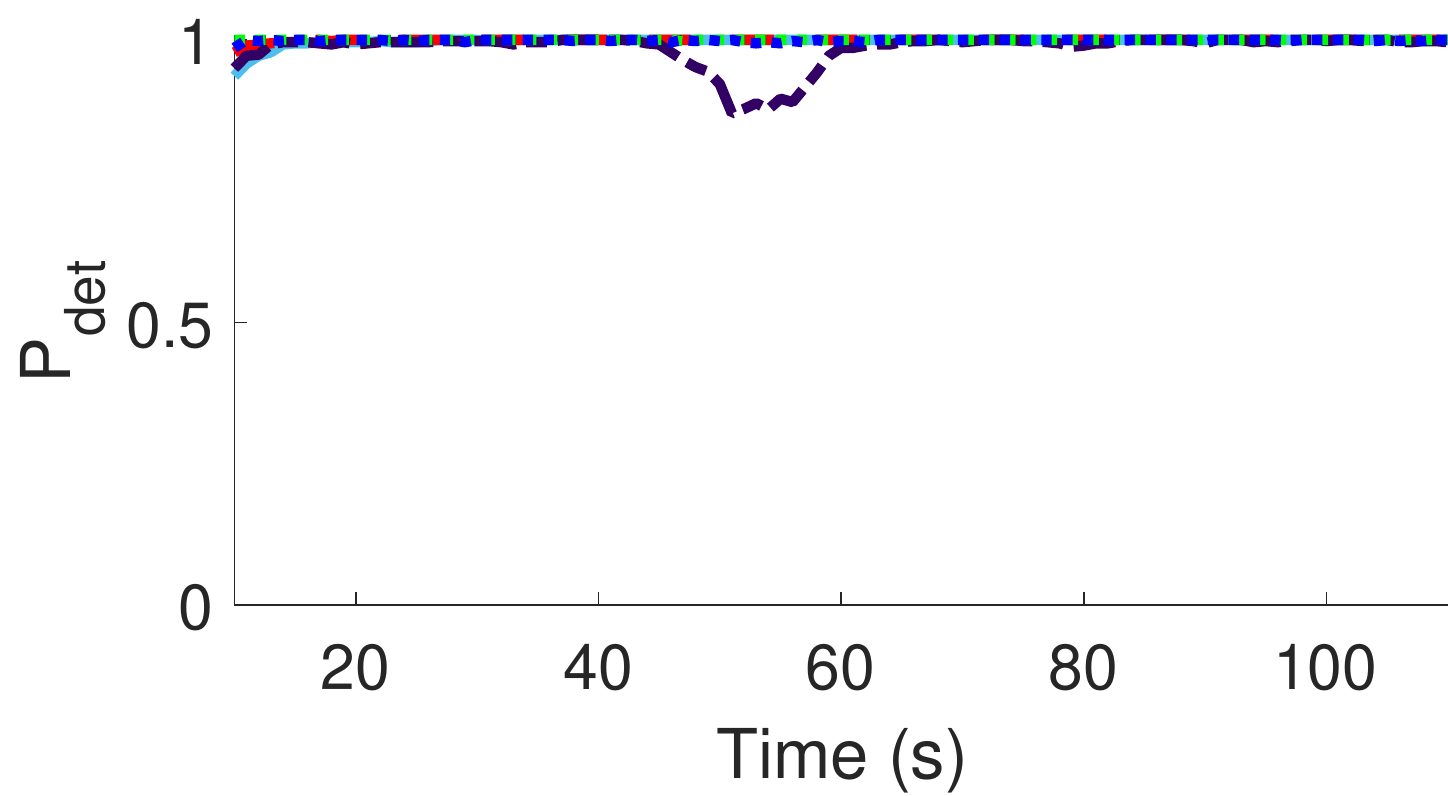}
     \caption{$R_{gap}=30, R_{HPS} = 60$}
    \end{subfigure}%
    ~
    \begin{subfigure}[t]{0.32\textwidth}
        \centering
        \includegraphics[width=\textwidth]{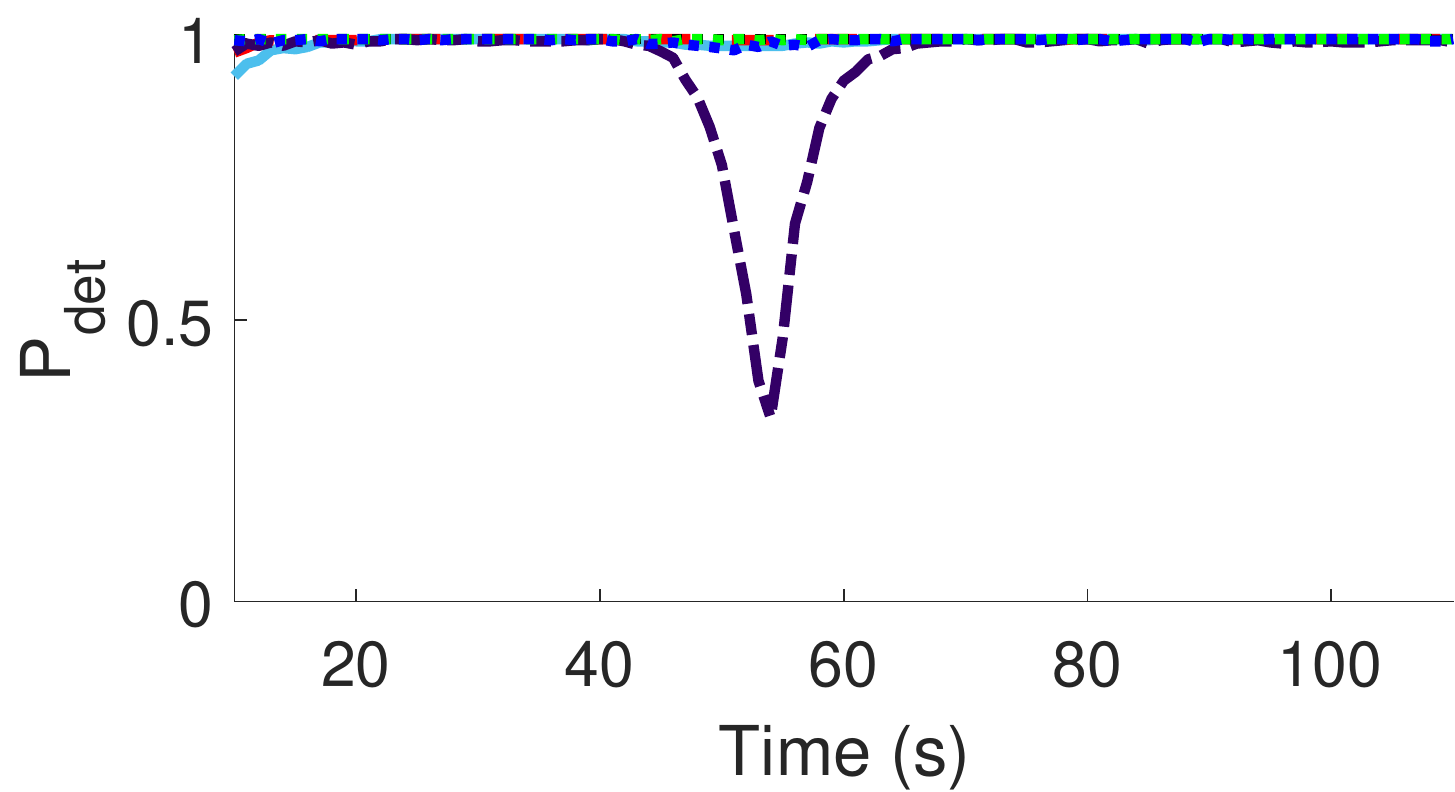}
    \caption{$R_{gap}=40, R_{HPS} = 60$}
    \end{subfigure}
    ~
    \begin{subfigure}[t]{0.32\textwidth}
        \centering
        \includegraphics[width=\textwidth]{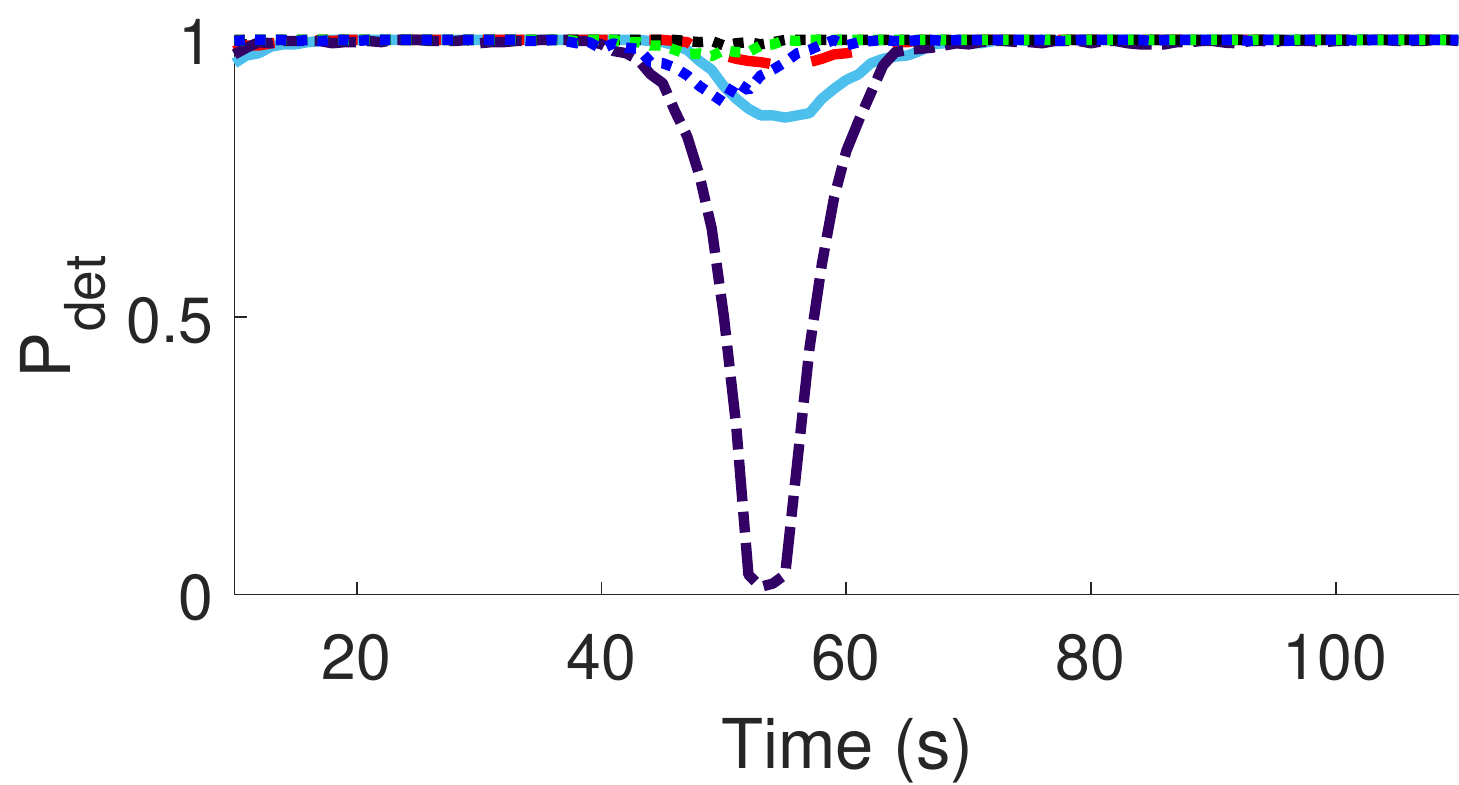}
     \caption{$R_{gap}=50, R_{HPS} = 60$}
    \end{subfigure}\vspace{9pt}

    \begin{subfigure}[t]{\textwidth}
        \centering
        \includegraphics[width=\textwidth]{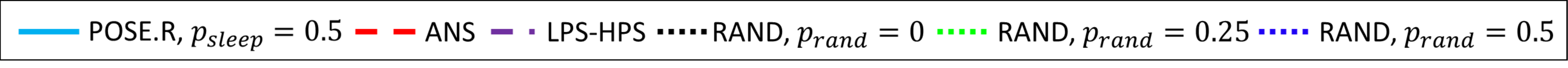}
    \end{subfigure}

    \caption{Probability of detection as target travels through a coverage gap around $t=50s$.}
    \label{fig:time_gap} \vspace{6pt}
\end{figure*}

\vspace{12pt}
\subsubsection{Network Resilience Comparison}\label{rescomp}
Fig.~\ref{fig:Snap_Shots} illustrates the workings of the POSE.R algorithm as the target travels through regions of high and low network densities as well as coverage gaps. It shows how the POSE.R algorithm selects the nodes and adapts their  HPS  sensing ranges to track the target when it travels through different regions. Fig.~\ref{fig:snap_25} shows a situation when the target is traveling in a high density region. In this situation, the  HPS  nodes are selected using EGDOP with a sensing range $R_{HPS}^{s_i}=R_1$. Figs. \ref{fig:snap_43}, \ref{fig:snap_50}, \ref{fig:snap_58}, \ref{fig:snap_68} and \ref{fig:snap_103} show situations when the target is traveling through low density regions or a coverage gap, i.e., the base coverage degree $\widetilde{D}_b(k+1)<N_{sel}$. In these situations, the selected nodes adjust their  HPS  sensing ranges to ensure target tracking. Thus, POSE.R enables the nodes to autonomously adapt their sensing ranges in an optimal manner to maintain tracking throughout the target's trajectory, even in the presence of low network densities and coverage gaps, thereby exhibiting resilience.

Fig.~\ref{fig:time_gap} compares the detection performance of POSE.R with other methods when the target travels through a
region where multiple spatially co-located nodes have failed or a coverage gap is present. A network with a density of $\rho=1.4e^{-3}$ was considered with a single target. To simulate a coverage gap, the nodes located within a circle of radius $R_{gap}\in\{30, 40, 50\}$ around the target's position at time $t=50s$, are assigned an initial energy value $E_0=0$. This simulates a group of ineffective sensor nodes creating a coverage gap of size $\geq R_{gap}$. As seen in Fig.~\ref{fig:time_gap}, the probability of detection, $P_{det}$, is presented for various $R_{gap}$ values and for different $R_{HPS}$ used by the other methods. In any single row of Fig.~\ref{fig:time_gap}, $R_{HPS}$ is fixed while $R_{gap}$ is increased. For any row, as $R_{gap}$ increases, the detection performance of the other methods deteriorate and their $P_{det}$ decreases and reaches zero when the target passes through the coverage gap. On the other hand, POSE.R yields a $P_{det}$ close to $1$, thus exhibiting resilience via adaptive node and range selection. When the other methods use a high $R_{HPS}$, as seen in a single column of Fig.~\ref{fig:time_gap}, their performance improves but at the expense of consuming more energy. This result indicates that the other methods lose the target for low  HPS  ranges when it travels through the coverage gap. However, POSE.R is able to continuously track the target by adaptive node selection and optimal sensor range selection.

\begin{figure*}[t!]
    \centering
    \begin{subfigure}[t]{0.31\textwidth}
        \centering
        \includegraphics[width=0.90\textwidth]{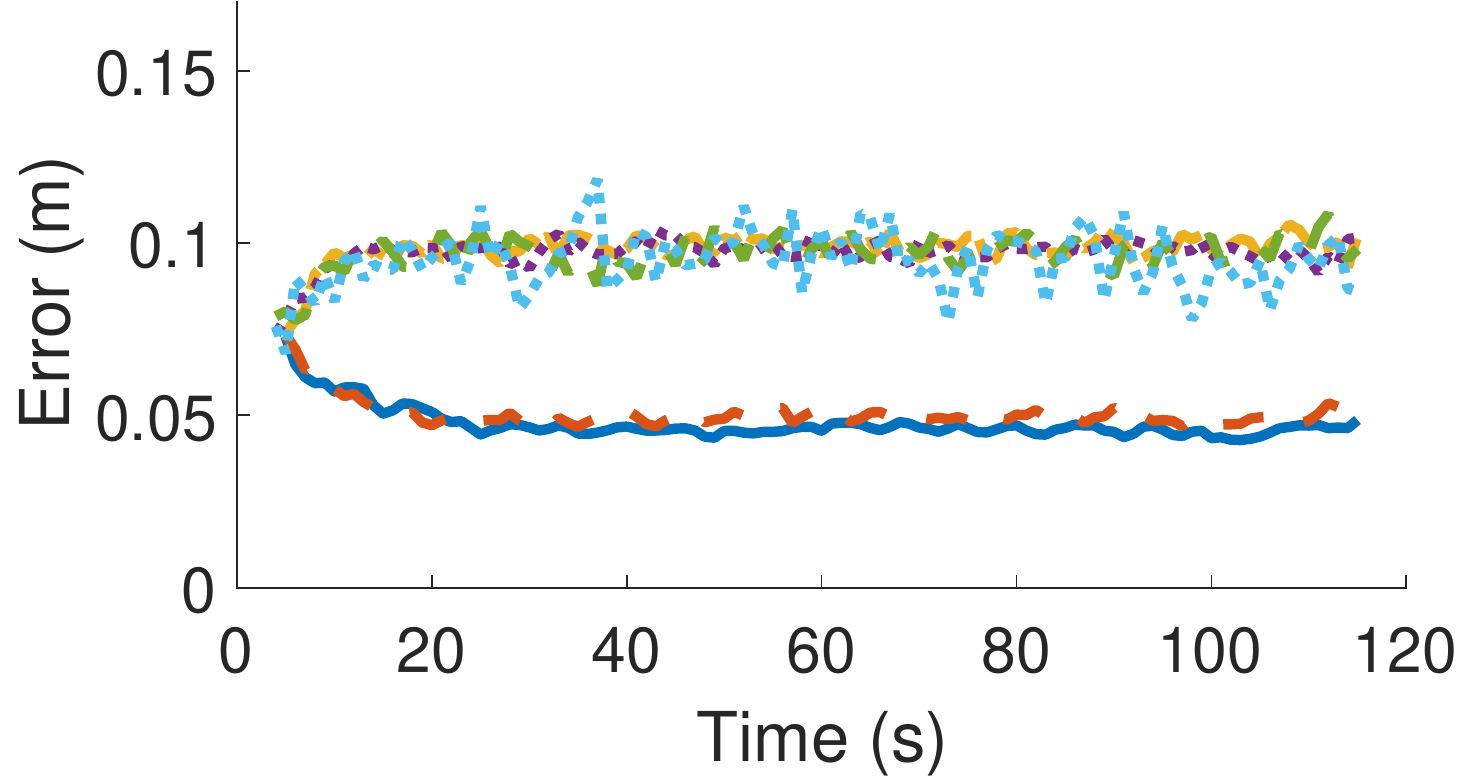}
        \caption*{$\textrm{i)}$ $R_{HPS}=30$}
    \end{subfigure}%
	\begin{subfigure}[t]{0.31\textwidth}
        \centering
        \includegraphics[width=0.90\textwidth]{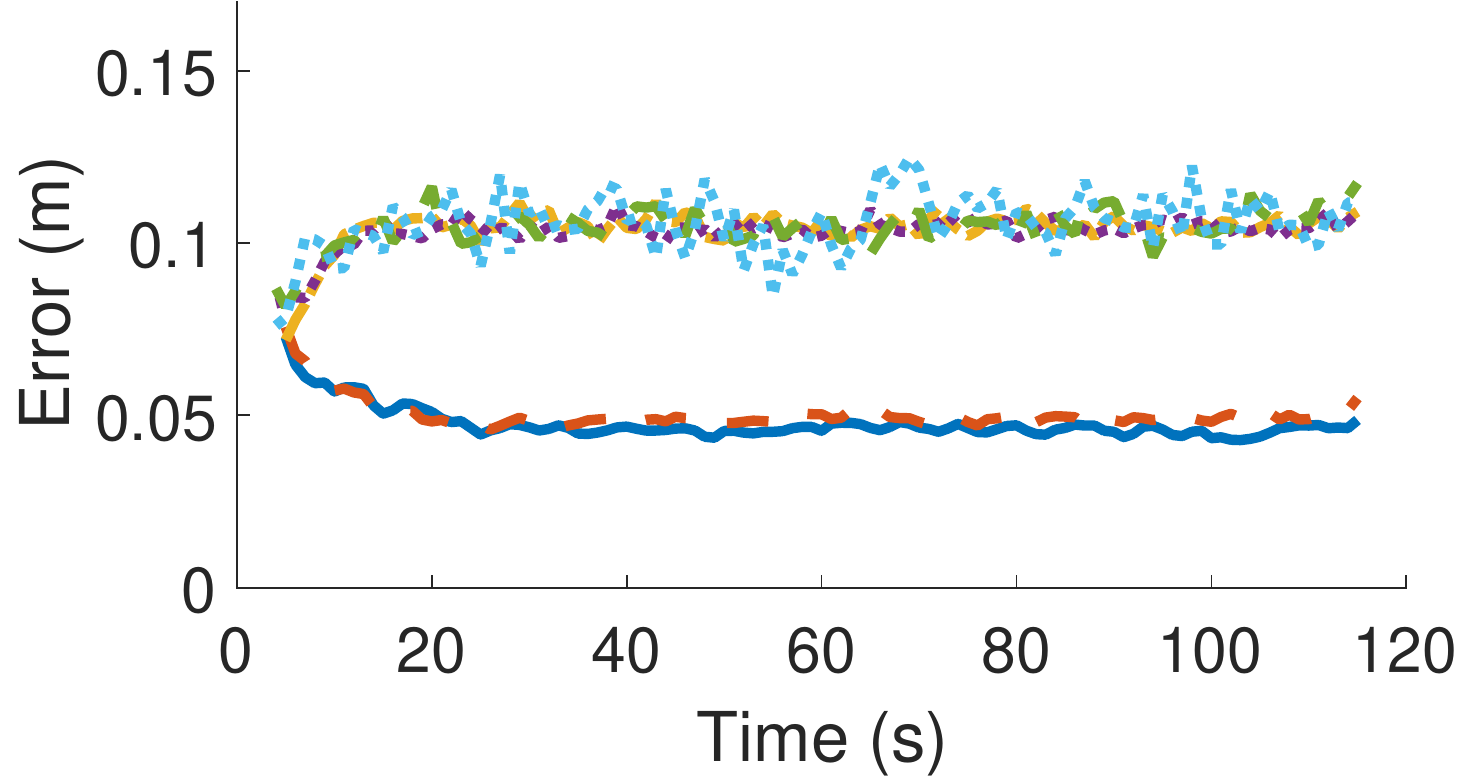}
        \caption*{$\textrm{ii)}$ $R_{HPS}=36$}
    \end{subfigure}%
    \begin{subfigure}[t]{0.31\textwidth}
        \centering
        \includegraphics[width=0.90\textwidth]{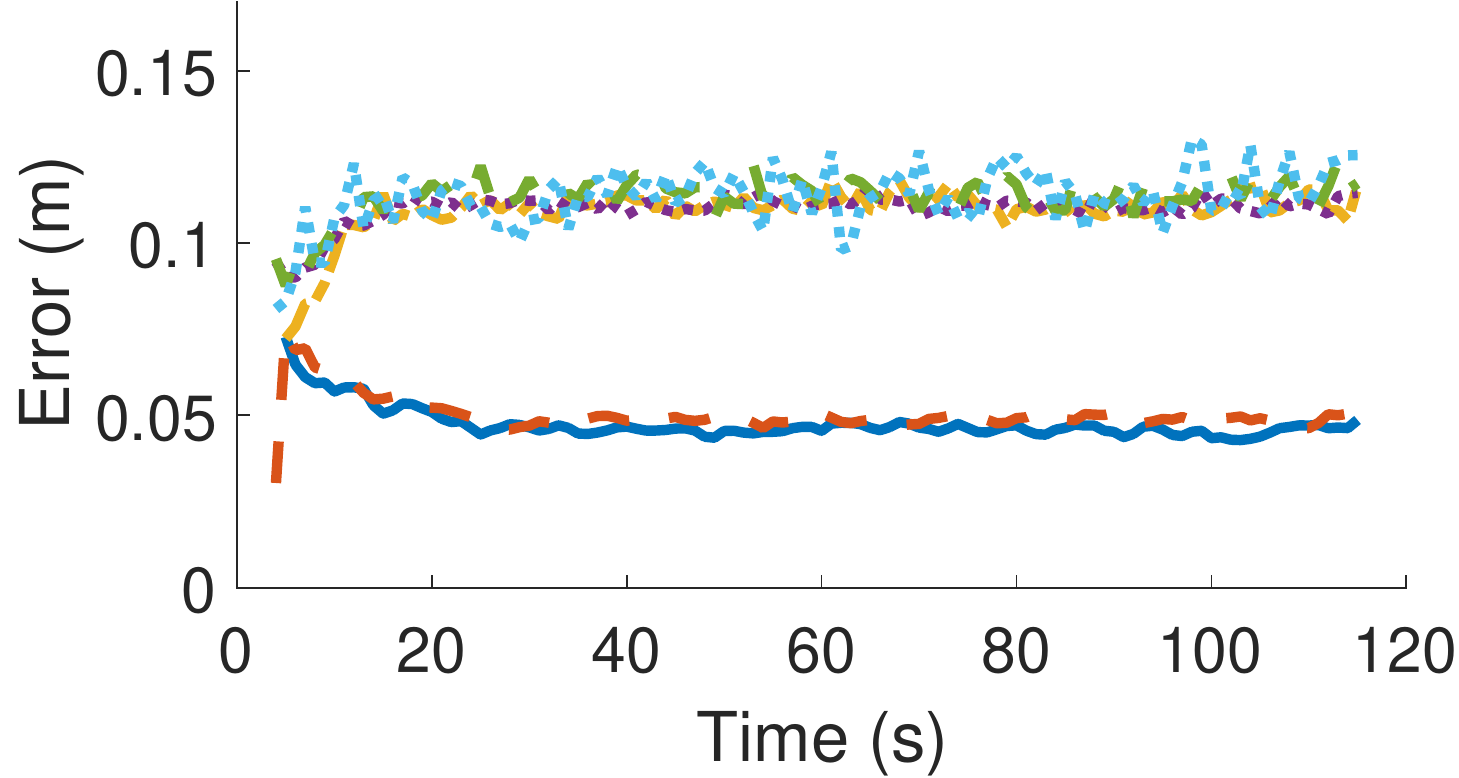}
        \caption*{$\textrm{iii)}$ $R_{HPS}=42$}
    \end{subfigure}%

	\begin{subfigure}[t]{0.31\textwidth}
        \centering
        \includegraphics[width=0.90\textwidth]{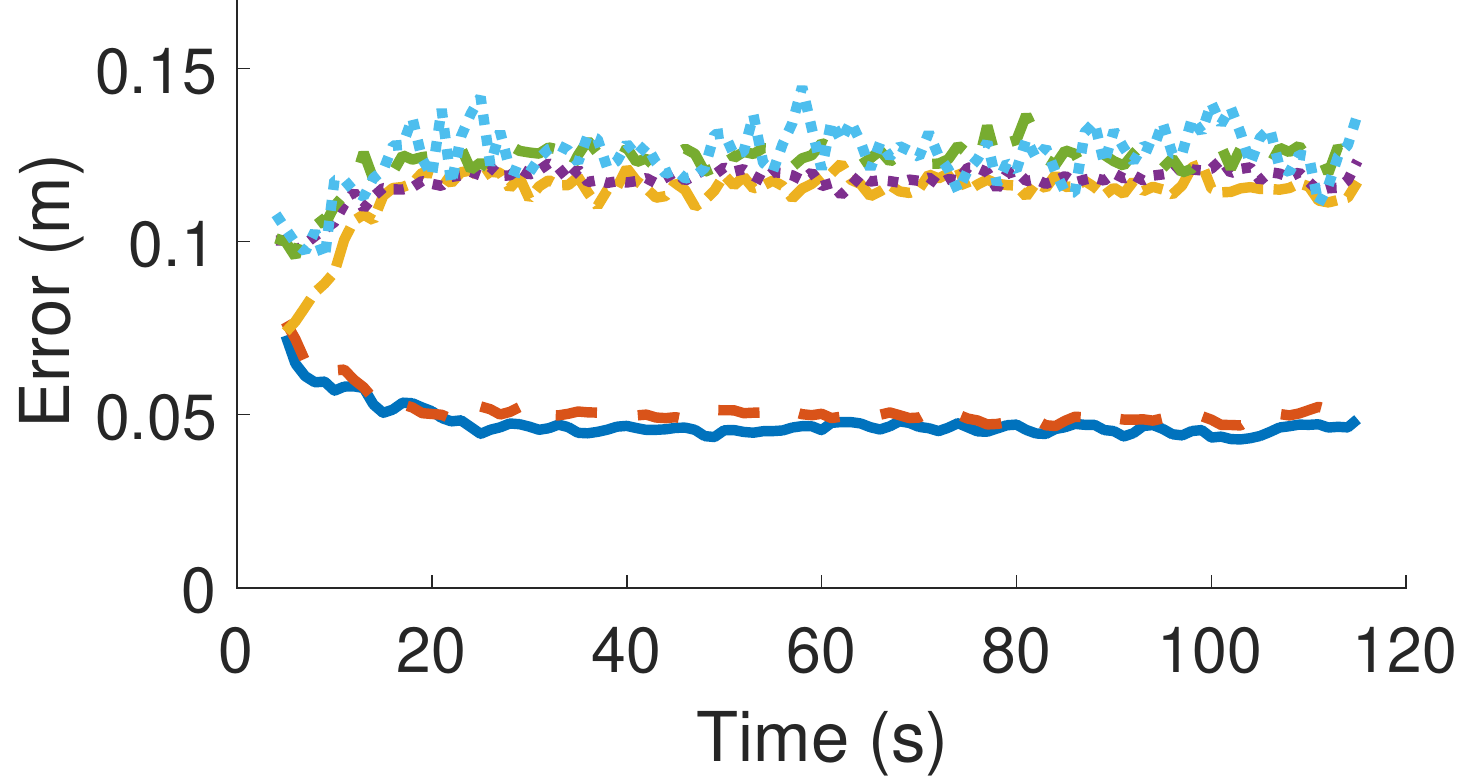}
        \caption*{$\textrm{iv)}$ $R_{HPS}=48$}
    \end{subfigure}%
	\begin{subfigure}[t]{0.31\textwidth}
        \centering
        \includegraphics[width=0.90\textwidth]{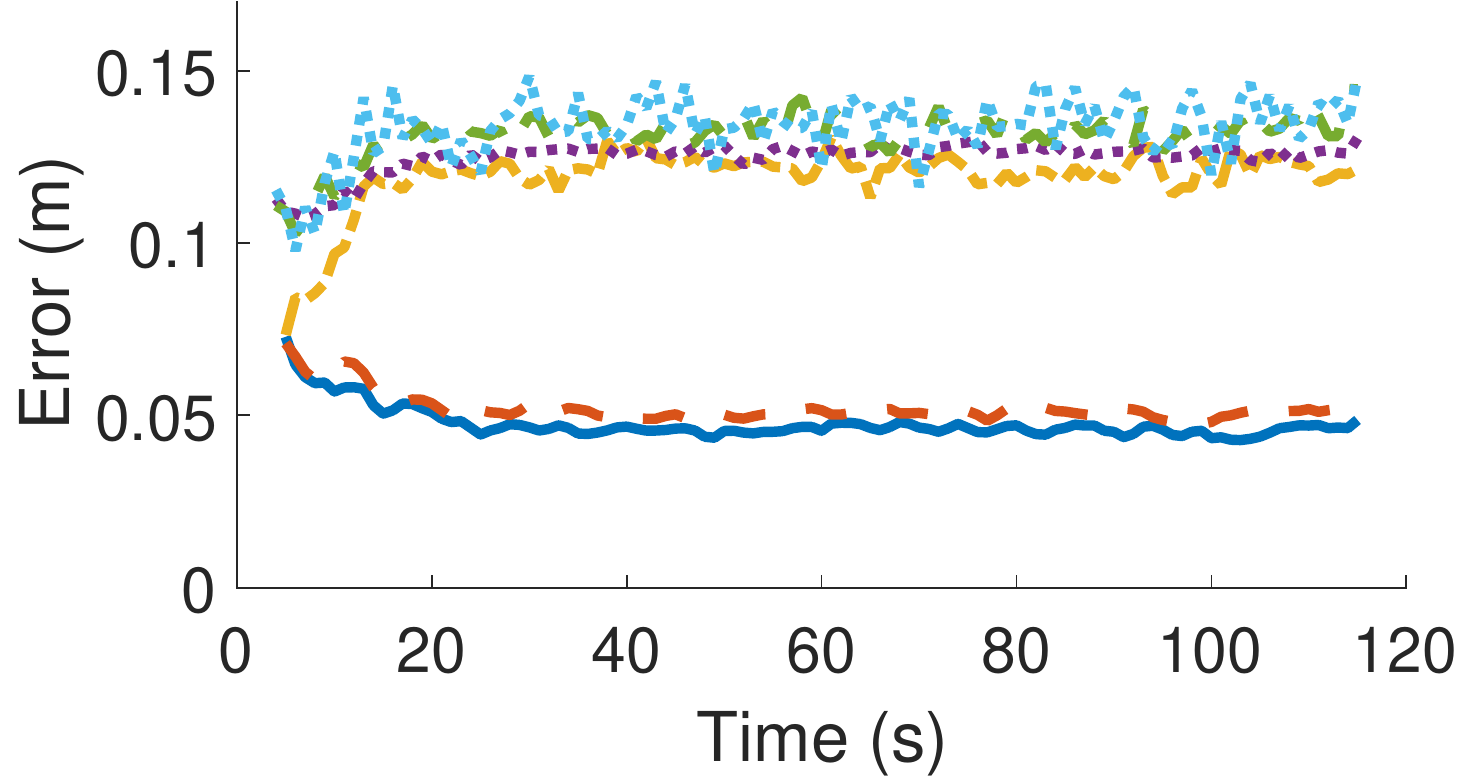}
        \caption*{$\textrm{v)}$ $R_{HPS}=52$}
    \end{subfigure}%
    \begin{subfigure}[t]{0.31\textwidth}
        \centering
        \includegraphics[width=0.90\textwidth]{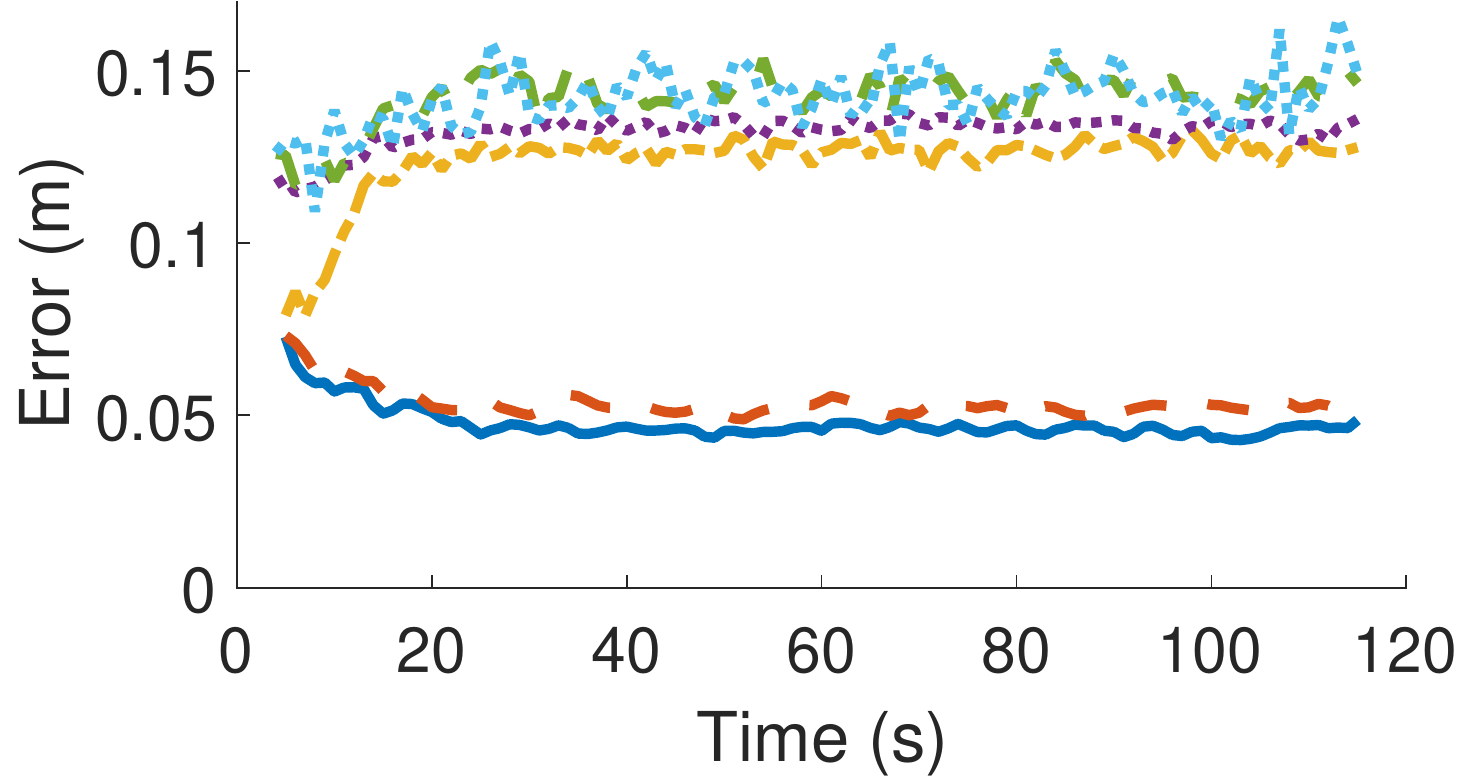}
        \caption*{$\textrm{vi)}$ $R_{HPS}=60$}
    \end{subfigure}%

    \caption*{a) Target Position Root Mean Squared Error: Comparison of POSE.R with other methods} \vspace{6pt}

    \begin{subfigure}[t]{0.31\textwidth}
        \centering
        \includegraphics[width=0.90\textwidth]{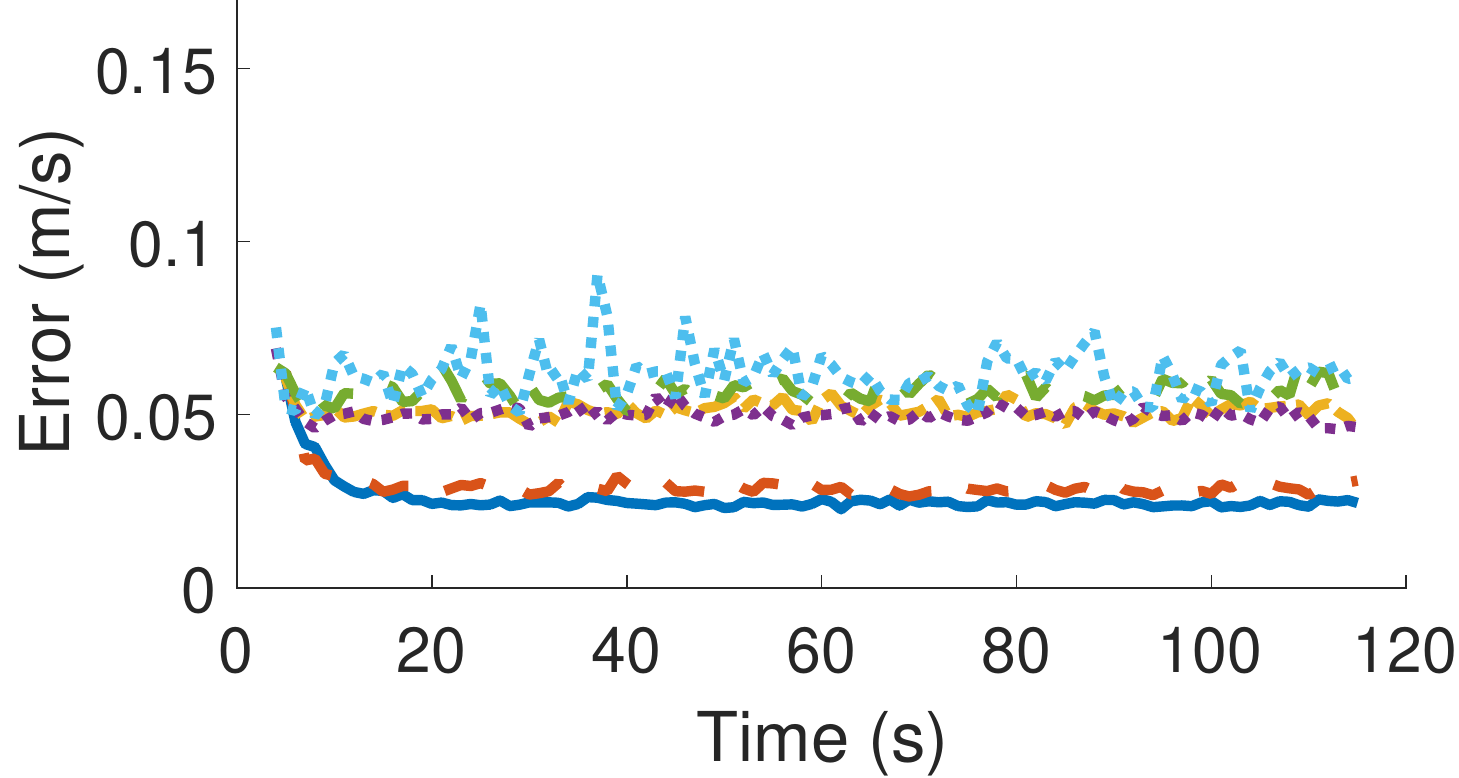}
        \caption*{$\textrm{i)}$ $R_{HPS}=30$}
    \end{subfigure}%
	\begin{subfigure}[t]{0.31\textwidth}
        \centering
        \includegraphics[width=0.90\textwidth]{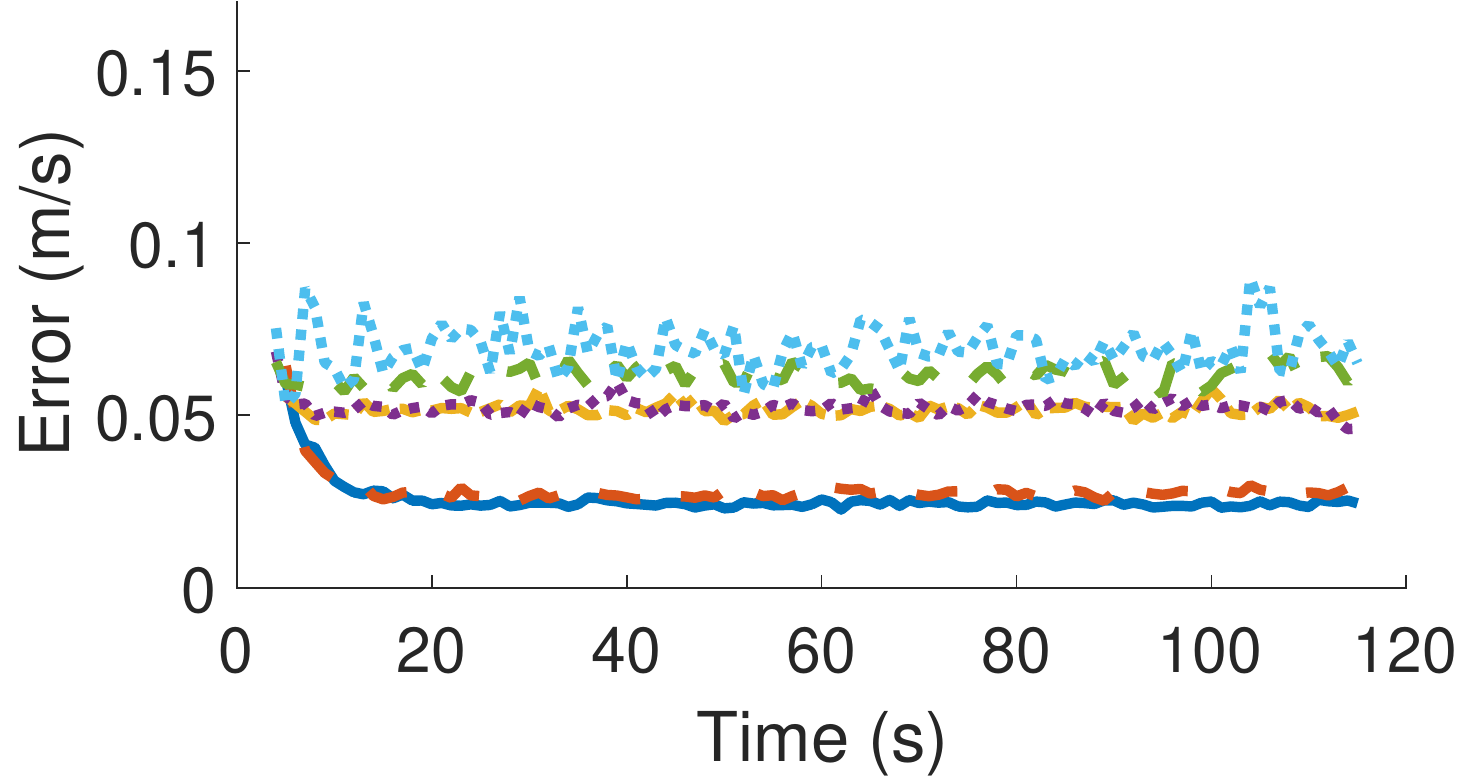}
        \caption*{$\textrm{ii)}$ $R_{HPS}=36$}
    \end{subfigure}%
    \begin{subfigure}[t]{0.31\textwidth}
        \centering
        \includegraphics[width=0.90\textwidth]{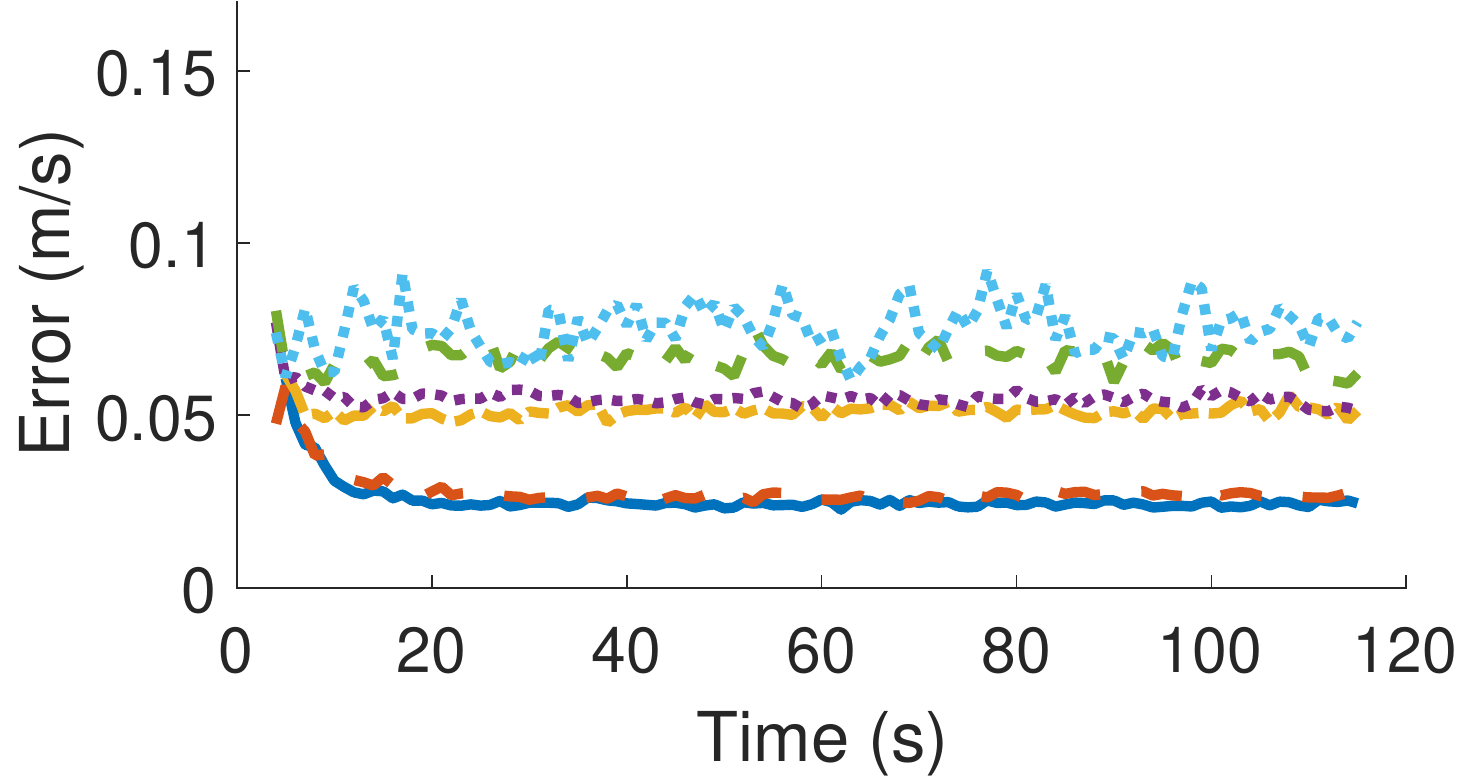}
        \caption*{$\textrm{iii)}$ $R_{HPS}=42$}
    \end{subfigure}%

	\begin{subfigure}[t]{0.31\textwidth}
        \centering
        \includegraphics[width=0.90\textwidth]{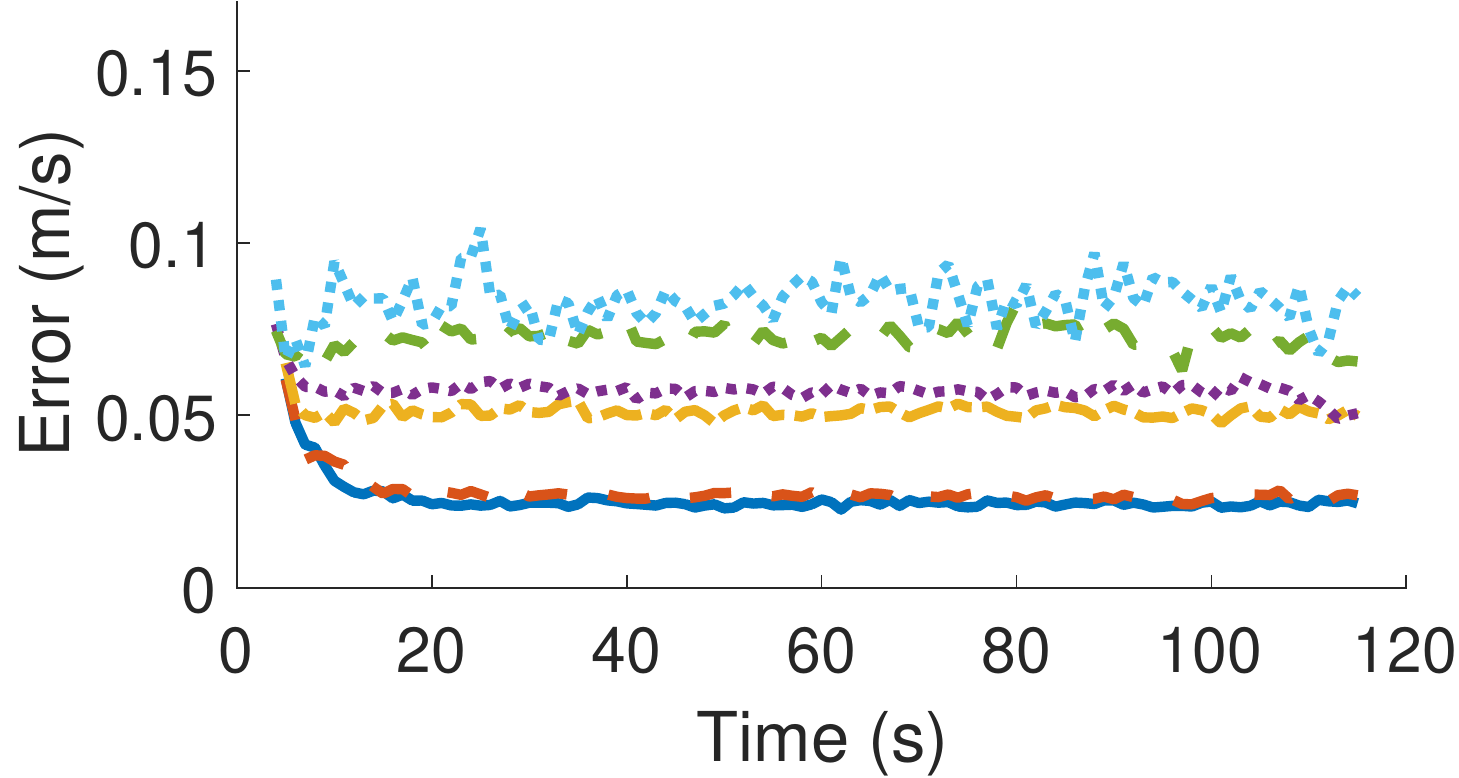}
        \caption*{$\textrm{iv)}$ $R_{HPS}=48$}
    \end{subfigure}%
	\begin{subfigure}[t]{0.31\textwidth}
        \centering
        \includegraphics[width=0.90\textwidth]{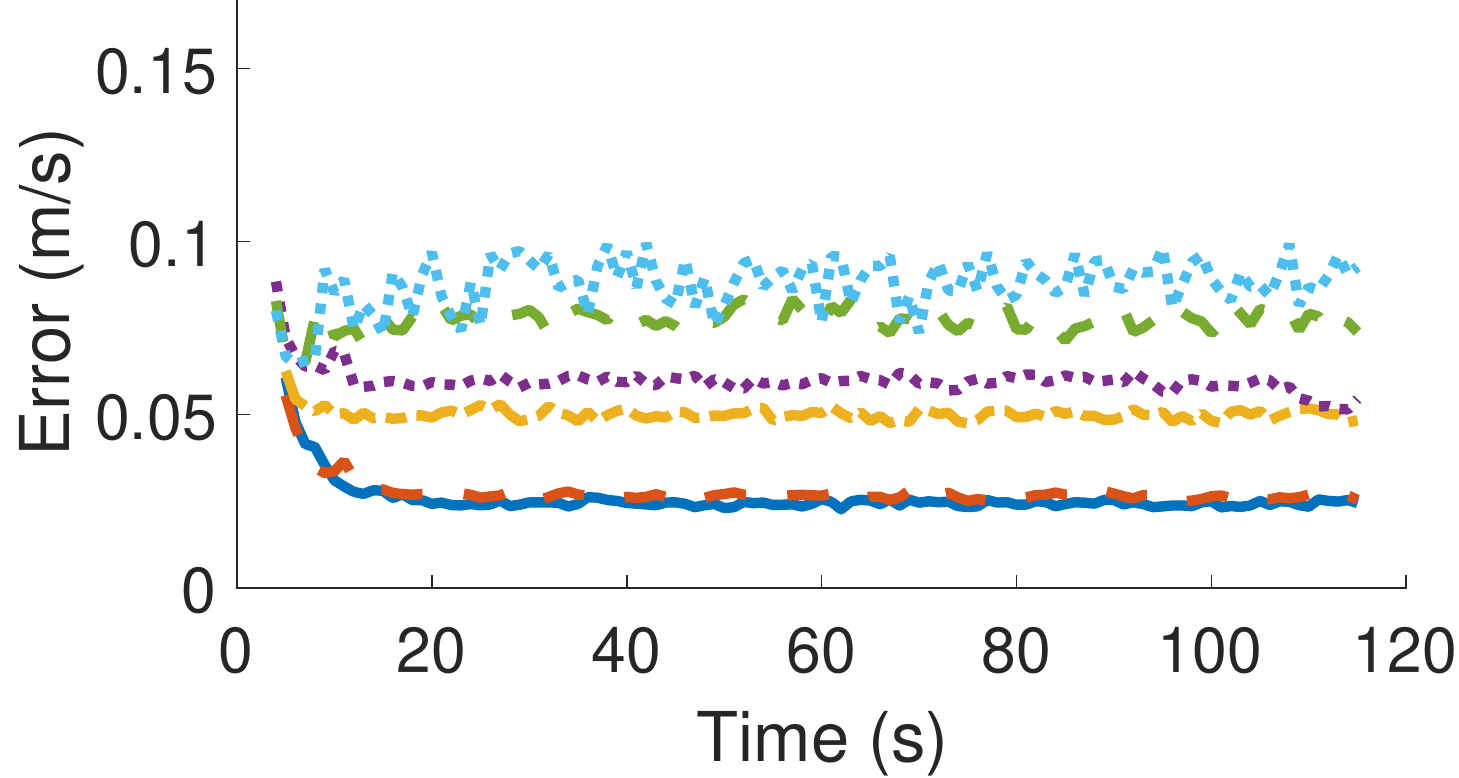}
        \caption*{v) $R_{HPS}=52$}
    \end{subfigure}%
    \begin{subfigure}[t]{0.31\textwidth}
        \centering
        \includegraphics[width=0.90\textwidth]{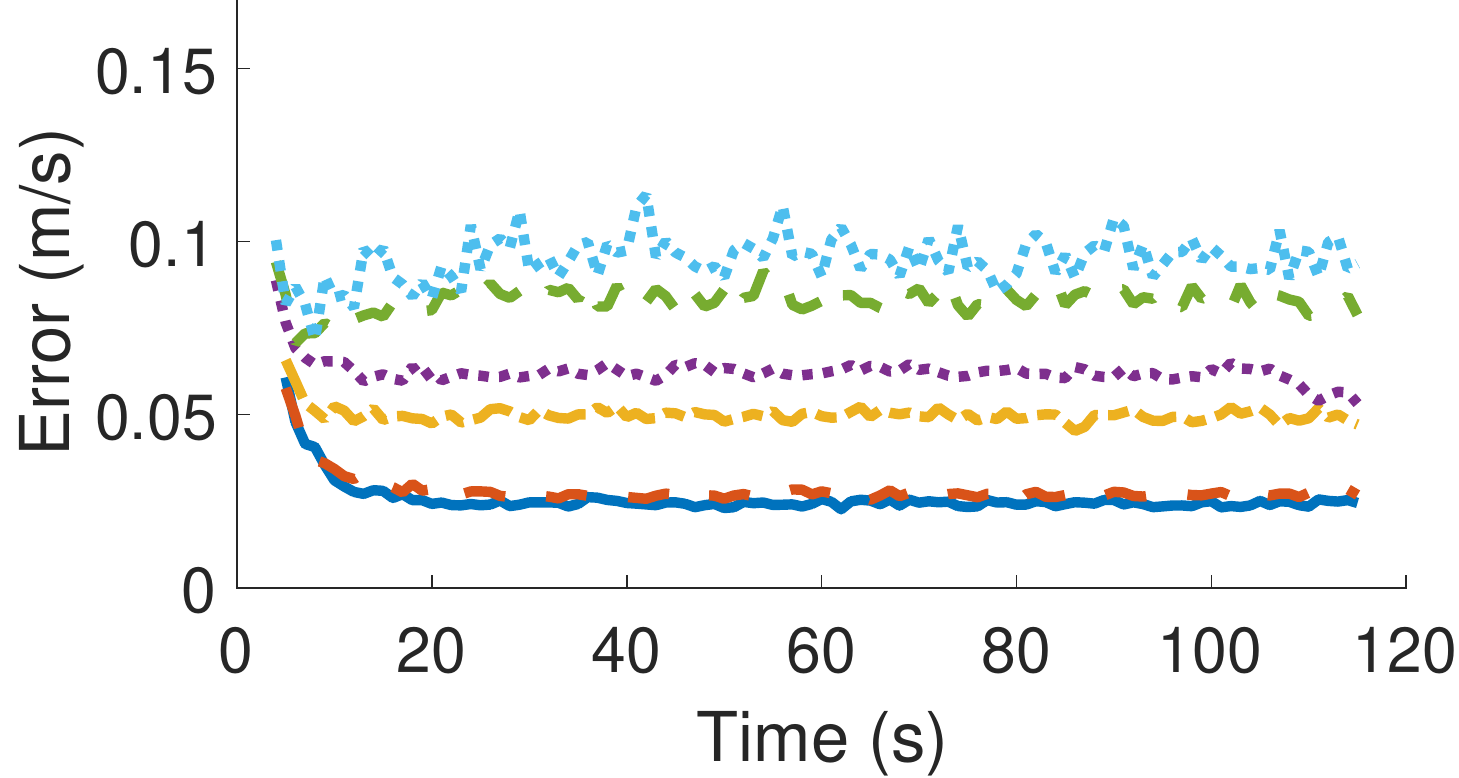}
        \caption*{$\textrm{vi)}$ $R_{HPS}=60$m}
    \end{subfigure}%
    \caption*{b) Target Velocity Root Mean Squared Error: Comparison of POSE.R with other methods} \vspace{6pt}

    \begin{subfigure}[t]{\textwidth}
    	\centering
        \includegraphics[width=0.75\textwidth]{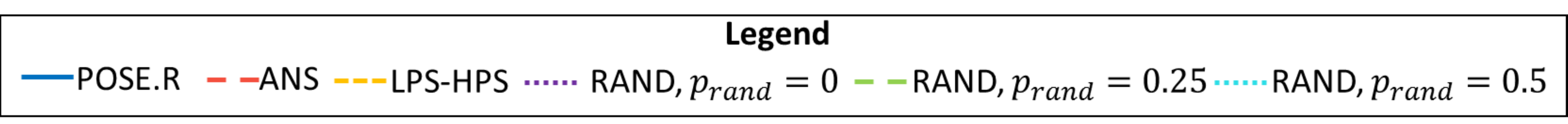}
    \end{subfigure}%
    \vspace{-3pt}
    \caption{Tracking Performance of POSE.R compared with other methods.} \label{fig:poser_pos_rmse} \vspace{-15pt}
\end{figure*}

\vspace{12pt}
\subsubsection{Tracking Performance Comparison} Fig.~\ref{fig:poser_pos_rmse} compares the tracking performance of the POSE.R algorithm with the other methods in terms of position and velocity root mean square error (RMSE), respectively. For this comparison, based on the missed detection characteristics, the parameters $\rho=1.4\times 10^{-3}$ and $p_{sleep}=0.5$ were chosen to ensure low missed detection rates. As seen, the POSE.R and ANS algorithms achieve significantly lower position and velocity RMSEs as compared to the other methods. This is because the ANS and POSE.R algorithms perform the same tracking and fusion strategies, which reduce the covariance error in the target estimates. The key difference between the ANS and POSE.R algorithms is that ANS selects the set of HPS nodes using GDOP that minimizes the predicted RMSE error, while the POSE.R algorithm incorporated energy remaining into the cost function. However, adding energy into the cost function does not degrade the tracking performance. Additionally, it can be seen that as the HPS sensing range increases, the RMSE of the ANS, LPS-HPS, and Random methods increases, while for POSE.R it stays the same. This is because the measurement noise of the HPS devices increases with distance.

\vspace{12pt}
\section{Conclusion} \label{sec:con}
This paper developed the POSE.R algorithm for distributed control of a heterogeneous sensor network for resilient and energy-efficient target tracking. The distributed network control approach consists of detecting and fusing target's state information to predict its trajectory, which is used to opportunistically track the target using a dynamic cluster of optimal sensor nodes. In the areas of high node density, the POSE.R algorithm provides energy-efficiency by tracking the target using optimal sensors in terms of remaining energy and geometric diversity around the target. In the areas of low node density or coverage gaps, the POSE.R algorithm provides resilience, that imparts the capability of self-healing to track the target by expanding the sensing ranges of surrounding sensors. The performance of the POSE.R algorithm was compared against existing methods using several metrics including missed detection rates, network lifetime and tracking performance. The simulation experiments yield that the POSE.R algorithm significantly improves the network lifetime, provides resilient tracking in presence of coverage gaps, and produces very low tracking errors and missed detection rates.

\bibliographystyle{IEEEtran}
\bibliography{POSE-R_BIB}

\appendices

%\section*{Appendix}

\section{EGDOP Comparison with GDOP and ME} \label{app:EGDOP}
\vspace{10pt}
In this section, the properties of the proposed EGDOP node selection method are compared with the GDOP and the Max Energy (ME) based node selection methods.

The GDOP node selection method \cite{K2006} selects $N_{sel}$ sensors from the set $\mathcal{S}_{can}$ that are geometrically diverse and minimize the predicted covariance error, to output the optimal node set $\mathcal{S}^*_{GDOP}$. On the other hand, the ME node selection method selects $N_{sel}$ sensors from $\mathcal{S}_{can}$ that have the maximum remaining energies, to form the node set $\mathcal{S}^*_{ME}$. To compare the above three methods, we simulated a high density network with a single target and conducted $500$ Monte Carlo runs. During each run, each node was assigned a random initial energy drawn from a uniform distribution to simulate the stochastic effects of energy variations amongnst nodes due to battery uncertainties or long deployment. For each run, each method was executed simultaneously for the same network and the resulting energies and covariance errors of the corresponding optimal sets were evaluated.

The results achieved are presented in Table~\ref{tb:EGDOP}, where the columns represent different bounds of the initial energy distribution. Several performance metrics were evaluated as presented below.

\begin{itemize}
\item \textit{Percentage of Energy Savings by EGDOP as Compared to GDOP}:
\begin{eqnarray}
E_{savings}=\frac{1}{N_{sel}E_0}\left(\sum_{s_i\in \mathcal{S}^*_{EGDOP}} \hat{E}_R^{s_i}(k+1)  -  \sum_{s_i\in\mathcal{S}^*_{GDOP}} \hat{E}_R^{s_i}(k+1)\right)\times 100\%,
\end{eqnarray}
where $\hat{E}_R^{s_i}(k+1) = E_0-E^{s_i}(k)-e_{HPS}R_1\Delta T$ is the predicted remaining energy at time $k+1$. As seen in Table~\ref{tb:EGDOP}, the energy savings of EGDOP vs. GDOP are always positive and the savings are higher if the variance of initial energy distribution is higher. This result shows that EGDOP selects healthy nodes with higher remaining energies as compared to GDOP.

\vspace{6pt}
\item \textit{Energy-efficiency of EGDOP and GDOP as Compared to ME}:
\begin{eqnarray}
E_{eff}(EGDOP) & = &\frac{\sum_{s_i\in \mathcal{S}^*_{EGDOP}} \hat{E}_R^{s_i}(k+1)}{\sum_{s_i\in \mathcal{S}^*_{ME}} \hat{E}_R^{s_i}(k+1)}  \\
E_{eff}(GDOP) & = & \frac{\sum_{s_i\in \mathcal{S}^*_{GDOP}} \hat{E}_R^{s_i}(k+1)}{\sum_{s_i\in \mathcal{S}^*_{ME}} \hat{E}_R^{s_i}(k+1)}.
\end{eqnarray}
These represent the efficiency of the EGDOP and GDOP methods as compared to ME and the results are presented in the second and third row of Table~\ref{tb:EGDOP}, respectively. Clearly, the efficiency of EGDOP is always higher than GDOP. Also, if the variance of initial energies is low, then both EDGOP and GDOP are very energy-efficient. However, for higher variances of initial energies, EGDOP achieves higher energy-efficiency than GDOP. Thus, EGDOP will result in even energy distribution in the network by always selecting the high energy nodes.

\vspace{6pt}
\item \textit{Kullback-Leibler (KL) Divergence between GDOP and EGDOP and GDOP and ME}:
\begin{eqnarray}
D_{KL}(\mathcal{N}_{GDOP}||\mathcal{N}_{\bigstar}) = \int_{-\infty}^{\infty} \int_{-\infty}^{\infty} \mathcal{N}_{GDOP}log\left(\frac{\mathcal{N}_{GDOP}}{\mathcal{N}_{\bigstar}}\right)dy dx,
\end{eqnarray}
where $\mathcal{N}_{GDOP} = \mathcal{N}(H\hat{\mathbf{x}}^{s_i}(k+1|k),\hat{\mathbf{\Sigma}}_{GDOP})$, $\mathcal{N}_{\bigstar} = \mathcal{N}(H\hat{\mathbf{x}}^{s_i}(k+1|k),\hat{\mathbf{\Sigma}}_{\bigstar})$, and
\begin{equation}
\mathbf{\Sigma}_\bigstar = \left(\sum_{s_h\in \mathcal{S}_{\bigstar}^*} \frac{1}{\sigma_{\theta}^2 r_{s_h}^2} \left[ \begin{array}{cc} s^2\phi_{s_h} & -s\phi_{s_h}c \phi_{s_h} \\ -s\phi_{s_h}c \phi_{s_h} & c^2\phi_{s_h} \end{array} \right]\right)^{-1}
\end{equation}
is the state covariance error achieved with the set of nodes selected using the method $\bigstar\in \{EGDOP, ME\}$. This measure compares the predicted covariance error of GDOP vs EDGOP and ME. As seen in the last two rows of Table~\ref{tb:EGDOP}, the KL divergence between GDOP and EGDOP is much smaller than between GDOP and ME. This means that the EGDOP method is not losing much divergence information by incorporating energy into the cost function. This in turn implies that the estimation error resulting from EGDOP is similar to that of GDOP, while the purely energy based method ME results in high estimation error.

\end{itemize}

Overall, these results indicate that the EGDOP method is more energy-efficient than the GDOP method and still selects nodes that are geometrically diverse which result in low estimation error.

\begin{table*}[t]
\centering
\caption{Performance Measures to Evaluate Energy Geometric Dilution Of Precision (EGDOP) Method}
\label{tb:EGDOP}
\begin{tabular}{ccccccc}
\hline
\multicolumn{1}{l}{}                  & \multicolumn{6}{c}{\textbf{Bounds of initial energy distribution among nodes}}                                           \\
                                       & $[0.5, 1]E_0$ & $[0.6, 1]E_0$ & $[0.7, 1]E_0$ & $[0.8, 1]E_0$ & $[0.9, 1]E_0$ & $[1, 1]E_0$ \\ \hline
\textbf{$E_{savings}$ (\%)} & 7.301             & 5.692            & 3.920            & 2.482            & 1.177            & 0.004         \\
\textbf{$E_{eff}(EGDOP)$}        & 0.962             & 0.968            & 0.977            & 0.985            & 0.992            & 1.00          \\
\textbf{$E_{eff}(GDOP)$}         & 0.933             & 0.946            & 0.962            & 0.976            & 0.988            & 1.00          \\
\textbf{$D_{KL}(GDOP||EGDOP)$}         & 0.0315            & 0.0317           & 0.0309           & 0.0307           & 0.0309           & 0.0424        \\
\textbf{$D_{KL}(GDOP||ME)$}    & 0.3190            & 0.3278           & 0.3205           & 0.3154           & 0.3138           & 0.4093        \\
\hline
\end{tabular} \vspace{-0pt}
\end{table*}

\section{Proof of Theorem \ref{thm:main}} \label{app:proof}
\vspace{10pt}
Since the target's predicted position has uncertainty, we define the expected coverage degree of the target given a joint action $a'$ as follows.

\begin{defn}[\textbf{Expected Coverage Degree}]
The expected coverage degree of target for a joint action $a'$ is defined as
\begin{equation}
\mathcal{E}(D(k+1)|a')=\sum_{j=1}^{N'_{sel}}j \sum_{g=1}^U \sum_{h=1}^V  \omega_{g,h}\mathcal{I}(J_{g,h}(a')=j)
\end{equation}
where $\mathcal{I}$ is an indicator function that equals to $1$ when $J_{g,h}(a')=j$ and $0$ otherwise; and $N'_{sel}$ is the number of players.
\end{defn}

Before proving the theorem, we need  the following lemma that allows us to determine when a node will select a new action over its previous action.

\begin{lem} \label{lem}
A node $s_i$ will switch its action from $a_i'$ to $a_i''$, where $a_i''>a_i'$, if
\begin{eqnarray}
\mathcal{E}\Big(D(k+1)|a''\Big)-\mathcal{E}\Big(D(k+1)|a'\Big) > \frac{\Delta R}{\Delta b_1 N'_{sel}R_L},
\end{eqnarray}
\end{lem}

\begin{proof}

A node $s_i$ will switch its action from $a_i'$ to $a_i''$, where $a_i''>a_i'$, if
\begin{eqnarray}\label{eq:utility_change}
\mathcal{U}_i(a_i'',a_{-i})-\mathcal{U}_i(a_i',a_{-i}) > 0.
\end{eqnarray}

Using Eq. (\ref{eq:utility}), the above condition becomes
\begin{eqnarray}\label{eq:utility_change2}
\sum_{g=1}^{U} \sum_{h=1}^{V} \omega_{g,h}\Big(B_{g,h}\big(J_{g,h}(a_i'',a_{-i})\big) - B_{g,h}\big(J_{g,h}(a_i',a_{-i})\big)\Big) - \frac{E_c(a_i'')-E_c(a_i')}{N'_{sel} E_c(R_L)} > 0.
\end{eqnarray}

Now, examining these terms individually, the target coverage achieved for action $a_i'$ is
\begin{eqnarray}
\sum_{g=1}^{U} \sum_{h=1}^{V} \omega_{g,h} B_{g,h}\big(J_{g,h}(a_i',a_{-i})\big) &=& \sum_{j=1}^{N'_{sel}} B_{g,h}(j) \sum_{g=1}^U \sum_{h=1}^V \omega_{g,h}\mathcal{I}\big(J_{g,h}(a_i',a_{-i})=j\big).
\end{eqnarray}
Define $\chi'(j)=Pr(D(k+1)=j|a')=\sum_{g=1}^U \sum_{h=1}^V \omega_{g,h}\mathcal{I}(J_{g,h}(a')=j)$. A node is motivated to increase its action only when $J_{g,h}\leq N_{sel}, \forall g,h$. Thus, using $B_{g,h}(j) = \Delta b_1 j$, we get
\begin{eqnarray}
\sum_{g=1}^{U} \sum_{h=1}^{V} \omega_{g,h} B_{g,h}\big(J_{g,h}(a_i',a_{-i})\big) &=& \Delta b_1 \sum_{j=1}^{N_{sel}} j\chi'(j).
\end{eqnarray}

Similarly, the target coverage achieved for action $a_i''$ is
\begin{eqnarray}
\sum_{g=1}^{U} \sum_{h=1}^{V} \omega_{j,h} B_{g,h}\big(J_{g,h}(a_i'',a_{-i})\big) &=& \sum_{j=1}^{N'_{sel}} B_{g,h}(j) \sum_{g=1}^U \sum_{h=1}^V \omega_{g,h}\mathcal{I}\big(J_{g,h}(a_i'',a_{-i})=j\big).
\end{eqnarray}
Define $\chi''(j)=Pr(D(k+1)=j|a'')=\sum_{g=1}^U \sum_{h=1}^V \omega_{g,h}\mathcal{I}(J_{g,h}(a'')=j)$. Using $B_{g,h}(j) = \Delta b_1j$ we get
\begin{eqnarray}
\sum_{g=1}^{U} \sum_{h=1}^{V} \omega_{j,h} B_{g,h}\big(J_{g,h}(a_i'',a_{-i})\big) &=& \Delta b_1 \sum_{j=1}^{N_{sel}} j\chi''(j).
\end{eqnarray}

Then, plugging the above two equations into (\ref{eq:utility_change2}), we get

\begin{eqnarray}
\Delta b_1 \sum_{j=1}^{N_{sel}} j(\chi''(j)-\chi'(j)) > \frac{E_c(a_i'')-E_c(a_i')}{N'_{sel} E_c(R_L)}\ge \frac{\Delta R}{N'_{sel}R_L}.
\end{eqnarray}

This implies that
\begin{eqnarray}
\mathcal{E}\Big(D(k+1)|a''\Big)-\mathcal{E}\Big(D(k+1)|a'\Big) > \frac{\Delta R}{\Delta b_1 N'_{sel}R_L}.
\end{eqnarray}

\end{proof}

Lemma \ref{lem} provides a criteria that allows the agent to select action $a_i''$ over $a_i'$.

\vspace{6pt}
\textit{Proof of Theorem \ref{thm:main}}:

Suppose that the players' actions are initialized as $a_i=0$, $\forall s_i\in\mathcal{S}'(k+1)$. The objective is to achieve the Nash equilibrium $a^*=(a_i^*,a_{-i}^*)$, where every cell is covered by $N_{sel}$ nodes. To achieve this solution, the utility function is designed to provide incentives for players to select an action that increases the coverage degree in every cell up to $N_{sel}$, while discouraging them to take an action that increases the coverage degree above $N_{sel}$. Since the Maxlogit learning algorithm sequentially selects a single player to attempt to change its action during each interaction, the number of players covering each cell will also sequentially increase.

First, consider the situation when $J_{g,h}<N_{sel}$, $\forall g,h$. Here, we will find a bound on $\Delta b_1$, that encourages a player to take an action that increases the coverage degree in every cell up to $N_{sel}$. In order to incentivize a player to select an action, the slope $\Delta b_1$ is designed to favor switching from the current action $a_i'$ to $a_i^*$,  $a_i^*>a_i'$,  such that overall we achieve $\mathcal{E}(D(k+1)|a^*)\le N_{sel}$. Using Lemma \ref{lem} we get
\begin{eqnarray}
\mathcal{E}(D(k+1)|a')< \mathcal{E}(D(k+1)|a^*)-\frac{\Delta R}{\Delta b_1 N'_{sel}R_L}\leq N_{sel}-\frac{\Delta R}{\Delta b_1 N'_{sel}R_L}. \nonumber
\end{eqnarray}
Now, suppose that the above condition is violated and $\mathcal{E}(D(k+1)|a')> N_{sel}-\frac{\Delta R}{\Delta b_1 N'_{sel}R_L}$,  then the agent would select the action $a'$ and this would become the true Nash equilibrium. Thus, the expected coverage degree at the Nash equilibrium $a^*$ is bounded as follows.
\begin{eqnarray}\label{eq:exp_bound}
N_{sel}-\frac{\Delta R}{\Delta b_1 N'_{sel}R_L} <  \mathcal{E}(D(k+1)|a^*) \leq N_{sel}.
\end{eqnarray}
Additionally,
\begin{eqnarray}
\mathcal{E}(D(k+1)|a^*) &=& N_{sel}\chi^*(N_{sel})+(N_{sel}-1)\chi^*(N_{sel}-1) + \ldots 1\chi^*(1) \nonumber \\
&\leq & N_{sel}\chi^*(N_{sel})+(N_{sel}-1)\big(1-\chi^*(N_{sel})\big) \nonumber \\
&=&(N_{sel}-1)+\chi^*(N_{sel}).
\end{eqnarray}
where $\chi^*(N_{sel})=Pr(D(k+1)=N_{sel}|a^*)$. Then,
\begin{eqnarray}
(N_{sel}-1)+ \chi^*(N_{sel})&>&N_{sel}-\frac{\Delta R}{\Delta b_1 N'_{sel}R_L} \nonumber  \\
\chi^*(N_{sel})&>& 1-\frac{\Delta R}{\Delta b_1 N'_{sel}R_L}.
\end{eqnarray}
Then to achieve $Pr\big(D(k+1)=N_{sel}|a^*\big) > 1-\delta$, one must have $\frac{\Delta R}{\Delta b_1 N'_{sel}R_L}<\delta$. Thus, we obtain
\begin{eqnarray}
\Delta b_1>\frac{\Delta R}{\delta N'_{sel}R_L}.
\end{eqnarray}

Thus far, we have considered the condition when the Nash equilibrium  $(a_i^*,a_{-i}^*)$ has $J_{g,h}\le N_{sel}, \ \forall g,h$. Now, consider a situation when a player $i$ takes an action $a''_i$, $a_i''>a_i^*$, then it may cause some of the cells to be covered by $N_{sel}+1$ nodes. In this situation, it is desired to discourage the player to choose action $a_i''$, such that
\begin{eqnarray} \label{eq:utility_do_not_change}
\mathcal{U}_i(a_i'', a_{-i})-\mathcal{U}_i(a^*_i,a_{-i})< 0.
\end{eqnarray}
Following the same process of Lemma~\ref{lem}, we get
\begin{eqnarray}\label{eq:diff2}
\sum_{g=1}^{U} \sum_{h=1}^{V} \omega_{g,h}\Big(B_{g,h}\big(J_{g,h}(a_i'',a_{-i})\big) - B_{g,h}\big(J_{g,h}(a_i^*,a_{-i})\big)\Big) - \frac{E_c(a_i'')-E_c(a_i^*)}{N'_{sel} E_c(R_L)} < 0.  \\
  \sum_{j=N_{sel}+1}^{N'_{sel}}\big(\Delta b_1N_{sel}- \Delta b_2(j-N_{sel})\big)\chi''(j)+\Delta b_1 \sum_{j=1}^{N_{sel}} j\chi''(j) - \Delta b_1 \sum_{j=1}^{N_{sel}} j\chi^*(j) < \frac{\Delta R}{N'_{sel}R_L},  \\
\big(\Delta b_1N_{sel}- \Delta b_2\big)\chi''(N_{sel}+1) < \frac{\Delta R}{N'_{sel}R_L} - \Delta b_1 \sum_{j=1}^{N_{sel}} j\chi''(j) + \Delta b_1 \sum_{j=1}^{N_{sel}} j\chi^*(j)
\end{eqnarray}
Now, $\sum_{j=1}^{N_{sel}} j\chi''(j)\le N_{sel}(1-\chi''(N_{sel}+1))$ and $\sum_{j=1}^{N_{sel}} j\chi^*(j)>N_{sel}-\frac{\Delta R}{\Delta b_1N'_{sel}R_L}$
from (\ref{eq:exp_bound}). Thus,
\begin{eqnarray}
\big(\Delta b_1N_{sel}- \Delta b_2\big)\chi''(N_{sel}+1) &< & \frac{\Delta R}{N'_{sel}R_L} - \Delta b_1 N_{sel}(1-\chi''(N_{sel}+1)) - \Delta b_1 \frac{\Delta R}{\Delta b_1N'_{sel}R_L} + \Delta b_1 N_{sel} \nonumber \\
& < &\frac{\Delta R}{N'_{sel}R_L} + \Delta b_1 N_{sel}\chi''(N_{sel}+1) - \frac{\Delta R}{N'_{sel}R_L} \nonumber \\
&=& \Delta b_1 N_{sel}\chi''(N_{sel}+1) \\
\Delta b_2 & > & 0
\end{eqnarray}

\section{Potential Game Validation} \label{app:Game}
\vspace{10pt}

To validate the performance of potential games, we ran $500$ Monte Carlo simulations of a target traveling in a sensor network of density $\rho = 1.4\times 10^{-3}$. A coverage gap of size $R_1$ was inserted in these networks such that  at least one potential game is triggered in each Monte Carlo run. The following measures were evaluated:

\begin{itemize}

\item \textit{Average Probability of Coverage Degree to be $N_{sel}$}: First, we computed the average probability per game of getting coverage degree to be $N_{sel}$ as follows:
\begin{eqnarray}
\mathbf{\chi}^\star(N_{sel}) = \frac{1}{N_{g}} \sum_{i=1}^{N_{g}} \sum_{g=1}^U \sum_{h=1}^V \omega_{g,h}\mathcal{I}(J_{g,h}(a^\star(i))=N_{sel}),
\end{eqnarray}
where $N_g\ge500$ is the total number of games played in all of the Monte Carlo simulations. As seen in the first row of Table~\ref{tb:game}, the average probability per game is very close to $1$. This validates that the equilibrium action set $a^\star$ selects $N_{sel}$ nodes to cover almost all cells of the partition region. We can see that as the number of players $N'_{sel}$ increases, the average probability increases, indicating that more players allows the game to identify an action set that covers the entire partition region.

\begin{table}[t]
\centering
\caption{Game Performance Results}
\label{tb:game}\vspace{-6pt}
\begin{tabular}{c|ccccc}\hline
                                    & \multicolumn{5}{c}{$N'_{sel}$}              \\
                                    & $3$    & $4$    & $5$    & $6$    & $7$     \\ \hline
$\mathbf{\chi}^\star(N_{sel})$ & 0.9981 & 0.9995 & 0.9996 & 0.9997 & 0.9998  \\
$\Phi_{eff}$    & 0.999  & 0.980  & 0.976  & 0.974  & 0.973   \\
$t_{game} (s)$                          & 0.050  & 0.065  & 0.079  & 0.098  & 0.129   \\
$t_{opt} (s)$                           & 0.088  & 0.808  & 8.137  & 56.300 & 517.957 \\ \hline
\end{tabular} \vspace{-0pt}
\end{table}

\vspace{6pt}
\item \textit{Game Efficiency}: We compared the equilibrium solution obtained by the potential game against the optimal solution obtained using an exhaustive search. We define the game efficiency as
\begin{eqnarray}
\Phi_{eff}=\frac{\Phi(a^\star)}{\Phi(a^\star_{opt})},
\end{eqnarray}
where the optimal action set $a^\star_{opt}$ was computed as follows
\begin{eqnarray}
a^\star_{opt} = \argmax{a\in\mathcal{A}_{\mathcal{S}'(k+1)}} \Phi(a).
\end{eqnarray}
The second row of Table~\ref{tb:game} shows that the efficiency of games with respect to the optimal solution is close to $1$. This implies that the potential game and Maxlogit learning allows for the agents to select an action set that is close to the optimal solution.

\vspace{6pt}
\item \textit{Computation Time}: Finally, we compared the amount of time it takes for the players to compute the action sets $a^\star$ and $a^\star_{opt}$. As seen in the last two rows of Table~\ref{tb:game}, the time $t_{game}$ taken by the potential games and Maxlogit learning to obtain the equilibrium solution is significantly less than the time $t_{opt}$ taken by exhaustive search to obtain the optimal solution. Once the number of players $N'_{sel}>4$, the computation time of the exhaustive search becomes impractical for real-time implementation. This validates the feasibility of potential games for optimal sensing range adjustment in real-time target tracking applications.
\end{itemize}

\section{Comparison of POSE.R with POSE and POSE.3C Networks} \label{app:net_comp}
\vspace{10pt}

This section compares the performance of the POSE.R algorithm with the POSE~\cite{HGW2017} and POSE.3C~\cite{HGW2019} algorithms. POSE.R is an advanced algorithm designed for optimal node selection and adaptive sensor range selection to provide resilient and energy-efficient target tracking even for low density networks and in the presence of coverage gaps. On the other hand, POSE and POSE.3C algorithms were designed primarily for energy-efficiency and considered only fixed range HPS sensors. Furthermore, the POSE algorithm was a primitive version of the POSE.3C algorithm that did not include efficient node selection, thus leading to redundant nodes activated around the target resulting in energy wastage. In contrast, POSE.R performs node selection via joint optimization of energy and geometric diversity thus allowing reliable nodes to track the target. The adaptive range selection provides resilience towards irregular node distribution and coverage gaps, and yields high tracking accuracy and low missed detection rates.

For performance comparison we simulated the POSE and POSE.3C networks for each fixed sensing range $R_{HPS}^{s_i}$, while the POSE.R network can perform adaptive sensor range selection as needed. Furthermore, the density for each of the networks was varied and $500$ Monte Carlo simulation runs were conducted for each of these scenarios.

\begin{figure*}[t!]
	\begin{subfigure}{0.33\textwidth}
        \centering
        \includegraphics[width=\textwidth]{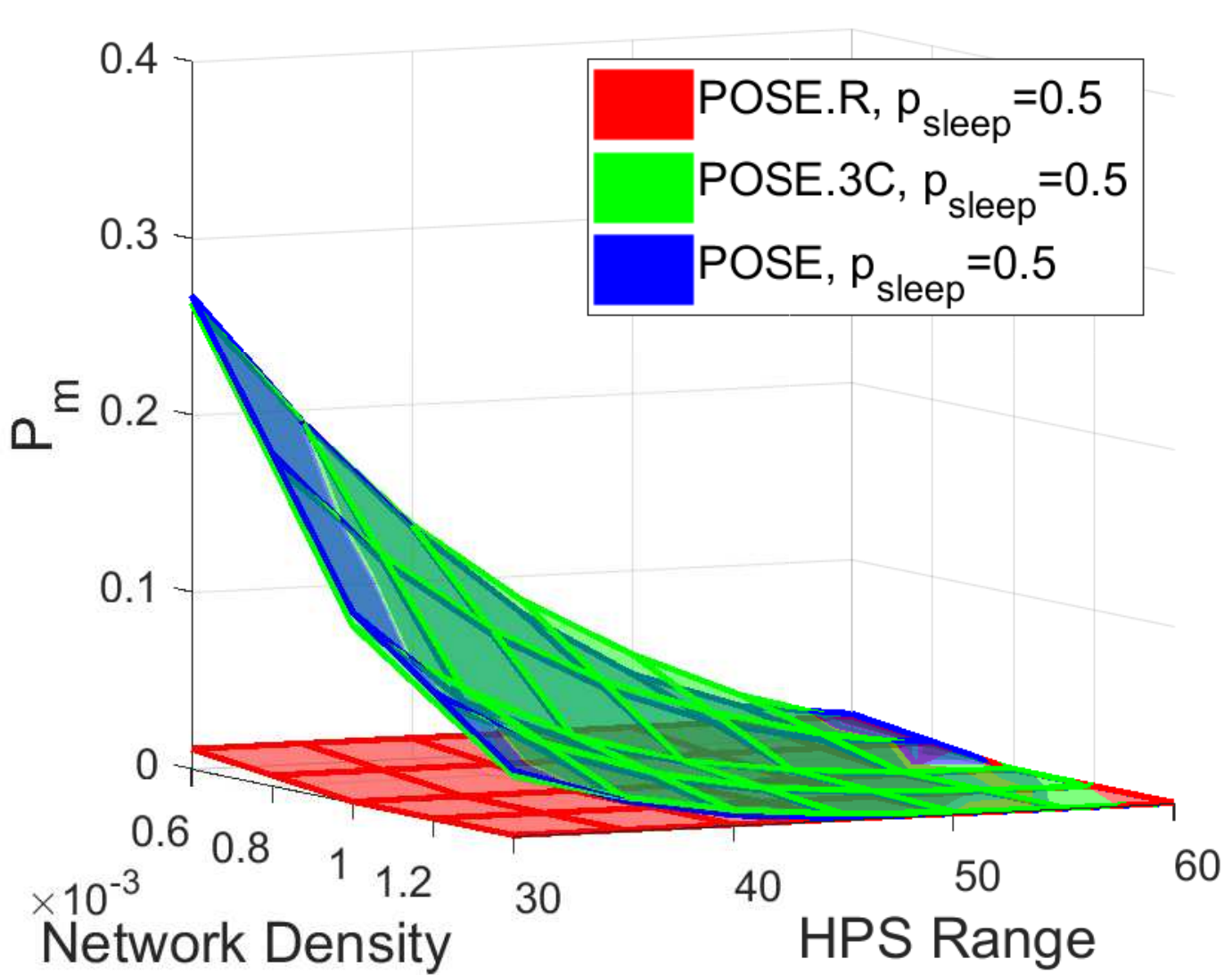}
        \caption{Probability of missed detection of the target.}\label{fig:pm_pose_comp}
    \end{subfigure}%
	\begin{subfigure}{0.33\textwidth}
        \centering
        \includegraphics[width=\textwidth]{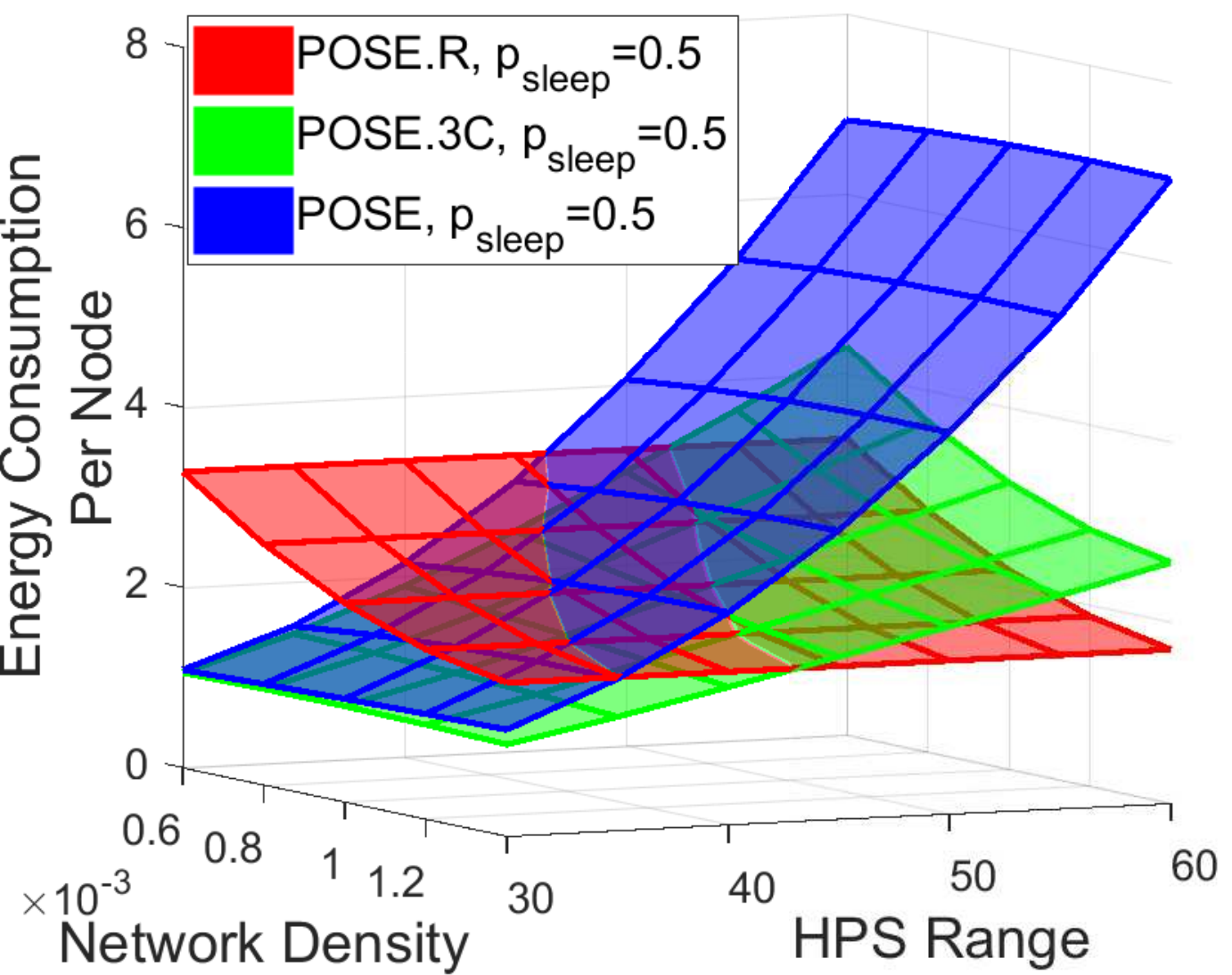}
        \caption{Average energy consumption around the target.}\label{fig:e_pose_comp}
    \end{subfigure}%
    \begin{subfigure}{0.33\textwidth}
        \centering
        \includegraphics[width=\textwidth]{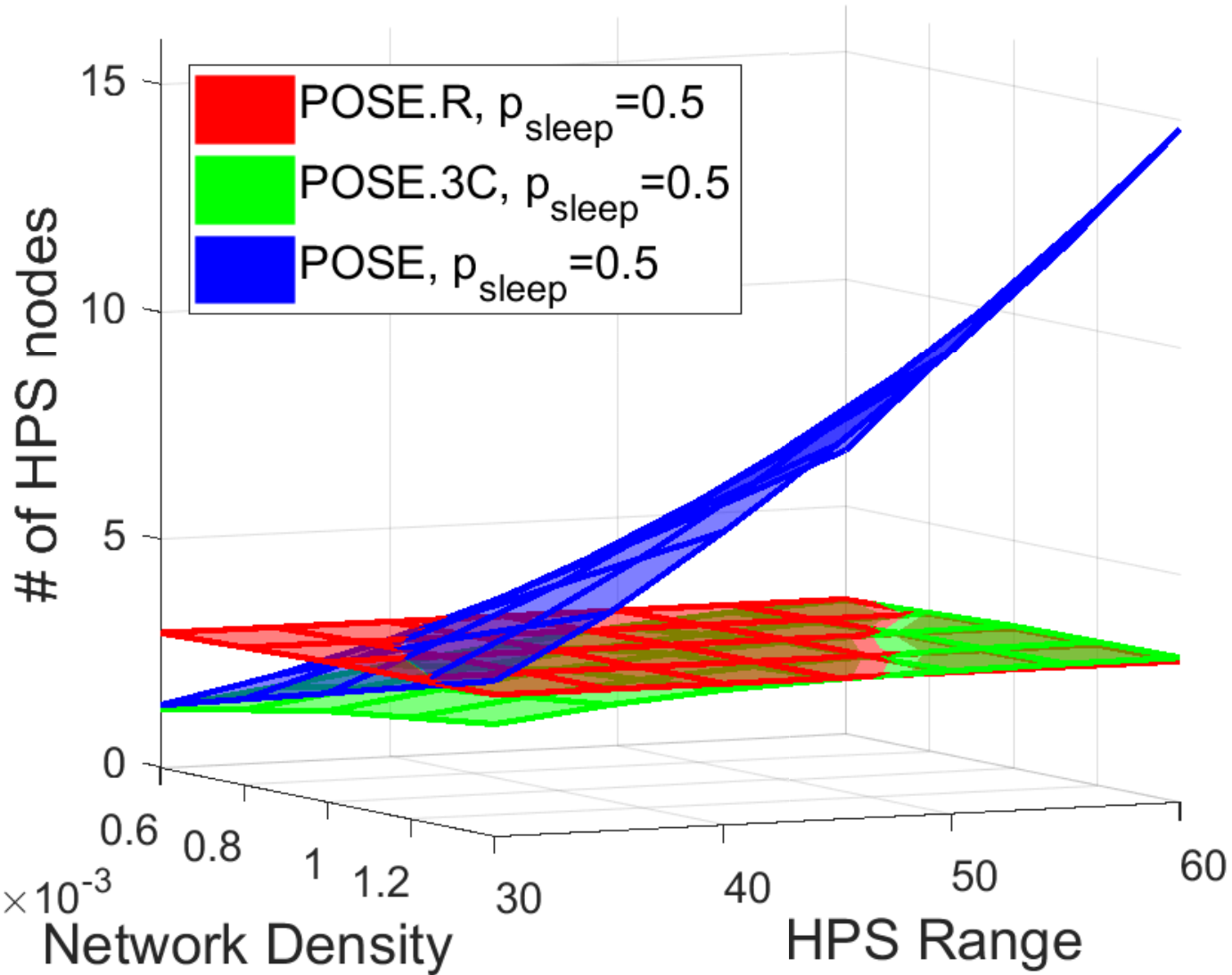}
        \caption{Average number of  HPS  nodes at each time step.}\label{fig:hps_pose_comp}
    \end{subfigure}%
	\vspace{-6pt}
	\caption{Performance comparison of POSE.R, POSE.3C and POSE algorithms.} \label{fig:pose_comp} \vspace{-12pt}
\end{figure*}

Fig.~\ref{fig:pose_comp} compares the probability of missed detection $P_m$, the average energy consumption around the target, and the average number of  HPS  activated nodes for the three algorithms. Fig.~\ref{fig:pm_pose_comp} shows that the POSE and POSE.3C networks result in significantly high missed detection rates as compared to the POSE.R network, especially for low density networks and low  HPS  ranges. This is due to the adaptability of the POSE.R network to allow the nodes to extend their sensing ranges when a target is predicted to travel within a low density region or a coverage gap. Thus, POSE.R provides opportunistic resilience to the network, i.e., a self-healing capability to track the target when it passes through low density regions or coverage gaps.

Fig.~\ref{fig:e_pose_comp} shows the average energy consumption around the target for the three networks. We can see that POSE and POSE.3C consume less energy than POSE.R for lower sensing ranges because the POSE.R network is expanding the sensing ranges of selected sensors to maintain target tracking. However, as the  HPS  range increases, the POSE.R provides energy savings as compared to POSE and POSE.3C because the selected nodes can decrease their sensing ranges to ensure target coverage. Thus, the adaptability of the nodes sensing range can improve the energy efficiency of the network.

Fig.~\ref{fig:hps_pose_comp} presents the average number of  HPS  nodes active at each time step. As we can see, the POSE.R network is able to maintain $N_{sel}=3$  HPS  nodes enabled during each time step even for low network densities, thus providing low missed detection rates. On the other hand, POSE and POSE.3C are unable to maintain $N_{sel}=3$  HPS  nodes for tracking the target for low network densities and low  HPS  ranges. Since POSE does not have node selection, it activates a large number of nodes as the  HPS  range and network density increase.

In summary, the above results show that POSE.R provides both resilience and energy-efficiency to the network and results in high tracking accuracy and low missed detection rates for target tracking, hence improving the overall network performance and providing longevity.

\printnomenclature

\end{document}